\documentclass[12pt]{iopart}

\usepackage{hyperref}
\newcommand{\titletext}{Equilibration, thermalisation, and the emergence of statistical mechanics in closed quantum systems}

\makeatletter 
\hypersetup{pdftitle={\titletext},
	  pdfauthor={Christian Gogolin, Jens Eisert},
	  pdfkeywords={Review, equilibration, thermalisation, statistical mechanics, closed systems, many-body localisation, typicality}
	 }
\makeatother

\usepackage{amsopn}

\usepackage[square,sort&compress,comma,numbers]{natbib}

\newcommand*{\doi}[1]{\href{http://dx.doi.org/\detokenize{#1}}{doi: \detokenize{#1}}}

\expandafter\let\csname equation*\endcsname\relax 
\expandafter\let\csname endequation*\endcsname\relax 
\usepackage{amsmath,amsthm,amssymb,mathrsfs,wasysym,dsfont,bbm}
\usepackage[automark]{scrpage2}
\usepackage{filecontents}
\usepackage[T1]{fontenc}
\usepackage{csquotes}
\usepackage[german,ngerman,english]{babel}
\usepackage{centernot}
\usepackage{mathtools}
\usepackage{ifthen}
\usepackage{tikz}
\usetikzlibrary{calc}
\usetikzlibrary{decorations.pathreplacing}
\usetikzlibrary{decorations}
\usetikzlibrary{positioning}
\usepackage{pgfplots}
\usepackage{xxcolor}
\definecolor{structure}{rgb}{0.23,0.4,0.7}
\usepackage[normalem]{ulem}

\newtheorem{theorem}{Theorem}
\newtheorem{corollary}{Corollary}
\newtheorem{definition}{Definition}
\newtheorem{conjecture}{Conjecture}
\newtheorem{observation}{Observation}
\newsavebox{\blocksavebox}
\newenvironment{block}[1]{%
\begin{lrbox}{\blocksavebox}%
\begin{minipage}%
{\linewidth}%
}{%
\end{minipage}\end{lrbox}%
\usebox{\blocksavebox}
}

\definecolor{niceblue}{rgb}{0.33,0.5,0.8}

\newcommand{\refsub}[2]{\hyperref[#1]{\ref*{#1}#2}}
\newcommand{\refitem}[1]{\mbox{(\ref{#1})}}

\newcommand{\+}{\mkern2mu}
\newcommand{\coloneqq}{\mathrel{\vcentcolon\mkern-1.2mu=}} %fix a bug in the old version of amsmath
\newcommand{\eqqcolon}{\mathrel{=\mkern-1.2mu\vcentcolon}}

\newcommand{\union}{\cup}
\newcommand{\intersection}{\cap}
\newcommand{\dunion}{\mathbin{\dot{\cup}}}

\newcommand{\texteqref}[1]{Eq.~\eqref{#1}}
\newcommand{\argdot}{{\,\cdot\,}}
\newcommand{\oftype}{\colon}
\newcommand{\suchthat}{\colon}
\newcommand{\itholds}{\colon\mathchoice{\quad}{}{}{}}
\renewcommand{\max}{\mathchoice{\operatorname*{max}}{\operatorname*{max}}{\mathrm{max}}{\mathrm{max}}} %It allows you to use \max both as a math operator (like in \max_{x\in[0,1]} x^2 = 1) and in sub- and superscripts (like in f_\max) and it will in both cases be typeset in the correct font and with the correct spacing.
\renewcommand{\min}{\mathchoice{\operatorname*{min}}{\operatorname*{min}}{\mathrm{min}}{\mathrm{min}}}
\newcommand{\ltwo}{\ell^2}
\newcommand{\levelspaceing}{\epsilon}
\renewcommand{\H}{H}
\newcommand{\Basis}{\mathds{B}}
\newcommand{\CH}{\mathscr{H}}
\newcommand{\rhog}{g}
\newcommand{\rhomc}{{\sqcap}}
\newcommand{\muhaar}{\mu_{\mathrm{Haar}}}

\newcommand{\muC}{\mu_C}
\newcommand{\animalc}{\alpha}
\newcommand{\Vset}{\mathcal{V}}
\newcommand{\Eset}{\mathcal{E}}
\newcommand{\bra}[1]{\langle #1|}
\newcommand{\ket}[1]{|#1\rangle}
\newcommand{\braket}[2]{\langle #1 | #2 \rangle}
\newcommand{\ketbra}[2]{\ket{#1}\!\bra{#2}}
\newcommand{\ex}[2]{\langle #1 \rangle_{#2}}

\newcommand{\norm}[2][]{
  \ifthenelse{\equal{#1}{}}
    {\left\| {#2} \right\|}
    {\ifthenelse{\equal{#1}{uinv}}
      {\left\vert\kern-0.25ex\left\vert\kern-0.25ex\left\vert {#2} \right\vert\kern-0.25ex\right\vert\kern-0.25ex\right\vert}
      {\left\| {#2} \right\|_{#1}}
    }
}
\newcommand{\notimplies}{\centernot\implies}

\newcommand{\taverage}[2][]{
  \ifthenelse{\equal{#1}{}}
  {\overline{#2}}
  {\overline{#2}^{#1}}
}

\newcommand{\deff}{d^{\mathrm{eff}}}
\newcommand{\tracedistance}[3][]{
  \ifthenelse{\equal{#2}{}}
  {\ifthenelse{\equal{#3}{}}
    {\mathcal{D}_{#1}}{}
  }{
    \ifthenelse{\equal{#1}{}}
    {\mathchoice{\operatorname{\mathcal{D}}\left(#2,#3\right)}{\operatorname{\mathcal{D}}(#2,#3)}{\operatorname{\mathcal{D}}(#2,#3)}{\operatorname{\mathcal{D}}(#2,#3)}}
    {\mathchoice{\operatorname{\mathcal{D}}_{#1}\left(#2,#3\right)}{\operatorname{\mathcal{D}}_{#1}(#2,#3)}{\operatorname{\mathcal{D}}_{#1}(#2,#3)}{\operatorname{\mathcal{D}}_{#1}(#2,#3)}}
  }
}
\newcommand{\fidelity}[3][]{
  \ifthenelse{\equal{#2}{}}
  {\ifthenelse{\equal{#3}{}}
    {\mathcal{F}_{#1}}{}
  }{
    \ifthenelse{\equal{#1}{}}
    {\mathchoice{\operatorname{\mathcal{F}}\left(#2,#3\right)}{\operatorname{\mathcal{F}}(#2,#3)}{\operatorname{\mathcal{F}}(#2,#3)}{\operatorname{\mathcal{F}}(#2,#3)}}
    {\mathchoice{\operatorname{\mathcal{F}}_{#1}\left(#2,#3\right)}{\operatorname{\mathcal{F}}_{#1}(#2,#3)}{\operatorname{\mathcal{F}}_{#1}(#2,#3)}{\operatorname{\mathcal{F}}_{#1}(#2,#3)}}
  }
}

\DeclareMathOperator{\landauO}{O}
\DeclareMathOperator{\landauOmega}{\Omega}
\DeclareMathOperator{\landauTheta}{\Theta}
\DeclareMathOperator*{\probability}{\mathbb{P}}
\DeclareMathOperator*{\expectation}{\mathbb{E}}
\DeclareMathOperator*{\maxprime}{\mathrm{max}^\prime}
\newcommand{\Sr}[3][]{
  \ifthenelse{\equal{#1}{}}
    {\operatorname{\mathnormal{S}}(#2\|#3)}
    {\operatorname{\mathnormal{S}}_{#1}(#2\|#3)}
}
\newcommand{\dd}[1]{\mathop{\mathrm{d}#1}}
\newcommand{\del}{\mathop{}\!\partial}
\newcommand{\ddel}{\mathop{}\!d}
\newcommand{\ad}{^\dagger}
\newcommand{\compl}[1]{{{#1}^c}}
\newcommand{\trunc}[2]{{#1}_{\upharpoonright \mathnormal{#2}}}

\newcommand{\iu}{\mathrm{i}}
\DeclareMathOperator{\1}{\mathds{1}}
\DeclareMathOperator{\id}{\mathrm{id}}
\newcommand{\POVMs}{\mathcal{M}}
\DeclareMathOperator{\Bop}{\mathcal{B}}
\DeclareMathOperator{\Tcl}{\mathcal{T}}
\DeclareMathOperator{\Obs}{\mathcal{O}}
\DeclareMathOperator{\Qst}{\mathcal{S}}
\DeclareMathOperator{\Cst}{\mathscr{S}}
\DeclareMathOperator{\Qch}{\mathcal{T^+}}
\DeclareMathOperator{\Chann}{\mathnormal{C}}
\DeclareMathOperator{\Svn}{\mathnormal{S}}

\DeclareMathOperator{\supp}{supp}
\DeclareMathOperator{\cov}{cov}
\DeclareMathOperator{\dist}{d}
\DeclareMathOperator{\rank}{rank}

\DeclareMathOperator{\spec}{spec}

\newcommand{\mc}[1]{\mathcal{#1}}

\newcommand{\mcH}{\mc{H}}
\newcommand{\mcS}{\mc{S}}

\newcommand{\mcG}{\mc{G}}
\newcommand{\mcC}{\mc{C}}
\makeatletter
\newenvironment{sesac}{%
  \let\@ifnextchar\new@ifnextchar
  \left.%
  \array{@{}l@{\quad}l@{}}%
}{\endarray\right\}}
\makeatother

\newcommand{\mb}[1]{\mathbb{#1}}
\newcommand{\N}{\mb{N}}
\newcommand{\Z}{\mb{Z}}
\newcommand{\R}{\mb{R}}
\newcommand{\C}{\mb{C}} %\C seems to be defined by the russian babel package
\begin{document}

\review{\titletext}
\author{Christian Gogolin$^{1,2,3}$ and Jens Eisert$^1$}
\address{$^1$Dahlem Center for Complex Quantum Systems, Freie Universitat Berlin, 14195 Berlin, Germany\\
$^2$ICFO-Institut de Ciencies Fotoniques, Mediterranean Technology Park, 08860 Castelldefels (Barcelona), Spain\\
$^3$Max-Planck-Institut f{\"u}r Quantenoptik, Hans-Kopfermann-Stra{\ss}e 1, 85748 Garching, Germany} 
\ead{publications@cgogolin.de}
\ead{jense@physik.fu-berlin.de}

\begin{abstract}
  We review selected advances in the theoretical understanding of complex quantum many-body systems with regard to emergent notions of quantum statistical mechanics.
  We cover topics such as equilibration and thermalisation in pure state statistical mechanics, the eigenstate thermalisation hypothesis, the equivalence of ensembles, non-equilibration dynamics following global and local quenches as well as ramps.
  We also address initial state independence, absence of thermalisation, and many-body localisation.
  We elucidate the role played by key concepts for these phenomena, such as Lieb-Robinson bounds, entanglement growth, typicality arguments, quantum maximum entropy principles and the generalised Gibbs ensembles, and quantum (non-)integrability.
  We put emphasis on rigorous approaches and present the most important results in a unified language.
\end{abstract}

\maketitle

%%%%Table of contents%%%%%%%%%%%%%%%%%%%%%%%%%%%%%%%%%%
\tableofcontents \label{toc}

\section{Introduction}
\label{sec:introduction}

At the time when quantum theory was developed in a creative rush in the last years of the twenties of the previous century,
classical statistical physics was already a mature field of research.
The landmark book ``Elementary principles in statistical mechanics'' \cite{Gibbs1902} 
 authored by Gibbs had already been published in 1902, which is seen by many as the 
 the birth of modern statistical mechanics \cite{UffinkFinal}.
So, as soon as the mathematical framework of the ``neue Mechanik'', the new mechanics, as von Neumann called it in 1929, was established, significant efforts were made by him to prove ergodicity and a tendency to evolve into states that maximise entropy, which became known as the \emph{H-theorem}, in this setting \cite{vonneumann1929}.
The field of quantum statistical mechanics soon emerged and can by now be considered an important pillar of theoretical physics \cite{Landau1980}.

Still, some foundational yet fundamental questions remain open, much related to the questions raised by von Neumann.
While maximum entropy principles provide a starting point for the understanding of the ensembles of quantum statistical mechanics, it seems much less clear how quantum states taking extremal values for the entropy are being achieved via microscopic dynamics.
After all, at the fundamental level, quantum many-body systems follow the Schr{\"o}dinger equation, giving rise to unitary dynamics.
It is far from obvious, therefore, in what precise way interacting quantum many-body systems can equilibrate.
The microscopic description of quantum mechanical systems following the dynamical equations of motion is still in some tension with the picture arising from the ensemble description of quantum statistical mechanics.
That is to say, the questions of equilibration and thermalisation in what we will call \emph{pure-state quantum statistical mechanics} remained largely unresolved until very recently.

These foundational questions came back with a vengeance not too long ago.
Old puzzles and new questions of quantum many-body systems out of equilibrium re-entered the centre of attention and are again much in the focus of present-day research in quantum many-body theory.
This remarkable renaissance is primarily due to three concomitant factors in physics research.

The first and arguably most important source of inspiration has been an \emph{experimental revolution}.
Fueled by enormous improvements in experimental techniques it became feasible to control quantum systems with many degrees of freedom.
An entirely new arena for the study of physics of interacting quantum many-body systems emerged.
This is particularly true for the development of techniques to cool and trap
\emph{ultra-cold atoms} and to subject them to \emph{optical lattices} \cite{Greiner2002a,Greiner2002,Tuchman2006,Aidelsburger2011,Bloch2005,Langen2014a}
or suitable confinements, giving rise to \emph{low-dimensional continuous systems} \cite{Sadler2006,Regal2006,Kinoshita2006,Hofferberth2007,Weller2008,Strohmaier2007}.
Similarly, systems of \emph{trapped ions} \cite{Porras2004,Friedenauer2008}, 
as well as hybrid systems \cite{Bissbort2013}, allow to precisely study the physics of interacting systems in the laboratory
\cite{Haffner2005,Jurcevic2014,Lanyon_etal11,Schindler2012,Blatt2012,Britton2012}.
In such highly controlled settings, equilibration and thermalisation dynamics has been studied \cite{1101.2659v1,1111.0776v1,Langen2013,1112.0013v1,Ronzheimer2013}.
Especially setups with optical lattices allow for the realisation of condensed-matter-like interacting many-body systems in the laboratory, but with fine grained control over the model parameters and geometries.
Questions concerning the out of equilibrium
dynamics of such systems were suddenly not only important out of academic curiosity, but became pragmatically motivated questions important to understand experimentally realisable physical situations.

The second major development is the broad availability of new machines: \emph{supercomputers}.
With the vastly increased computing power and massive parallelisation as well as novel \emph{numerical techniques} such as tensor network methods \cite{OrusReview} and the 
density-matrix renormalisation group method \cite{Schollwock201196}, it has become possible to simulate the dynamics of large quantum systems for relatively long times.
Methods for exact diagonalisation have been brought to new levels \cite{Rigol08,Rigol09,Luitz2014}, complemented by quantum Monte Carlo techniques \cite{TimeMonteCarlo,Foulkes2001}, and applications of dynamical mean field theory \cite{Eckstein_DMFT}  and density functional theory \cite{Engel2011,Kohn1999}.
There is an enormous body of numerical works on questions of equilibration and thermalisation in closed quantum systems and the dynamics of quantum phase transitions \cite{Braun2014a,Rigol07,Lesanovsky10,1112.3424v1.pd,Beugeling2013,Yukalov2011,Ikeda2013a,Fine2013,Jensen1985}, often with a focus on so-called \emph{quenches}, i.e., rapid changes of the Hamiltonian \cite{Moeckel2008,Kollath07,Rigol08,Rigol09,1011.0781v1,Venuti09,1108.2703v1,Rigol11,Torres-Herrera2013}.
This body of numerical work is complemented by partly or entirely analytical studies that capture these and related phenomena in concrete systems or classes of models (often integrable ones) \cite{Sengupta2004,Flesch08,Ates2011,Fagotti2012,Calabrese2007,Cazalilla2006,Fioretto2010,1104.0154v1,Eckstein2008,Campos10}.
We will discuss these works in more detail later.

Last but not least, our understanding of quantum mechanics has improved significantly since the time of von Neumann.
The availability of new \emph{mathematical methods} ---  in part motivated by research in quantum information theory --- is the third driving force.
These techniques have made some of the old questions become more tractable than before, while at the same time
new paradigms of approaching the key questions have emerged.
This lead to works inspired by notions of \emph{typicality} and \emph{random states} \cite{slloydthesis,Popescu05,Gemmer09,Garnerone2013a,Garnerone2013}.
Also, notions of quantum information propagation, such as \emph{Lieb-Robinson bounds} \cite{Hastings2006,Nachtergaele_Locality,Bravyi06-1,Kliesch2013}, and research on \emph{entanglement in many-body systems} \cite{Calabrese2007a,DeChiara2005,Eisert06,Kollath08,VanAcoleyen2013,Jurcevic2014} can be classified as
contributing to this development.

All in all, this is already too large a topic to cover in full depth in a single review of reasonable length.
Hence, in this article we address and cover only a subset of these developments and questions.
We will mostly concentrate on the theoretical and analytical insights, however always making an effort to put them into the context of evidence collected through numerical simulations and important experimental developments.

In physics, one can often say a ``lot about little'', or ``little about a lot''.
In this review, we take the latter approach, by sticking to general and conceptual
statements on interacting many-body systems in a quantum-information inspired rigorous language, so where only relatively ``little'' can be said.
These statements, however, apply to
``a lot'', that is, to an immense variety of models.

At the heart of the approach advocated here lies the attempt to use only standard quantum mechanics and no additional postulates to explain the emergence of thermodynamic behaviour, and to do this in a mathematically rigorous and general way.
It is an invitation to elaborate how much of statistical mechanics and thermodynamics can be \emph{derived} from quantum mechanics.
The term ``\emph{derived}'' here means to justify the well established methods and postulates of equilibrium and non-equilibrium statistical 
mechanics by means of the microscopic picture provided by quantum mechanics.
Following Refs.~\cite{slloydthesis,lloyed13}, we shall call this approach \emph{pure state 
quantum statistical mechanics}.

The level of detail and rigour that we are aiming at in this work necessarily also mean that an awkwardly large number of interesting questions and research results 
will have to be left unmentioned.
In this sense, this article is not meant to be a comprehensive review.
\begin{enumerate}
\item We have authored together with Mathis Friesdorf a complementing accompanying review \cite{Eisert2014} in Nature Physics 
that takes a much more physical perspective, where local interacting many-body systems out of equilibrium 
are in the focus of attention and experimental developments are more comprehensively discussed.
\end{enumerate}
A lot what is left out here is covered there.
In this article, in contrast, we advocate a more mathematical mindset and language, and at the same time have a more limited scope, but the covered topics are discussed in more depth.
To complete the picture of the subject, and in addition to Ref.~\cite{Eisert2014}, we recommend a number
of further review articles and books that cover what we do not have the space to cover here:

\begin{enumerate}
\setcounter{enumi}{1}
\item  The book by Gemmer, Michel, and Mahler \cite{Gemmer09} entitled ``Quantum thermodynamics''
advertises an approach towards the foundations of thermodynamics that is in spirit close to the approach of this work.
The focus is, however, more on notions of \emph{typicality}, which we will discuss in Section~\ref{sec:typicality}, but which is not a central topic of the present review.
Moreover, the first edition of the book is from 2004, and even though it has been extended in the second edition from 2009, much of the newer material that takes the centre stage in this work is not covered.

\item The editorial of a New Journal of Physics focus issue on the ``Dynamics and thermalisation in isolated quantum many-body systems'' by Cazalilla and Rigol \cite{Cazalilla10}
not only explains the significance of the individual articles published in the focus issue to the more general endeavour of developing a better understanding of the coherent dynamics of quantum many-body systems.
On top of that it gives an overview of many of the currently pursued research directions and many additional references.
This renders this editorial a good entry point into the more recent literature on the subject and makes it an excellent read.
At the same time, it provides only very little background information, almost no historical context, and assumes that the reader is already familiar with the jargon of the field.

\item A colloquium in Reviews in Modern Physics by Polkovnikov, Sengupta, Silva, and Vengalattore \cite{Polkovnikov11} is entitled ``Non-equilibrium dynamics of closed interacting quantum systems''.
This work gives an excellent overview of recent theoretical and experimental insights concerning such systems, but focuses mainly on the dynamics following so-called \emph{quenches}, i.e., rapid changes in the Hamiltonian of a system and the \emph{eigenstate thermalisation hypothesis} (ETH).
We will discuss the ETH in Section~\ref{sec:thermalisationunderassumptionsontheeigenstates}, but the scope of the present work is considerably broader and we will also take a slightly different, quantum information theory inspired, point of view and put the focus more on analytical results.

\item  A review entitled ``Equilibration and thermalisation in finite quantum systems'' by Yukalov \cite{Yukalov2011}
contains a review of the history of both the experimental realisation of coherently evolving, well controlled quantum systems and the observation and numerical investigation of 
equilibration and thermalisation in such systems.
In addition it contains results on equilibration in closed systems with a continuous density of states and in systems undergoing so-called \emph{non-destructive measurements}.

\item The review Ref.~\cite{Seifert2012} on the thermodynamics of stochastic processes.
It covers important topics such as fluctuation(-dissipation) theorems, entropy production, and (autonomous) thermal machines, which have been extensively studied in recent years and which are related to but not elaborated on in this work.

\item \label{item:lastrecomandedworkitem} Finally, the review ``The role of quantum information in thermodynamics'' \cite{QuantumThermoReview2015} overviews recent developments in the interplay between the fields of quantum information and thermodynamics.
  It focuses on foundations of statistical mechanics and on resource-theoretic aspects of thermodynamics.
  More explicitly, it covers equilibration and thermalisation, state transformation under different constraints and resources, work extraction, the work cost of information-processing tasks, inconvertibility of energy and correlations, and fluctuation relations.
\end{enumerate}

These articles and books together, in conjunction with the present review, rather accurately cover the state of affairs.
It is the purpose of this article to fill the gap left by the above mentioned works.

\section{Preliminaries and notation}
\label{sec:preliminaries}
In order to facilitate the discussion in later chapters, we carefully introduce the notation and introduce a number of fundamental concepts in this section.
The presentation is limited to the minimum necessary to make the following statements well-defined.
An effort has been made to make this introduction self-contained.
However, a basic knowledge of quantum theory, analysis, linear algebra, group theory and related subjects is assumed.

To begin with, we fix some general notation.
Given a positive integer $n \in \Z^+$ we use the short hand notation $[n] \coloneqq (1,\dots,n)$ for the (ordered) \emph{range} of numbers from $1$ to $n$ and set $[\infty] \coloneqq \Z^+$.
Given a \emph{set} $X$ we denote its \emph{cardinality} by $|X|$.
If $X$ has a \emph{universal superset} $\Vset \supset X$, we write $\compl X \coloneqq \Vset \setminus X$ for its complement.
Given two sets $X,Y$ we write $X \union Y$ and $X \intersection Y$ for their \emph{union} and \emph{intersection}.
To stress that a set $\Vset$ is the union of two \emph{disjoint} sets $X,Y$, i.e., $X \intersection Y = \emptyset$ we write 
$\Vset = X \dunion Y$.
Given a set $X$ of sets we write $\union X \coloneqq \bigcup_{x \in X} x$ for the union of the sets in $X$.
For \emph{sequences} $S$, $|S|$ denotes the length of the sequence.
When we define sets or sequences in terms of their elements we use curly $\{\argdot\}$ or round $(\argdot)$ brackets respectively.

We use the \emph{(Bachmann-)Landau} symbols $\landauO$, $\landauOmega$ and $\landauTheta$ to denote asymptotic growth rates of real functions $f,g\oftype\R\to\R$.
In particular
\begin{align}
  f(x) \in \landauO(g(x)) &\iff \limsup_{x\to\infty}|f(x)/g(x)| < \infty ,\\
\intertext{and for $\landauOmega$ we adopt the convention from complexity theory that}
  f(x) \in \landauOmega(g(x)) &\iff g(x) \in \landauO(f(x)) 
\end{align}
and write $f(x) \in \landauTheta(g(x))$ if both $f(x) \in \landauO(g(x))$ and $f(x) \in \landauOmega(g(x))$.
To simplify the notation we work with \emph{natural}, or \emph{Planck units} such that in particular the Planck constant $\hbar$ and the Boltzmann constant $k_B$ are equal to $1$.

Let $\mcH$ be a \emph{separable Hilbert space} over $\C$ with inner product $\braket\varphi\psi$ for $\ket\varphi,\ket\psi\in\mcH$.
We denote by $\Bop(\mcH)$ be the space of \emph{bounded operators} and by $\Tcl(\mcH)$ the space of \emph{trace class operators} on the Hilbert space $\mcH$, i.e., those $A \in \Bop(\mcH)$ whose trace $\Tr A$ is finite.
The trace class operators $\rho \in \Tcl(\mcH)$, whose associated linear functional $\Tr(\rho\,\argdot)$ is non-negative, i.e., $\forall A \geq 0\itholds \Tr(\rho\,A) \geq 0$ and which have unit trace $\Tr \rho = 1$, form the convex set $\Qst(\mcH)$ of \emph{(quantum) states} or \emph{density operators}. An operator $A\in \Bop(\mcH)$ is \emph{self-adjoint} if $A = A\ad$.
An operator $\Pi \in \Bop(\mcH)$ is a \emph{projector} if $\Pi\,\Pi = \Pi$.
The \emph{rank} of an operator $A \in \Bop(\mcH)$, denoted by $\rank(A)$, is the dimension of its image.
An operator $U \in \Bop(\mcH)$ is called \emph{unitary} if $U\ad\,U = U\,U\ad = \1$.
It turns out that in the finite dimensional setting considered here $\Qst(\mcH) \subset \Obs(\mcH)$ is the convex set of self-adjoint, non-negative operators with unit trace.
The extreme point of that set are rank one projectors and are called \emph{pure states}.
The elements of the subspace $\Obs(\mcH) \subset \Bop(\mcH)$ of self-adjoint operators are called \emph{observables}.

Given a bounded operator $A \in \Bop(\mcH)$ and a state $\rho \in \Qst(\mcH)$, we will write the \emph{expectation value} of $A$ in state $\rho$ is as
\begin{equation} \label{eq:expectationvalue}
 \ex A \rho \coloneqq \Tr(A\,\rho).
\end{equation}

The most general \emph{measurements} possible in quantum mechanics are so-called \emph{positive operator valued measurements} (POVMs) \cite{nielsenchuang}.
A POVM with $K$ measurement outcomes is a sequence $M = (M_k)_{k=1}^K$ of operators $M_k \in \Bop(\mcH)$, called \emph{POVM elements}, with the property that
\begin{equation}
  \sum_{k=1}^K M_k = \1 .
\end{equation}
Upon measuring a system in state $\rho \in \Qst(\mcH)$ with the POVM $M$, outcome number $k$ is obtained with probability $\Tr(M_k\,\rho)$.
When we say that an observable $A \in \Obs(\mcH)$, with \emph{spectral decomposition} $A = \sum_{k=1}^{d_A} a_k\,\Pi_k$, is \emph{measured}, we mean that the POVM $M = (\Pi_k)_{k=1}^{d_A}$ is measured and the measurement device outputs the value $a_k$ when outcome $k$ is obtained.
The average value output by the device in measurements of identically prepared systems is then indeed given by \texteqref{eq:expectationvalue}.
A measurement of a POVM where all the POVM elements are projectors is called a \emph{projective measurement}.
The \emph{measurement statistic} of a POVM in a state $\rho$ is the vector of probabilities $\Tr(M_k\,\rho)$.

The most general \emph{(quantum) operations} in quantum mechanics are captured by so-called \emph{completely positive trace preserving maps}, also-called \emph{quantum channels} \cite{nielsenchuang}.
We call maps $\Bop(\mcH)\to\Bop(\mcH)$ \emph{superoperators}.
We denote the identity superoperator by $\id\oftype\Bop(\mcH)\to\Bop(\mcH)$.
A linear map $\Chann\oftype\Obs(\mcH)\to\Obs(\mcH)$ is then called \emph{completely positive trace preserving} if for all separable Hilbert spaces $\mcH'$ it holds that 
\begin{equation} \label{eq:conditionforcompletepositivity}
  \forall \rho \in \Qst(\mcH \otimes \mcH')\itholds (\Chann \otimes \id)\,\rho \in \Qst(\mcH \otimes \mcH') .
\end{equation}
In the finite dimensional setting considered here, it turns out that fixing $\mcH' = \mcH$ in \texteqref{eq:conditionforcompletepositivity} already gives a necessary and sufficient condition for a map $\Chann\oftype\Obs(\mcH)\to\Obs(\mcH)$ to be completely positive trace preserving \cite{nielsenchuang}.
We denote the set of all completely positive trace preserving maps on $\Qst(\mcH)$ by $\Qch(\mcH)$.

Throughout most of this review we will work in the framework of finite dimensional quantum mechanics.
That is, if not explicitly stated otherwise, we consider systems that are described by a Hilbert space $\mcH$ over $\C$ whose dimension $d\coloneqq\dim(\mcH)$ is finite,
bosonic systems constituting an important exception.

For every $1\leq p<\infty$ the \emph{Schatten $p$-norm} of an operator $A \in \Bop(\mcH)$ is defined as \cite{bhatia}
\begin{equation} \label{eq:schattenpnorms}
  \norm[p]A \coloneqq \left[\sum_{j=1}^d (s_j(A))^p \right]^{1/p} ,
\end{equation} 
where $(s_j(A))_{j=1}^d$ is the ordered, i.e., $s_1(A) \geq \dots \geq s_d(A)$, sequence of non-negative, real \emph{singular values} of $A$.
We refer to the Schatten $\infty$-norm as the \emph{operator norm} and call the Schatten $1$-norm \emph{trace norm}
The Schatten $p$-norms are ordered in the sense that \cite{bhatia}
\begin{equation}
  \forall A\in\Bop(\mcH)\colon \norm[p]A \leq \norm[p']A \iff p\geq p' 
\end{equation}
and in the converse direction the following inequalities hold \cite{bhatia}
\begin{equation}
  \norm[1]\argdot \leq \sqrt{d} \norm[2]\argdot \leq d \norm[\infty]\argdot .
\end{equation}

For quantum states a natural and frequently used distance measure is the \emph{trace distance} \cite{nielsenchuang}
\begin{equation}
  \forall \rho,\sigma\in\Qst(\mcH)\itholds \tracedistance\rho\sigma \coloneqq \frac{1}{2}\,\norm[1]{\rho-\sigma} .
\end{equation}
It is, up to the factor of $1/2$, the metric induced by the trace norm $\norm[1]\argdot$.
Its relevance stems from the fact that it is equal to the maximal difference between the expectation values of all normalised observables in the states $\rho$ and $\sigma$, i.e., \cite{nielsenchuang}
\begin{equation} \label{eq:tracedistanceasmaxoverobservables}
  \tracedistance\rho\sigma = \max_{A\in\Obs(\mcH)\colon 0\leq A\leq\1} \Tr(A\,\rho) - \Tr(A\,\sigma) .
\end{equation}
The trace distance is non-increasing under completely positive trace preserving maps $\Chann \in \Qch(\mcH)$, i.e., $\tracedistance{\Chann(\rho)}{\Chann(\sigma)} \leq \tracedistance\rho\sigma$ and invariant under unitary operations, i.e., $\tracedistance{U\,\rho\,U\ad)}{U\,\sigma\,U\ad} = \tracedistance\rho\sigma$.
Moreover, if one is given an unknown quantum system and is promised that with probability $1/2$ it is either in state $\rho$ or state $\sigma$, then the maximal achievable probability $p_{\max}$ for correctly identifying the state after a single measurement of the optimal observable from \texteqref{eq:tracedistanceasmaxoverobservables} is given by \cite{1012.4622v1,Aubrun2013}
\begin{equation} \label{eq:maxsucessprobabilityfordistinguishinginsingelshotmeasurement}
  p_{\max} = \frac{1+\tracedistance\rho\sigma}{2} .
\end{equation}
Inspired by this, one can define the \emph{distinguishability} of two quantum states under a restricted set $\POVMs$ of POVMs.
The optimal success probability for single shot state discrimination is then again given by an expression of the form \eqref{eq:maxsucessprobabilityfordistinguishinginsingelshotmeasurement}, but with $\tracedistance\rho\sigma$ replaced by \cite{1110.5759v1}
\begin{equation} \label{eq:distinguishabilityunderrestrictedsetsofpovms}
  \tracedistance[\POVMs]\rho\sigma \coloneqq \sup_{M \in \POVMs} \frac{1}{2}\,\sum_{k=1}^{|M|} |\Tr(M_k\,\rho) - \Tr(M_k\,\sigma)| ,
\end{equation}
and it holds that
\begin{equation} \label{eq:tracedistanceboundsrestrictedtracedistance}
 \tracedistance[\POVMs]\rho\sigma \leq \tracedistance\rho\sigma .
\end{equation}
with equality for all $\rho,\sigma \in \Qst(\mcH)$ if and only if $\POVMs$ is a dense subset of the set of all POVMs \cite{1110.5759v1}.
It is worth noting that $\tracedistance[\POVMs]\argdot\argdot$ is a pseudometric on $\Qst(\mcH)$, i.e., it is a symmetric, positive semidefinite bilinear form, but $\tracedistance[\POVMs]\rho\sigma = 0 \notimplies \rho = \sigma$.
For further properties of the distinguishability $\tracedistance[\POVMs]{}{}$ see for example Ref.~\cite{Aubrun2013}.

Another frequently employed distance measure is the \emph{fidelity}, defined for any two quantum states $\rho,\sigma \in \Qst(\mcH)$ as\footnote{Some authors define the fidelity as the square root of the $\fidelity{}{}$ used here.}
\begin{equation}
  \fidelity\rho\sigma \coloneqq \Tr\left(\left({\rho^{1/2}\,\sigma\,\rho^{1/2}}\right)^{1/2}\right)^2 .
\end{equation}
Similar to the trace distance, the fidelity is symmetric, i.e., $\fidelity\rho\sigma = \fidelity\sigma\rho$, non-decreasing under completely positive maps, and invariant under unitary operations.
The fidelity is not a metric, but it is related to the trace distance via $1-{\fidelity\rho\sigma}^{1/2} \leq \tracedistance\rho\sigma \leq (1-\fidelity\rho\sigma)^{1/2}$.
For pure states $\psi = \ketbra\psi\psi$ and $\varphi = \ketbra\varphi\varphi$ it reduces to the square of their overlap $\fidelity\psi\varphi = |\braket\psi\varphi|^2$.

The (time independent) \emph{Hamiltonian} $\H \in \Obs(\mcH)$ of a finite dimensional quantum system has the \emph{spectral decomposition}
\begin{equation} \label{eq:hamiltonianspectraldecomposition}
  \H = \sum_{k=1}^{d'} E_k\,\Pi_k
\end{equation}
where the $\Pi_k \in \Obs(\mcH)$ are its orthogonal (and mutually orthogonal) \emph{spectral projectors} and $d' \coloneqq |\spec(\H)| \leq d = \dim(\mcH)$ is the number of distinct, ordered \emph{(energy) eigenvalues} $E_k \in \R$ of $\H$, i.e., $k<l \implies E_k < E_l$.
The subspaces on which the $\Pi_k$ project are called \emph{(energy) eigenspaces} or \emph{energy levels}.
If $\H$ is non-degenerate it holds that $\Pi_k = \ketbra{E_k}{E_k}$ with $(\ket{E_k})_{k=1}^d$ a sequence of orthonormal \emph{energy eigenstates} of $\H$ and $d \coloneqq \dim(\mcH)$ the dimension of $\mcH$.

The Hamiltonian $\H$ governs the \emph{time evolution} $\rho \colon \R \to \Qst(\mcH)$ of the state of a quantum system via the \emph{(Schr{\"o}dinger-)von-Neumann-equation}, which in the Schr{\"o}dinger picture reads
\begin{equation}
  \frac{\del}{\del t} \rho(t) = - \iu [\H,\rho(t)] .
\end{equation}
Its formal solution can be given in terms of the \emph{time evolution operator}, which in the case of time independent Hamiltonian dynamics is given by the operator exponential
\begin{equation} \label{eq:timeevolution}
  \forall t\in\R\itholds U(t) \coloneqq \e^{-\iu\,\H\,t} \in \Bop(\mcH) .
\end{equation}
The time evolved quantum state at time $t$ is then
\begin{equation}
  \rho(t) \coloneqq U^\dagger(t)\,\rho(0)\,U(t) ,
\end{equation}
with $\rho(0)$ the \emph{initial state} at time $t=0$.

The temporal evolution of the expectation value of an observable $A \in \Obs(\mcH)$ then solves
\begin{equation} \label{eq:timeevolutionofexpectationvalueinschroedingerpicture}
  \ex{A}{\rho(t)} = \Tr\big(A\,U^\dagger(t)\,\rho(0)\,U(t)\big) = \Tr\big(U(t)\,A\,U^\dagger(t)\,\rho(0)\big) .
\end{equation}
One can thus equally well-define the time evolution of an observable $A\oftype \R \to \Obs(\mcH)$, with the initial value $A(0)$ given by the operator $A$ from \texteqref{eq:timeevolutionofexpectationvalueinschroedingerpicture}, by setting $A(t) \coloneqq U(t)\,A(0)\,U^\dagger(t)$, and consider a fixed quantum state $\rho \in \Qst(\mcH)$, equal to the initial state $\rho(0)$ in \texteqref{eq:timeevolutionofexpectationvalueinschroedingerpicture}.
Then $\ex{A(t)}{\rho}$ is equal to $\ex{A}{\rho(t)}$ from \texteqref{eq:timeevolutionofexpectationvalueinschroedingerpicture} for all $t \in \R$.
The time evolution $A\oftype \R \to \Obs(\mcH)$ of an observable in the Heisenberg picture solves the differential equation
\begin{equation}\label{eq:schroedingervonneumannequationinheisenbergpicture}
  \frac{\del}{\del t} A(t) = \iu [\H,A(t)] .
\end{equation}

We call all observables $A \in \Obs(\mcH)$ that commute with the Hamiltonian, i.e., for which $[\H,A] \coloneqq A\,B - B\,A = 0$, \emph{conserved quantities}.
It follows directly from \texteqref{eq:schroedingervonneumannequationinheisenbergpicture} that the expectation value of all conserved quantities is independent of time, irrespective of the initial state, which justifies the name.
If the Hamiltonian $\H$ is non-degenerate, then exactly the observables that are diagonal in the same basis as $\H$ are conserved quantities.
In the presence of degeneracies exactly the observables $A \in \Obs(\mcH)$ for which some basis exists in which both $A$ and $\H$ are diagonal are conserved quantities.

Given a function $f$ depending on time, we define its \emph{finite time average}
\begin{equation} \label{eq:generalfinitetimeaverage}
  \taverage[T]{f} \coloneqq \frac{1}{T}\,\int_0^T f(t) ,
\end{equation}
and its \emph{(infinite) time average}
\begin{equation} \label{eq:generalinfinitetimeaverage}
  \taverage{f} \coloneqq \lim_{T\to\infty} \taverage[T]{f} ,
\end{equation}
whenever the limit exists.
In all cases we will be interested in, the existence of the limit in \texteqref{eq:generalinfinitetimeaverage} is guaranteed by the theory of \emph{(Besicovitch) almost-periodic} functions \cite{Besicovitch1926}.

In particular we will encounter the \emph{time averaged state} $\omega \coloneqq \taverage{\rho}$, which is, in the finite dimensional case considered here, equal to the initial state $\rho(0)$ \emph{dephased} with respect to the Hamiltonian $\H$, i.e., $\omega = \$_\H(\rho(0))$, with the \emph{de-phasing map}
acting as
\begin{equation}
  \rho\mapsto  \$_\H(\rho) \coloneqq \sum_{k=1}^{d'} \Pi_k\, \rho\, \Pi_k
\end{equation}
and $(\Pi_k)_{k=1}^{d'}$ the sequence of orthogonal spectral projectors of $\H$.

We will encounter systems consisting of smaller subsystems.
Often their Hamiltonian can be written as a sum of Hamiltonians that each act non-trivially only on certain subsets of the whole system.
We will refer to such systems as \emph{composite (quantum) systems} or as \emph{locally interacting (quantum) systems}, depending on whether we want to stress that they consist of multiple parts or that the interaction between the parts has a special structure.
The notion of locally interacting quantum systems can be formalised by means of an \emph{interaction (hyper)graph} $\mcG \coloneqq (\Vset, \Eset)$, which is a pair of a \emph{vertex set} $\Vset$ and an \emph{edge set} $\Eset$.

The \emph{vertex set} $\Vset$ is the set of indices labeling the sites of the system and we will work under the assumption that $|\Vset| < \infty$.
The Hilbert space $\mcH$ of such a system is either, in the case of spin systems, the \emph{tensor product} $\bigotimes_{x\in\Vset} \mcH_{\{x\}}$ of the Hilbert spaces $\mcH_{\{x\}}$ of the individual sites $x \in \Vset$, or, in the case of fermionic or bosonic systems, the \emph{Fock space}, or a subspace of the latter.

We will encounter bosons, which usually need to be described using infinite dimensional Hilbert spaces, only in Section~\ref{sec:equlibrationinthestrongsense}, hence we want to avoid the technicalities of a proper treatment of infinite dimensional Hilbert spaces and unbounded operators in the framework of functional analysis.
We will thus only introduce the minimal notation necessary to formulate the statements we will discuss in Section~\ref{sec:equlibrationinthestrongsense}.

The sites $x \in \Vset$ of fermionic and bosonic composite systems are often called \emph{modes}.
In the case of fermions each mode is equipped with the Hilbert space $\mcH^f_{\{x\}} = \C^2$ with orthonormal basis $((\ket n_f)_{n=0}^1$, and in the case of bosons with the Hilbert space $\mcH^b_{\{x\}} = \ltwo$ of square summable sequences with orthonormal basis $(\ket n_b)_{n=0}^\infty$.
For composite systems with exactly $N$ fermions or bosons in $M$ modes, i.e., $\Vset = [M]$, the Hilbert space is given by a so-called \emph{Fock layer}.
The Fock layer to particle number $N$ is the complex span of the orthonormal \emph{Fock (basis) states} $\ket{n_1,\dots,n_M}_f$ or $\ket{n_1,\dots,n_M}_b$ respectively, where for each $x \in \Vset$, $n_x$ is the number of particles in mode $x$ and thus $\sum_{x \in \Vset} n_x = N$ with $n_x \in \{0,1\}$ in the case of fermions, and $n_x \in [N]$ in the case of bosons.

The full \emph{Fock space} of a system of fermions or bosons is the Hilbert space completion of the direct sum of the Fock layers for each possible total particle number.
For fermions it holds that $N \leq M$ due to the \emph{Pauli exclusion principle}, and the resulting Hilbert space is hence finite dimensional.
In the case of bosons $N$ is independent of $M$ and the Fock space is thus infinite dimensional already for a finite number of modes.

We define the fermionic and bosonic \emph{annihilation operators} $f_x$ and $b_x$ on site $x$ and the corresponding \emph{creation operators} $f\ad_x$ and $b\ad_x$ (collectively often referred to as simply the \emph{fermionic/bosonic operators}) via their action on the Fock basis states given by
\begin{align}
  f_x \ket{n_1,\dots,n_M}_f &=  n_x (-1)^{\sum_{y=1}^{x-1} n_y} \ket{\dots,n_{x_1},n_x - 1,n_{x+1},\dots}_f ,\label{eq:fermionicannihilationoperator}\\
  f\ad_x \ket{n_1,\dots,n_M}_f &=  (1-n_x) (-1)^{\sum_{y=1}^{x-1} n_y} \ket{\dots,n_{x_1},n_x + 1,n_{x+1},\dots}_f \\
  \intertext{and}
  b_x \ket{n_1,\dots,n_M}_b &=  \sqrt{n_x} \ket{n_1,\dots,n_{x_1},n_x - 1,n_{x+1},\dots,n_M}_b ,\label{eq:bosonicannihilationoperator}\\
  b\ad_x \ket{n_1,\dots,n_M}_b &=  \sqrt{n_x+1} \ket{n_1,\dots,n_{x_1},n_x + 1,n_{x+1},\dots,n_M}_b .
\end{align}
They satisfy the \emph{(anti) commutation relations}
\begin{align}
  \{f_x,f_y\} = \{f\ad_x,f\ad_y\} &= 0 , & \{f_x,f\ad_y\} &= \delta_{x,y} ,\\
  [b_x,b_y] = [b\ad_x,b\ad_y] &= 0 , & [b_x,b\ad_y] &= \delta_{x,y} , \label{eq:bosonicommutationrelations}
\end{align}
where for any two operators $A,B \in \Bop(\mcH)$ $[A,B] \coloneqq A\,B - B\,A$ is the \emph{commutator} and $\{A,B\} \coloneqq A\,B + B\,A$ the \emph{anti-commutator}.
We say that $A,B$ \emph{commute} or \emph{anti-commute} if $[A,B] = 0$ or $\{A,B\}=0$ respectively.

Any operator that commutes with the \emph{total particle number operator} $\sum_{x\in\Vset} f\ad_x\,f_x$ or $\sum_{x\in\Vset}  b\ad_x\,b_x$ respectively is called \emph{particle number preserving}.
In systems with particle number preserving Hamiltonians a constraint on the particle number can be used to make the description of bosonic systems with finite dimensional Hilbert spaces possible.
The Hilbert space is then a finite direct sum of Fock layers.
We say that a state has a \emph{finite particle number} if it is completely contained in such a finite direct sum of Fock layers.

In systems of fermions, all operators can be written as polynomials of the fermionic operators.
A polynomial of fermionic operators is called \emph{even/odd} if it can be written as a linear combination of monomials that are each a product of an even/odd number of creation and annihilation operators.
According to the \emph{fermion number parity superselection rule} \cite{Banuls2009}, only observables that are even polynomials in the fermionic operators can occur in nature.
The same holds for the Hamiltonians and density matrices of such systems.
Consequently, whenever we make statements about systems of fermions we assume that all observables, states and the Hamiltonian are even.

We refer to subsets of the vertex set $\Vset$ as \emph{subsystems}.
Generalising the notation introduced for the Hilbert spaces of the individual sites we denote the Hilbert spaces associated with a subsystem $X \subseteq \Vset$ by $\mcH_X$ and its dimension by $d_X\coloneqq\dim(\mcH_X)$.
In the case of composite systems of fermions or bosons it is understood that if an upper bound on the total number of particles has been imposed, then $\mcH_X$ is taken to be the direct sum of Fock layers corresponding to the sites in $X$ up to the total number of particles.
The \emph{size} of a (sub)system $X \subseteq \Vset$ is given by the number of sites or modes $|X|$, not the dimension of the corresponding Hilbert space.

For spin systems we define the \emph{support} $\supp(A)$ of an operator $A \in \Bop(\mcH)$ as the smallest subset of $\Vset$ such that $A$ acts like the identity outside of $X$.
For systems of fermions or bosons we define the support of an operator via its representation as a polynomial in the respective creation and annihilation operators.
The \emph{support} is then the set of all site indices $x\in\Vset$ for which the polynomial contains a fermionic or bosonic operator acting on site $x$, e.g., $b_x\ad$ or $f_x$.
The support of a POVMs is simply the union of the supports of its POVM elements.
Similarly, we define the \emph{support} $\supp(\Chann)$ of a superoperator $\Chann\oftype\Bop(\mcH)\to\Bop(\mcH)$ as the smallest subset of $\Vset$ such that
\begin{equation}
  \forall A \in \Bop(\mcH) \itholds \supp(A) \subseteq \compl{\supp(\Chann)} \implies \Chann(A) = A .
\end{equation}
We say that an observable, POVM, or superoperator is \emph{local} if the size of its support is small compared to and/or independent of the system size.

In order to fully exploit the notion of a subsystem we need to understand how the description of a joint system fits together with the description of a subsystem as an isolated system, i.e., how systems can be combined and decomposed.
For every subsystem $X \subseteq \Vset$ there is a \emph{canonical embedding} of $\Bop(\mcH_X)$ into $\Bop(\mcH)$ that bijectively maps $\Bop(\mcH_X)$ onto the subalgebra of bounded linear operators $A \in \Bop(\mcH)$ with $\supp(A) \subseteq X$, and similarly for all operators that are polynomials of bosonic operators.
In the case of spin systems the embedding is simply the natural embedding $A \in \Bop(\mcH_X) \mapsto A \otimes \1_{\compl X} \in \Bop(\mcH)$, where $\1_{\compl X}$ denotes the identity operator on $\mcH_{\compl X}$.
In systems of fermions or bosons we associate to each operator on $\mcH_X$ the operator on $\mcH$ that has the same representation as a polynomial in the fermionic/bosonic operators, but, of course, in terms of the fermionic/bosonic operators of the full system with Fock space $\mcH$ rather than the fermionic/bosonic operators that act on $\mcH_X$.
For systems of fermions, because of the phase in \texteqref{eq:fermionicannihilationoperator} that depends non-locally on the state, this embedding depends on the exact position the sites in $X$ have in the vertex set $\Vset$.
The vertex set should hence rather be called \emph{vertex sequence}, but for even operators the phases cancel out, which is why we ignore this subtlety.

Conversely, for any $A \in \Bop(\mcH)$ and any subsystem $X \subseteq \Vset\colon X \supseteq \supp(A)$ that contains $\supp(A)$ we define the \emph{truncation} $\trunc A X \in \Bop(\mcH_X)$ of $A$ as the operator that acts on the sites/modes in the subsystem $X$ ``in the same way'' as $A$, in the sense that a truncation followed by a canonical embedding gives back the original operator.
In particular, for spin systems any $A \in \Bop(\mcH)$ is of the form $A = \trunc A {\supp(A)} \otimes \1_{\compl{\supp(A)}}$.
For general systems, the identity operator $\1$ of course satisfies $\1_X = \trunc \1 X$ for any $X \subset \Vset$.

We now turn to the edge set.
The \emph{edge set} $\Eset$ is the set of all subsystems $X \subset \Vset$ for which a non-trivial Hamiltonian term $\H_X$ with $\supp(\H_X) = X$ exists that couples the sites in $X$.
The Hamiltonian of a locally interacting quantum system with edge set $\Eset$ ---  often just called a \emph{local Hamiltonian} ---  is of the form
\begin{equation} \label{eq:localhamiltonian}
  \H = \sum_{X \in \, \Eset} \H_X ,
\end{equation}
with $\supp(\H_X) = X$ for all $X \in \Eset$.
Most Hamiltonians in the condensed-matter context or of cold atoms in optical lattices can be very well approximated by such locally interacting Hamiltonians.
Generalising this notation to subsystems $X \subset \Vset$ that are not in $\Eset$ we define for any subsystem $X \subset \Vset$ the \emph{restricted Hamiltonian}
\begin{equation} \label{eq:restrictedhamiltonian}
  \H_X \coloneqq \sum_{Y \in \Eset\colon Y\subseteq X} \H_Y  \in \Obs(\mcH) ,
\end{equation}
which obviously fulfils $\supp(\H_X) \subseteq X$.
Note that we adopt the convention that $\H_X$ is an element of $\Obs(\mcH)$ and not of $\Obs(\mcH_X)$.

We will also need the \emph{graph distance}.
In order to define it, we first need to give a precise meaning to a couple of intuitive terms:
We say that two subsystems $X,Y \subset \Vset$ \emph{overlap} if $X \intersection Y \neq \emptyset$, 
a set $X \subset \Vset$ and a set $F\subset \Eset$ \emph{overlap} if $F$ contains an edge that overlaps with $X$, and two sets $F,F'\subset \Eset$ \emph{overlap} if $F$ overlaps with any of the edges in $F'$.
A subset $F \subset \Eset$ of the edge set \emph{connects} $X$ and $Y$ if $F$ contains all elements of some sequence of pairwise overlapping edges such that the first overlaps with $X$ and the last overlaps with $Y$ and similarly for sites 
$x,y \in \Vset$.

The \emph{(graph) distance} $\dist(X,Y)$ of two subsets $X,Y\subset \Vset$ with respect to the (hyper)graph $(\Vset,\Eset)$ is zero if $X$ and $Y$ overlap and otherwise equal to the size of the smallest subset of $\Eset$ that connects $X$ and $Y$.
The \emph{diameter} of a set $F \subset \Eset$ is the largest graph distance between any two sets $X,Y \in F$.
We extend the definition of the graph distance to operators $A,B \in \Bop(\mcH)$ 
and set $\dist(A,B) \coloneqq \dist(\supp(A), \supp(B))$.

We will also make use of the notion of \emph{reduced states}, or \emph{marginals}.
Given a quantum state $\rho \in \Qst(\mcH)$ of a composite system with subsystem $X \subset \Vset$ we write $\rho^X$ for the \emph{reduced state} on $X$, which is defined as the unique quantum state $\rho^X \in \Qst(\mcH_X)$ with the property that for any observable $A \in \Obs(\mcH)$ with $\supp(A) \subseteq X$
\begin{equation} \label{eq:reducedstate}
  \Tr(\trunc A X \,\rho^X) = \Tr(A\,\rho) .
\end{equation}
Defining the reduced state in systems of fermions in this way is important to avoid ambiguities \cite{Friis2013}.
We will denote the linear map $\rho \mapsto \rho^X$ by $\Tr_{\compl{X}}$.
As $\Tr_{\compl{X}}$ is linear we can naturally extend its domain to all of $\Bop(\mcH)$ so that
\begin{equation} \label{eq:partialtrace}
  \Tr_{\compl{X}} \colon \Bop(\mcH) \to \Bop(\mcH_X) .
\end{equation}
In the case of spin systems $\Tr_{\compl{X}}$ is indeed the \emph{partial trace} over $\compl X = \Vset\setminus X$ as defined for example in Ref.~\cite{nielsenchuang}.
For time evolutions $\rho \colon \R \to \Qst(\mcH)$ we use the natural generalisation of the superscript notation, i.e., $\rho^X = \Tr_{\compl{X}} \circ \mathop\rho \colon \R \to \Qst(\mcH_X)$.

Correlations play a central role in the description of composite systems and hence in condensed matter physics and statistical mechanics.
It is beyond the scope of this work to give a comprehensive overview of the different types and measures of correlations (see for example Refs.~\cite{Kastoryano2011,Kastoryano2013,nielsenchuang,Plenio07}).
One important measure of correlation is the \emph{covariance}, which for a quantum state $\rho \in \Qst(\mcH)$ and two operators $A,B \in \Bop(\mcH)$ is defined to be
\begin{equation} \label{eq:covariance}
 \cov_\rho(A,B) \coloneqq \Tr(\rho\,A\,B) - \Tr(\rho\, A) \Tr(\rho \, B) .
\end{equation}
It satisfies
\begin{equation}
 |\cov_\rho(A,B)| \leq \left({\ex{A^2}\rho \, \ex{B^2}\rho} \right)^{1/2}
\end{equation}
and hence one often defines the \emph{correlation coefficient} as $\cov_\rho(A,B) / ({\ex{A^2}\rho \, \ex{B^2}\rho})^{1/2}$.
We will encounter a slightly generalised version of the covariance in Section~\ref{sec:propertiesofthermalstatesofcompositesystems}.

The covariance is most interesting as a correlation measure if $A$ and $B$ act on disjoint subsystems, i.e., $\supp(A) \intersection \supp(B) = \emptyset$.
If for a given state $\rho \in \Qst(\mcH)$ of a bipartite system with $\Vset = X \dunion Y$ and any two observables $A,B \in \Obs(\mcH)$ with $\supp (A) \subseteq X$ and $\supp (B)  \subseteq Y$ it holds that $\cov_\rho(A,B) = 0$, then we say that $\rho$ is \emph{uncorrelated} with respect to the bipartition $\Vset = X \dunion Y$.

Uncorrelated states of spin systems are \emph{product states}.
Consider a bipartite spin system with Hilbert space $\H$ and vertex set $\Vset = X \dunion Y$.
A quantum state $\rho \in \Qst(\mcH)$ is said to be \emph{product} with respect to this bipartition if $\rho = \rho^X \otimes \rho^Y$.
We call a basis that consists entirely of product states a \emph{product basis}.

Still in the setting of a bipartite spin system with Hilbert space $\H$ and vertex set $\Vset = X \dunion Y$, all quantum states of the form
\begin{equation}\label{Separable}
  \rho = \sum_j p_j\, \rho_j^X \otimes \rho_j^Y
\end{equation}
with $(p_j)_j$ a \emph{probability vector}, i.e., $\sum_j p_j = 1$ and $p_j \geq 0$ for all $j$, and $\rho_j^X \in \Qst(\mcH_X)$ and $\rho_j^Y \in \Qst(\mcH_Y)$ for all $j$, are called \emph{separable} with respect to the bipartition $\Vset = X \dunion Y$.
All states that can be prepared with \emph{local operations and classical communication} (LOCC) are called \emph{separable}, a notion that also holds
true for bosonic or fermionic systems.
Such states are correlated in general, but a classical mechanism can be held responsible for the correlations present.
All states that are not separable are called \emph{entangled}.

The \emph{Gibbs state} or \emph{thermal state} of a system with Hilbert space $\mcH$ and Hamiltonian $\H \in \Obs(\mcH)$ at inverse temperature $\beta \in \R$ is defined as
\begin{equation} \label{eq:defthermalstate}
  \rhog[\H](\beta) \coloneqq \frac{\e^{-\beta\,\H}}{Z[\H](\beta)} \in \Qst(\mcH) ,
\end{equation}
where $Z[\H]$ is the \emph{(canonical) partition function} defined as
\begin{equation}
  Z[\H](\beta) \coloneqq \Tr(\e^{-\beta\,\H}) .
\end{equation}
The Gibbs state has the important property that it is the unique quantum state that maximises the von Neumann entropy 
\begin{equation} \label{eq:vonneumannentropy}
  \Svn(\rho) \coloneqq - \Tr(\rho \log_2 \rho ) .
\end{equation}
given the expectation value of the Hamiltonian \cite{thirringquantu}.
This is a direct consequence of Schur's lemma \cite{bhatia} and the fact that the same statement holds in classical statistical mechanics, as can be seen from a straight forward application of the Lagrange multiplier technique.
In fact, the inverse temperature $\beta$ is nothing but the Lagrange parameter associated with the energy expectation value.

For locally interacting quantum systems with a Hamiltonian $\H \in \Obs(\mcH)$ of the form given in \eqref{eq:localhamiltonian} 
we adopt the convention that for any subsystem $X \subset \Vset$
\begin{equation}
  \rhog^X[\H](\beta) = \Tr_{\compl{X}}(\rhog[\H](\beta)) \in \Qst(\mcH_X)
\end{equation}
denotes the reduction of the Gibbs state of the full system to the subsystem $X$ (compare \texteqref{eq:reducedstate}).% , while we write
% \begin{equation}
%   \rhog_X[\H](\beta) \coloneqq \Tr_{\compl{X}}(\rhog[\H_X](\beta)) = \rhog[\trunc{\H_X} X](\beta) \in \Qst(\mcH_X) 
% \end{equation}
% for the reduced state on $X$ of the Gibbs state of the restricted Hamiltonian $\H_X$, or equivalently the Gibbs state of $\trunc{\H_X} X$  (compare \texteqref{eq:restrictedhamiltonian}).

The micro-canonical ensemble in quantum statistical mechanics takes the form of the \emph{micro-canonical state}.
Usually one defines the micro-canonical ensemble and state with respect to an energy interval $[E,E+\Delta]$.
Here we make the slightly more general definition that will be useful later:
The micro-canonical state to any subset $R \subseteq \R$ of the real numbers of a system with Hilbert space $\mcH$ and Hamiltonian $\H \in \Obs(\mcH)$ with spectral decomposition $\H = \sum_{k=1}^{d'} E_k \Pi_k$ is defined as
\begin{equation}
  \rhomc[\H](R) \coloneqq \frac{\sum_{k:E_k \in R} \Pi_k}{Z_{mc}[\H](R)} \in \Qst(\mcH) ,
\end{equation}
where $Z_{mc}[\H]$ is the \emph{micro-canonical partition function} defined as
\begin{equation}
  Z_{mc}[\H](R) \coloneqq \Tr(\sum_{k:E_k \in R} \Pi_k) .
\end{equation}

\section{Equilibration}
\label{sec:equilibration}
The dynamics of finite dimensional quantum system, as described in the previous section, is recurrent \cite{JChemPhys843,Schulman1978,PhysRevLett.49,PhysRev.107.33,Wallace2013} and time reversal invariant.
Hence, genuine equilibration in the sense of Boltzmann's \emph{H-Theorem} \cite{RevModPhys.27.289} that implies that entropy can only grow over time (see also Section~\ref{sec:boltzmannshtheorem}) is impossible.
This apparent contradiction between the microscopic theory of quantum mechanics and the thermodynamic behaviour observed in nature is one of the main issues that any derivation of statistical mechanics and thermodynamics from quantum theory needs to solve.

We will see in this section that the unitary time evolution of pure states of such systems does imply in a surprisingly general and natural way that certain time dependent properties of quantum systems do dynamically equilibrate and that hence this apparent contradiction can be resolved to a large extend.

We will concentrate on two notions of equilibration: \emph{equilibration on average} and \emph{equilibration during intervals}.
After an introduction of these two notions in Section~\ref{sec:notionsofequilibration} we will discuss them in detail in Sections~\ref{sec:equlibrationintheweaksense} and \ref{sec:equlibrationinthestrongsense}.
In particular we will give conditions under which equilibration in the respective sense can be ensured.
In Section~\ref{sec:othernotionsofequilibration} we touch upon other notions of equilibration that have been investigated in the literature.
Then we discuss Lieb-Robinson bounds, which limit the signal propagation in locally interacting quantum lattice systems, in Section~\ref{sec:liebrobinson} before we go on to survey results on the times scale on which equilibration happens in Section~\ref{sec:timescales}.
We end this section with a brief description of \emph{fidelity decay} in Section~\ref{sec:fidelitydecay}.
In the next section, Section~\ref{sec:numericalandanalyticalinvesitgationsofequilibration}, we then put the discussed rigorous results into the perspective of the picture emerging from numerical simulations and the insights gained from analytic investigations of more specific models.

\subsection{Notions of equilibration}
\label{sec:notionsofequilibration}
In this section we define and compare two notions of equilibration compatible with the recurrent and time reversal invariant nature of unitary quantum dynamics in finite dimensional systems.
These notions will capture the intuition that equilibration means that a quantity, after having been initialised at a non-equilibrium value, evolves towards some value and then stays close to it for an extended amount of time.
At the same time, what we will call \emph{equilibration} is less than what one usually associates with the evolution towards \emph{thermal equilibrium}.
We will define a quantum version of the latter, call it \emph{thermalisation}, and discuss it in detail in Section~\ref{sec:thermalisation}.

To keep the definition of equilibration as general as possible we will refer abstractly to \emph{time dependent properties} of quantum systems, by which we mean functions $f\oftype \R\to M$ that map time to some metric space $M$, for example $\R$ or $\mcS(\mcH)$.
The metric will allow us to quantify how close the value of such functions is for different times and in particular how close it is to the time average and ``equilibrium values'' of the function.

Properties that we will be interested in include for example the time evolution of expectation values of individual observables.
We will also encounter \emph{subsystem equilibration}.
In this case the property is the time evolution of the state of the subsystem and the metric the trace distance.
It will also be convenient to speak more generally of the \emph{apparent equilibration of the whole system} with the metric then being the distinguishability under a restricted set of POVMs.

We will discuss the following two notions of equilibration in more detail:
\begin{description}%[font=\normalfont\itshape]
\item[Equilibration on average:]
  We say that a time dependent property \emph{equilibrates on average} if its value is for \emph{most times during the evolution} close to some \emph{equilibrium value}.
\item[Equilibration during intervals:]
We say that a time dependent property \emph{equilibrates during a (time) interval} if its value is close to some \emph{equilibrium value} for \emph{all times in that interval}.
\end{description}

The use of the notion of equilibration on average in the quantum setting goes back to at least the work of von Neumann \cite{vonneumann1929} and has recently been developed further, in particular in Refs.~\cite{tasaki98,Reimann08,Linden09,1110.5759v1,1012.4622v1,Reimann2012,Reimann12,Huebener2015}.
We will see that equilibration on average, especially for expectation values of observables as well as for reduced states of small subsystems of large quantum systems, is provably a
very generic feature.
In contrast, equilibration during intervals is a property that is expected to be generically the case for locally interacting many-body systems, and there is compelling numerical
evidence for such a behaviour.
To date, however,
it has rigorously been proven only for specific models \cite{cramer10_1,PhysRevLett.10-5}.

Equilibration on average implies that the equilibrating property spends most of the time during the evolution close to its time average.
This allows for a reasonable definition of an \emph{equilibrium state}, which is then the time averaged or de-phased state.
As we will see later in Section~\ref{sec:thermalisation}, this makes it possible to tackle the question of \emph{thermalisation} in unitarily evolving quantum systems.

On the down side, a proof of equilibration on average alone does not immediately imply much about the time scale on which the equilibrium value is reached after a system is started in an out of equilibrium situation.
We will see that even though it is possible to bound these time scales, the bounds obtainable in the general settings considered here are only of very limited physical relevance (see Section~\ref{sec:timescales}).

As we will see in the following, the statements on equilibration during intervals are much more powerful in this respect.
They imply bounds on the time it takes to equilibrate that scale reasonably with the size of the system and hence prove equilibration on experimentally relevant time scales.
On the other hand, in the few settings in which equilibration during intervals of reduced states of subsystems has been proven, it is known that the equilibrium states are not close to thermal states of suitably restricted Hamiltonians.
In particulate, no proof of \emph{thermalisation} (in the sense of the word we will defined later in Section~\ref{sec:whatisthermalisation}) based on a result on equilibration during intervals is known to date.
We discuss both notions of equilibration in detail in the following two sections.

\subsection{Equilibration on average}
\label{sec:equlibrationintheweaksense}
In this section we discuss equilibration on average (see figure~\ref{fig:equilibrationonaverage} for a graphical illustration).
The outline is as follows:
After giving some historic perspective we will go through the main ingredients that feature in the known results on equilibration on average and discuss their role in the arguments and to what extent they are physically reasonable and mathematically necessary.
After this preparation we will state, prove and interpret the arguably strongest result on equilibration on average known to date.

Already the founding fathers of quantum mechanics realised that the unitary evolution of large, closed quantum systems, together with the immensely high dimension of their Hilbert space and quantum mechanical uncertainty, could possibly explain the phenomenon of equilibration.
Most notable is an article of von Neumann \cite{vonneumann1929} from 1929, which already contains a lot of the ideas and even variants of some of the results that can be found in the modern literature on the subject.
The renewed interest in the topic of equilibration was to a large extent a consequence of the two independent theoretical works Refs.~\cite{Reimann08,Linden09}.
The approach outlined there was then more recently refined and the results gradually strengthened.
Important contributions are in particular Refs.~\cite{Reimann12,1110.5759v1,1012.4622v1}.
Also very noteworthy is the often overlooked earlier work Ref.~\cite{tasaki98}.

The first fact that plays a prominent role in the proofs of equilibration on average is the immensely high dimension of the Hilbert space of most many-body systems.
The dimension of the Hilbert space of composite systems grows exponentially with the number of constituents.
What actually matters, of course, is the \emph{number of significantly occupied energy levels}, rather than the number of levels that are in principle available but not populated.
For each $k \in [d']$ we define the occupation $p_k \coloneqq \Tr(\Pi_k\,\rho(0))$ of the $k$-th energy level, where $d' \coloneqq |\spec(\H)| \leq d = \dim(\mcH)$ is the number of distinct such levels.
Refs.~\cite{tasaki98,Reimann08} use $\max_k p_k$, the occupation of the most occupied level, to quantify the number of significantly occupied energy levels.
Ref.~\cite{Linden09} uses a quantity called \emph{effective dimension}, denoted by $\deff(\omega)$, which in our notation can be defined as
\begin{equation} \label{eq:effectivedimension}
  \deff(\omega) \coloneqq \frac{1}{\sum_{k=1}^{d'} p_k^2} \geq \frac{1}{\max_k p_k} .
\end{equation}
If the initial state is taken to be an energy eigenstate, the resulting effective dimension is one, while that resulting from a uniform coherent superposition of $\tilde{d}$ energy eigenstates to different energies is $\tilde{d}$.
This justifies the interpretation of $\deff(\omega)$ as a measure of the number of significantly occupied states.
It is also reciprocal to a quantity that is known mostly in the condensed matter literature as \emph{inverse participation ratio} \cite{Neuenhahn10} and related to the time average of the \emph{Loschmidt echo} \cite{Levstein98,Campos10}.
While using the effective dimension instead of the occupation of the most occupied level can lead to tighter bounds it has the disadvantage that it cannot be efficiently computed given a state and the Hamiltonian.

There are a number of different ways to argue why it is acceptable to restrict oneself to initial states that populate a large number of energy levels when trying to prove the emergence of thermodynamic behaviour from the unitary dynamics of closed systems.
First, one can argue that initial states that only occupy a small subspace of the Hilbert space of a large system behave essentially like small quantum systems and such systems are anyway not expected to behave thermodynamically, but rather show genuine quantum behaviour.
Second, one can invoke the inevitable limits to the resolution and precision of experimental equipment to conclude that preparing states that overlap only with a handful of the roughly $2^{10^{23}}$ energy levels of a macroscopic system is impossible, even if we had apparatuses that were many orders of magnitude more precise than the equipment available today \cite{Reimann08,Reimann12}.
Finally, one can also take a more mathematical point of view and use results based on a phenomenon called \emph{measure concentration} \cite{ledoux01,CHATTERJEE07} that guarantees that uniformly random pure states drawn from sufficiently large subspaces of a Hilbert space have, with extremely high probability, an effective dimension with respect to any fixed, sufficiently non-degenerate Hamiltonian that is comparable to the dimension of that subspace \cite{Linden09,Popescu06,Popescu05,Gogolin10-masterthesis} (more on such typicality arguments in Section~\ref{sec:typicality}).
If one is willing to assume that such states are physically natural initial states, this can justify the assumption of a large effective dimension.
We will come back to this in Section~\ref{sec:typicality} where we discuss \emph{typicality}.
For an earlier work that directly mingles typicality arguments and de-phasing to derive an equilibration result see also Ref.~\cite{Bocchieri1959}.

As we will see below, it is actually sufficient for equilibration that $\maxprime_k p_k$, the second largest of the energy level occupations, is small.
Note that in the physically relevant situation of a system that is cooled close to its ground state $\maxprime_k p_k$ can be orders of magnitude smaller than $\max_k p_k$ or $1/\deff(\omega)$.
Although the proof of this extension of previous results is not trivial \cite{Reimann12}, the physical intuition behind it is clear:
The expectation values of all observables of a system that is initialised in an energy eigenstate are already in equilibrium.
What can prevent equilibration on average are not macroscopic populations of one energy level, but rather initial states that are coherent superpositions of a small number of energy eigenstates.
Such states can show a behaviour reminiscent of Rabi-Oscillations and not exhibit equilibration.

\begin{figure}[bt]
  \centering
  \begin{tikzpicture}[scale=1.2]
    \draw[thick,->] (-1,0) -- (6,0) node[at end,below] {$E$} node[at start,left] {$\spec(\H):$};
    % \foreach \i/\x in {1/0.1,2/0.3,3/1,4/1.5,5/2.5,6/3.1,7/3.2,8/4.5,9/4.8} {\node (n\i) at (\x,-0.1) {}; \draw[thick] (\x,-0.1) -- (\x,0.1);};
    \foreach \i/\x in {1/0.1,2/0.3,3/1,4/1.5,5/2.5,6/3.1,7/3.2,8/4.5,9/4.8} {\draw[thick] (\x,-0.1) -- (\x,0.1) node[at start,scale=0.01] (n\i) {};};
    \draw[decorate,decoration={brace}] (n4) -- (n3) node[midway,below=0.5] (gap1) {} node[at start,above=0.3] (el) {$E_l$} node[at end,above=0.3] (ek) {$E_k$};
    \draw[decorate,decoration={brace}] (n8) -- (n7) node[midway,below=0.5] (gap2) {} node[at start,above=0.3] (en) {$E_n$} node[at end,above=0.3] (em) {$E_m$};
    \node[anchor=north] (neq) at ($ (gap1) !.5! (gap2) $) {$\neq$};
    \draw[->] (neq) -- (gap1);
    \draw[->] (neq) -- (gap2);
  \end{tikzpicture}          
  \caption{(Reproduction from Ref.~\cite{Gogolin2014}) Illustration of the non-degenerate energy gaps condition.
No gap between two energy levels may occur more than once in the spectrum, but the individual levels may well be degenerate.}
  \label{fig:nondegenerateenergygaps}
\end{figure}
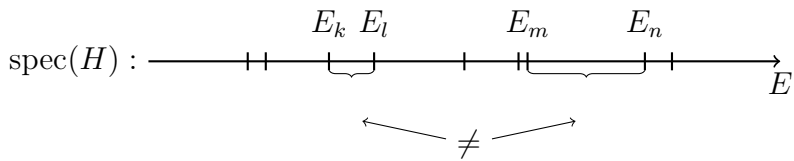

The second main ingredient to the proofs of Refs.~\cite{vonneumann1929,tasaki98,Reimann08,Linden09} is the condition of \emph{non-degenerate energy gaps} originally called the \emph{non-resonance} condition.
We say that a Hamiltonian $\H$ has \emph{non-degenerate energy gaps}, if for every $k,l,m,n \in [d']$
\begin{equation}
  E_k - E_l = E_m - E_n \implies (k = l \land m = n) \lor (k=m \land l=n) ,
\end{equation}
i.e., if every energy gap $E_k - E_l$ appears exactly once in the spectrum of $\H$ (see figure~\ref{fig:nondegenerateenergygaps}).
The original condition used in Refs.~\cite{vonneumann1929,tasaki98,Reimann08,Linden09} is stronger and excludes in addition all Hamiltonians with degeneracies, i.e., requires that $d' = d$.
Although the non-degenerate energy gaps condition appears to be pretty technical at first sight, the motivation for imposing it can be made apparent by the following consideration:
The main concern of Ref.~\cite{Linden09} is the equilibration on average of the reduced state $\rho^S(t)$ of a small subsystem $S$ of a bipartite system with $\Vset = S \dunion B$.
If the Hamiltonian of the composite system is of the form
\begin{equation} \label{eq:hamiltonianwithoutcouplingtodemonstratenondegenerategaps}
  \H = \H_S + H_B ,
\end{equation}
i.e., $S$ and $B$ are not coupled (remember the definition of the restricted Hamiltonian in \texteqref{eq:restrictedhamiltonian}), then $\rho^S(t)$ will simply evolve unitarily and equilibration of $\rho^S(t)$ is clearly impossible.
Hence, one needs a condition that excludes such non-interacting Hamiltonians.
Imposing the condition of non-degenerate energy gaps is a mathematically elegant, simple, and natural way to do this.
It is easy to see that Hamiltonians of the form given in \texteqref{eq:hamiltonianwithoutcouplingtodemonstratenondegenerategaps} have many degenerate gaps, as their eigenvalues are simply sums of the eigenvalues of $\H_S$ and $\H_B$.

In the more recent literature, the condition of non-degenerate energy gaps has been gradually weakened.
Ref.~\cite{1110.5759v1} defines the maximal number of energy gaps in any energy interval of width $\epsilon$,
\begin{equation}\label{eq:numeberofalmostdegenerategaps}
  N(\epsilon) \coloneqq \sup_{E \in \R} |\{(k,l) \in [d']^2\suchthat k\neq l \land E_k - E_l \in [E,E+\epsilon] \}| .
\end{equation}
Note that $N(0)$ is the number of degenerate energy gaps and a Hamiltonian $\H$ satisfies the \emph{non-degenerate energy gaps} condition if and only if $N(0) = 1$.
The above definition allows to prove an equilibration theorem that still works if a system has a small number of degenerate energy gaps.
Moreover, it has the advantage that it allows to make statements about the equilibration time.
As we will see in the next theorem, equilibration on average can be guaranteed to happen on a time scale $T$ that is large enough such that $T\,\epsilon \gg 1$ where $\epsilon$ must be chosen small enough such that $N(\epsilon)$ is small compared to the number of significantly populated energy levels.

The arguably strongest and most general result concerning equilibration on average in quantum systems can be obtained by combining the two recent works Refs.~\cite{1110.5759v1,Reimann12}.
In fact, we will see that it even goes slightly beyond a mere proof of equilibration on average, as it does have non-trivial implications for the time scales on which equilibration happens.

\begin{theorem}[Equilibration on average] \label{thm:equilibrationonaverage}
  Given a system with Hilbert space $\mcH$ and Hamiltonian $\H \in \Obs(\mcH)$ with spectral decomposition $\H = \sum_{k=1}^{d'} E_k\,\Pi_k$.
  For $\rho(0) \in \Qst(\mcH)$ the initial state of the system, let $\omega = \$_H(\rho(0))$ be the de-phased state and define the energy level occupations $p_k \coloneqq \Tr(\Pi_k\,\rho(0))$.
  Then, for every $\epsilon,T>0$ it holds that (i) for any operator $A \in \Bop(\mcH)$
  \begin{equation} \label{eq:equilibrationonaverageforexpectationvalues}
    \taverage[T]{( \ex A {\rho(t)} - \ex A \omega )^2} \leq \norm[\infty]{A}^2\,N(\epsilon)\,f(\epsilon\,T)\,g((p_k)_{k=1}^{d'})  ,
  \end{equation}
  and (ii) for every set $\POVMs$ of POVMs
  \begin{equation} \label{eq:equilibrationonaverageforrestricedpovms}
    \taverage[T]{\tracedistance[\POVMs]{\rho(t)}{\omega}} \leq h(\POVMs)\,
    \left({N(\epsilon)\,f(\epsilon\,T)\,g((p_k)_{k=1}^{d'}) } \right)^{1/2},
  \end{equation}
  where $N(\epsilon)$ is defined in \texteqref{eq:numeberofalmostdegenerategaps}, $f(\epsilon\,T) \coloneqq 1+8 \log_2(d')/(\epsilon\,T)$, 
  \begin{align}
    && g((p_k)_{k=1}^{d'}) &\coloneqq \min(\sum_{k=1}^{d'} p_k^2, 3  \maxprime_k p_k ) ,\label{eq:equilibrationonaveraggeneralisedeffectivedimenson}\\
    \text{and}&& h(\POVMs) &\coloneqq \min(|{\union \POVMs}|/4, \dim(\mcH_{\supp(\POVMs)})/2 ) \label{eq:equilibrationonaveragnumberofmeasurements},
  \end{align}
  with $\maxprime_k p_k$ the second largest element in $(p_k)_{k=1}^{d'}$, $\union \POVMs$ the set of all distinct POVM elements in $\POVMs$, and $\supp(\POVMs) \coloneqq \bigcup_{M \in \union \POVMs} \supp(M)$.
\end{theorem}

\begin{proof}
  \texteqref{eq:equilibrationonaverageforexpectationvalues} for $g((p_k)_{k=1}^{d'})$ equal to the first argument of the $\min$ in \texteqref{eq:equilibrationonaveraggeneralisedeffectivedimenson} is Theorem~1 in Ref.~\cite{1110.5759v1}.
  The same statement, but with $g((p_k)_{k=1}^{d'})$ equal to the second argument in the $\min$, follows from Eqs.~(44), (50), (61), and (63) in Ref.~\cite{Reimann12}.
  With $|{\union U}|$ in \texteqref{eq:equilibrationonaveragnumberofmeasurements} replaced by the total number of all measurement outcomes, i.e., $\sum_{M \in \POVMs} |M|$, \texteqref{eq:equilibrationonaverageforrestricedpovms}, for $g((p_k)_{k=1}^{d'})$ equal to the first argument of the $\min$ in \texteqref{eq:equilibrationonaveraggeneralisedeffectivedimenson}, is implied by Theorems~2 and 3 from Ref.~\cite{1110.5759v1}.
  A careful inspection of Eq.~(B.1) in Ref.~\cite{1110.5759v1}, however, reveals that the slightly stronger result holds.
  In particular, one can first use the bound 
  \begin{equation}
  \max_{M(t) \in \mathcal{M}} D_{M(t)}(\rho(t),\omega) \leq \sum_{M_a \in \union \mathcal{M}} |\mathop{tr}(M_a\,\rho(t)) - \mathop{tr}(M_a\,\omega)| 
 \end{equation}
 for the argument of the time average in the right hand side of the first line of Eq.~(B.1) and then use the triangle inequality to pull the time average into the sum.
  For $g((p_k)_{k=1}^{d'})$ equal to the second argument the result follows using \texteqref{eq:equilibrationonaverageforexpectationvalues} instead of Theorem~1 from Ref.~\cite{1110.5759v1} in the proofs of Theorems~2 and 3 from Ref.~\cite{1110.5759v1}.
\end{proof}

\begin{figure}[bt]
  \centering
  \begin{tikzpicture}[scale=0.7]
    \draw[-,thick] (-0.5,0) -- (7,0);
    \draw[-,dotted,thick] (7,0) -- (8,0) ;
    \draw[-,dotted,thick] (9,0) -- (10,0) ;
    \draw[->,thick] (10,0) -- (14,0) node[near end,below] {$t$} ;
    \draw[->,thick] (-0.5,0) -- (-0.5,5) node[midway,above,rotate=90] {$( \ex A {\rho(t)} - \ex A \omega )^2$} ;
    \draw[color=structure] plot file {distanceevolution.dat};
    \draw[color=structure] plot file {distanceevolution2.dat};
    \draw[-,dotted,thick] (0,0) -- (0,5) node[at start, below] {$0$};
  \end{tikzpicture}
  \caption{(Reproduced from Ref.~\cite{Gogolin2014}) Equilibration on average is compatible with the time reversal invariant and recurrent nature of the time evolution of finite dimensional quantum systems.
The figure shows a prototypical example of equilibration on average.
Started in a non-equilibrium initial condition at time $0$ the expectation value of some observable $A$ quickly relaxes towards the equilibrium value $\ex A \omega$ and then fluctuates around it, with far excursions from equilibrium being rare.
After very long times the system returns (close to) its initial state and so does the expectation value of the observable.
A similar behaviour is observed when the initial state is evolved backwards in time.}
  \label{fig:equilibrationonaverage}
\end{figure}
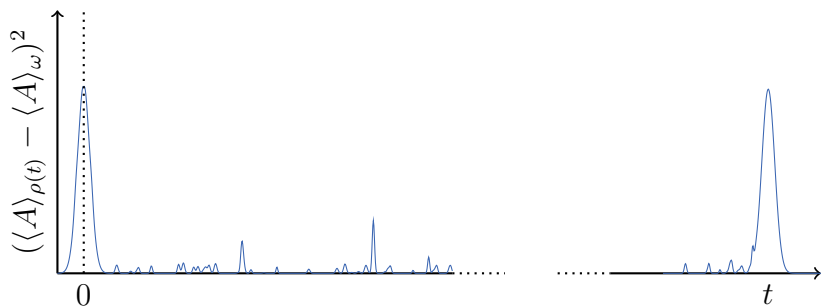

What is the physical meaning of the theorem?
The quantity $g((p_k)_{k=1}^{d'})$ is small, except if the initial state assigns large populations to few (but more than one) energy levels.
For initial states with a reasonable energy uncertainty and large enough systems it can be expected to be of the order of $\landauO(1/d')$, i.e., reciprocal to the total number of distinct energy levels.
The quantity $h(\POVMs)$ on the other hand can be thought of as a measure of the experimental capabilities in distinguishing quantum states and can reasonably be assumed to be much smaller than $d'$.
In particular, when all measurements in $\POVMs$ have a support contained inside of a small subsystem $S \subset \Vset$ it is bounded by $d_S/2$.
Because of the conditions for equality in \texteqref{eq:tracedistanceboundsrestrictedtracedistance}, the theorem then also implies an upper bound on $\taverage[T]{\tracedistance{\rho^S(t)}{\omega^S}}$ and hence proves \emph{subsystem equilibration on average}.

For fixed $\H$ and $\epsilon>0$ we have $\lim_{T\to\infty} f(\epsilon\,T) = 1$, hence the theorem proves, for a wide class of reasonable initial states, equilibration on average of all sufficiently small subsystems and \emph{apparent equilibration on average} of the state of the full system under realistic restrictions on the number of different measurements that can be performed.
In this sense it improves and generalises the results of Refs.~\cite{Reimann08,Linden09}.

On what time scales is equilibrium reached?
The product $N(\epsilon)\,f(\epsilon\,T)$, which is lower bounded by one, will typically be close to one only if $T$ is comparable to ${d'}^2$, i.e., to the total number of energy gaps, and will otherwise be roughly of the order of $\landauOmega({d'}^2/T)$ for smaller $T$.
So, even under the favorable assumption that $g((p_k)_{k=1}^{d'})$ is of the order of $\landauO(1/d')$, equilibration of a subsystem $S$ can only be guaranteed after a time $T$ that is roughly of the order of $\landauOmega(d_S^2\,d')$.
Both $d'$ and $d_S$ typically grow exponentially with the size of the composite system and the subsystem $S$, respectively.
Hence, times of the order of $\landauOmega(d_S^2\,d')$ are unphysical already for systems of moderate size.
This weakness of theorems such as Theorem~\ref{thm:equilibrationonaverage} has been criticised in Ref.~\cite{1109.4696v1} (see Section~\ref{sec:timescales} for more details on equilibration times).

There are at least two possible replies to this criticism:
First, it is known that there are systems in which equilibration does indeed take extremely long (see Section~\ref{sec:timescales}) and thus, being a very general statement, Theorem~\ref{thm:equilibrationonaverage} is probably close to optimal.
Proofs of shorter equilibration times will need further assumptions, such as locality or translation invariance of the Hamiltonian, and restrictions on the allowed measurements \cite{1110.5759v1,Linden09}.
Second, almost all systems in which equilibration has been studied and in which equilibration of some property on reasonable time scales could be demonstrated were found to exhibit equilibration towards the time average (see for example Refs.~\cite{Gemmer09,Campos10,Fagotti2012,Rigol11,1110.4690v1,Rigol07,Rigol2006,Zhuang2013}), so in these cases the upper bound on the equilibration time implied by Theorem~\ref{thm:equilibrationonaverage} is not tight, but the theorem still captures the relevant physics.
Transient equilibration to metastable states that precedes equilibration to the time average seems to require special structure in the Hamiltonian.
That the physics of such special systems is not captured by a result as general as Theorem~\ref{thm:equilibrationonaverage} is not too surprising.

An interesting variant of the subsystem equilibration setting is investigated in Ref.~\cite{MasterThesisHutter}, in which the subsystem $S$ can initially be correlated (either classically or even quantum mechanically) with a reference system $R$.
The ``knowledge'' about the initial state of $S$ stored in the reference $R$ can in principle help to distinguish the state $\rho^S(t)$ from $\omega^S$.
Still, by using \emph{decoupling theorems} \cite{Dupuis2010,1109.4348v1,Szehr2012} and properties of \emph{smooth min and max entropies} \cite{Koenig08,Ciganovic2013} it is possible to show subsystem equilibration on average under conditions similar to those of Theorem~\ref{thm:equilibrationonaverage}, in the sense that the combined state of $S$ and $R$ is on average almost indistinguishable from $\omega^{SR} = \taverage{\rho^{SR}}$.

In the above disquisition on equilibration we have put a focus on the more recent literature, however, many of the ideas behind the results mentioned above can already be found in the work of von Neumann \cite{vonneumann1929}.
We encourage the interested reader to consider the English translation \cite{Tumulka2010} of this article and the discussion of von Neumann's results in Ref.~\cite{0907.0108v1} and the brief summary of parts of this article in Section~\ref{sec:typicality} of this work.

Further statements concerning equilibration towards the de-phased state, which are related to those discussed above, can also be found in Refs.~\cite{1108.2985v3,1108.0374,1112.5295v1,1107.6035v1}.
We will discuss these works in more detail in Section~\ref{sec:timescales}.

\subsection{Equilibration during intervals}
\label{sec:equlibrationinthestrongsense}
In this section we investigate under which conditions equilibration during intervals can be guaranteed.
After a brief overview of the literature on the topic we will concentrate on the results presented in Ref.~\cite{cramer10_1}.
Instead of reproducing the full proof we will only give the intuition behind it and describe the general structure.
One reason for this is that Ref.~\cite{cramer10_1} is concerned with a special class of bosonic Hamiltonians, so-called \emph{quadratic bosonic Hamiltonians}, i.e., Hamiltonians that are quadratic polynomials in the bosonic creation and annihilation operators.
For these Hamiltonians there exists a special formalism based on so-called \emph{covariance matrices} that allows, for example, to calculate for a special class of initial states, namely \emph{Gaussian states}, the time evolution of the expectation values of certain observables in a computationally efficient way.
A full introduction of this formalism is beyond the scope of this review.
More details can be found for example in Refs.~\cite{Eisert03,Braunstein2005,Adesso2007,GaussianReview}.

Equilibration during intervals of non-Gaussian initial states under certain quadratic Hamiltonians has been proven in Ref.~\cite{PhysRevLett.10-5} and the results have later been generalised and improved in Ref.~\cite{cramer10_1}.
The techniques are inspired by earlier works \cite{Dudnikova2003} on classical harmonic crystals, i.e., systems of coupled classical harmonic oscillators, and can be seen as bounds on the pre-asymptotic behaviour and an extension to finite system sizes of the results on equilibration of Ref.~\cite{Lanford1972}.
See also Refs.~\cite{Lanford1972,Tegmark1994,Barthel08} for related results on equilibration starting from Gaussian initial states.

More precisely, the results on equilibration during intervals of Ref.~\cite{cramer10_1} concern systems evolving under certain quadratic Hamiltonians of the form
\begin{equation} \label{eq:quadratichamiltonian}
  \H = \frac{1}{2}\,\sum_{x,y\in\Vset} 
  \left( b_x^\dagger\,K_{x,y}\,b_y + b_x\,K_{x,y}\,b_y^\dagger \right) ,
\end{equation}
where $b_x,b_x^\dagger$ are the bosonic annihilation/creation operators on site $x \in \Vset$ and $K \in \R^{|\Vset|\times|\Vset|}$.
The operator $\H$, as defined in \texteqref{eq:quadratichamiltonian}, is unbounded and hence, in principle, a careful treatment of the system with the methods of functional analysis \cite{Reed1980,Reed1975} would be necessary.
The Hamiltonian in \texteqref{eq:quadratichamiltonian} is, however, particle number preserving.
Thus, when we restrict to initial states with finite particle number the whole evolution happens in a finite dimensional subspace of the Fock space.
The Hamiltonian $\H$ and all relevant observables can then be represented by bounded operators on this subspace.
We are hence back in the framework of finite dimensional quantum mechanics as introduced in Section~\ref{sec:preliminaries} and the following statement is well-defined:

\begin{theorem}[Equilibration during intervals] \label{thm:equilibrationduringintervals}
  Consider the class of systems with a finite number of bosons in $M$ modes on a ring with nearest neighbour interactions, i.e., $\Vset = [M]$ and $\Eset=\{(1,2),\allowbreak(2,3),\dots,(M,1)\}$, evolving under a Hamiltonian of the form given in \eqref{eq:quadratichamiltonian} with $K_{x,y} = -\delta_{|x-y| \mod M,1}$.
  Let $\mcH$ be the direct sum of Fock layers up to the maximal particle number.
  If the initial state $\rho(0) \in \Qst(\mcH)$ satisfies a form of decay of correlations (Assumptions~1--3 in Ref.~\cite{cramer10_1}) and has time independent second moments (see Ref.~\cite{cramer10_1}), 
  then for every $S \subset \Vset$ and every $\epsilon>0$ there exists a system size $M^\ast$, such that for all $M \geq M^\ast$ there exists a time $t_{\mathrm{relax}}$ independent of $M$ and a time $t_{\mathrm{rec}} \in \landauOmega(M^{6/7})$ such that there exists a Gaussian state $\tilde\omega \in\Qst(\mcH)$ such that
  \begin{equation}
    \forall t \in [t_{\mathrm{relax}},t_{\mathrm{rec}}]\itholds \tracedistance{\rho^S(t)}{\tilde\omega^S} \leq \epsilon .
  \end{equation}
\end{theorem}
\begin{proof}
  The theorem is essentially implied by Theorem~2 and Corollary~1 from Ref.~\cite{cramer10_1}, as well as the discussion between them.
  The scaling of the times $t_{\mathrm{relax}}$ and $t_{\mathrm{rec}}$ follows from Eq.~(61) and Lemma~4 in Ref.~\cite{cramer10_1}.
\end{proof}
The theorem proves equilibration during the interval $[t_{\mathrm{relax}},t_{\mathrm{rec}}]$ of all small subsystems of a sufficiently large system.
It is key to this type of equilibration that the state $\tilde\omega$ is a Gaussian state, even if the system was initially prepared in a non-Gaussian state.
In fact, a similar convergence to a Gaussian state can be proven even in instances where the second moments are not constant in time.
Then it is still true that non-Gaussian states become locally Gaussian over time, but local expectation values will then not become stationary. 
Again, it is important to note that the class of Hamiltonians considered here is special ---  the Hamiltonians are quadratic in the bosonic operators --- 
but this does not apply to the initial states.
The technical requirements on the initial state allow, for example, for ground states of gapped interacting local Hamiltonians.
These conditions, precisely laid out in Ref.~\cite{cramer10_1}, ask for an algebraic decay of two- and four-point-functions, as well as an algebraic decay of correlations between Weyl operators belonging to distant regions.

The time $t_{\mathrm{relax}}$ depends on the size of the subsystem $S$ under consideration, but is independent of the size of the composite system.
It depends on the speed at which the Hamiltonian is able to transport correlations through the system and the length scale on which the correlations in the initial state decay.
The time $t_{\mathrm{rec}}$ is a lower bound on the recurrence time and is slightly smaller than the time it takes for a signal to travel around the ring of bosonic modes.

\subsection{Other notions of equilibration}
\label{sec:othernotionsofequilibration}
In this section we briefly cover two other notions of equilibration for closed quantum systems.
The first alternative notion of equilibration we want to discuss was proposed in Ref.~\cite{0907.0108v1} and further investigated in Ref.~\cite{Goldstein2014}.
This work is closely related to an article of von Neumann \cite{vonneumann1929}.
There, von Neumann postulates that on large systems only a set of so-called \emph{macroscopic observables} is accessible.
The macroscopic observables are required to commute, thus they divide the Hilbert space in subspaces, so-called \emph{phase cells}, each containing states that belong to the same sequence of eigenvalues for all the macroscopic observables (see also the more detailed discussion of Ref.~\cite{vonneumann1929} in Section~\ref{sec:typicality}).
If one of the phase cells is particularly large, Ref.~\cite{0907.0108v1} associates it with \emph{thermal equilibrium} and says that a system is in \emph{thermal equilibrium} if and only if its state is almost entirely contained in that cell.
Variants of the results from Ref.~\cite{vonneumann1929} can then be used to prove equilibration in this sense.

Reminding oneself that measurements of quantum systems are ultimately \emph{sampling experiments} opens up an entirely new vista on the problem of equilibration, which leads us the second alternative notion of equilibration.
Performing a measurement of an observable does neither provide the experimentalist with the measurement statistic nor does it yield the expectation value of the observable.
Both can only be approximately determined by repeatedly performing the same experiment many times.
How many repetitions are needed to distinguish whether the measurement statistic of a given observable is close or far from that predicted by equilibrium statistical mechanics?
Such questions have been posed and partially answered in the fields of \emph{sample complexity} \cite{Batu2001,Batu2000,Canonne2012} and \emph{state discrimination} \cite{Audenaert2012,Audenaert2012a}.
Using the complexity of the task of collecting information about a quantum system as a justification for a statistical description was recently proposed in Ref.~\cite{Ududec2012}, which defines the concept of \emph{information theoretic equilibration}.
Essentially the authors of Ref.~\cite{Ududec2012} are able to show that with the use of very fine grained observables pure quantum states are practically indistinguishable from states corresponding to statistical ensembles.

\subsection{Lieb-Robinson bounds}
\label{sec:liebrobinson}

An important tool for the study of equilibration phenomena is provided by \emph{Lieb-Robinson bounds} \cite{Lieb1972,1004.2086v1,Hastings2006}.
They limit the speed at which excitations can travel through a quantum lattice system equipped with a locally interacting Hamiltonian.
They can be viewed as an upper bound on group velocity of any excitation.
In systems satisfying a Lieb-Robinson bound, information propagation is essentially contained within a causal cone, reminiscent of a  ``light cone'' or ``sound cone'' (see Figure~\ref{fig:liebrobinsoncones}).
Any excitations spreading faster than a maximum velocity are exponentially suppressed in the distance.
Such bounds make rigorous the expectation that no instantaneous information propagation should be possible in quantum lattice models, and thereby immediately provide lower bounds to equilibration times for such models.
The implications of Lieb-Robinson bounds to entanglement dynamics will be discussed in Section~\ref{sec:entanglement}.

Concretely, Lieb-Robinson bounds are statements of the following type:
\begin{theorem}[Lieb-Robinson bound (corollary of Theorem~1 from \cite{Kliesch2013})]
  Consider a locally interacting fermionic or spin system with Hilbert space $\mcH$ and Hamiltonian $\H \in \Obs(\mcH)$.
Let $A,B \in \Obs(\mcH)$ be observables and denote $B(t) \coloneqq \e^{-\iu\,\H\,t}\, B\, \e^{\iu\,\H\,t}$.
Then
  \begin{equation}
    \norm[\infty]{[A,B(t)]} \leq C \norm[\infty]A \norm[\infty]B \e^{v\,|t|-\dist(A,B)}
  \end{equation}
  where the \emph{Lieb-Robinson speed} $v$ depends only on the operator norm of the local terms of $\H$ and the coordination number $\max_{X \in \Eset} |\{Y \in \Eset \oftype X \intersection Y \neq \emptyset \}|$ of the interaction graph, and $C$ is a constant that depends only on $\min(|\supp(A)|,|\supp(B)|)$.
\end{theorem}
The theorem says that the commutator $[A,B(t)]$ is exponentially suppressed with the distance between the support of $A$ and a ``light-cone'' that grows with the time the observable $B$ is evolved under the Hamiltonian $\H$.
As $A$ could be one of the local terms of $\H$ this in particular implies that the distant terms of the Hamiltonian do not significantly influence the time evolution of $B$ and that for any time $t$ the operator $B(t)$ can be approximated by an observable with support only slightly larger than the base of the ``light-cone'' at that time.

Such ``light-cone''-like dynamics has been systematically explored and put into the context of equilibration analytically and numerically \cite{CalabreseCardy06,DeChiara2005,CalabreseQuenchEntanglement,PhysRevLett.10-5,cramer10_1,Eisert06,Kollath08,Geiger2013,Hauke2013,Mazza2013} as well as experimentally \cite{1111.0776v1,Jurcevic2014,Langen2013,Richerme2014a,Ronzheimer2013}.
Similar bounds also exist for more general settings, like local Liouvillian dynamics \cite{Nachtergaele11,Poulin10,Kliesch2013}, exponentially decaying but no longer strictly local interactions \cite{Hastings2006}, as well as for certain long-ranged, i.e., power law like decaying, interactions \cite{Hastings2006,EisertKastner,Hauke2013} as long as the exponent is sufficiently large. Such long-ranged
interactions have been experimentally investigated in systems of trapped ions
\cite{Jurcevic2014,Richerme2014a}.

Lieb-Robinson bounds can also be proven for certain systems with Hamiltonians with local terms with unbounded operator norm \cite{1010.4576v1,Nachtergaele09,HarmonicLiebRobinson}.
For example, for quadratic bosonic systems with Hamiltonians of the form
\begin{equation}\label{eq:generalquadratichamiltonian}
  \H=\frac{1}{2}\sum_{x,y\in \Vset}\left(b_x\ad \, K_{x,y}\, b_y+b_x \, K_{x,y}\, b_y\ad
    +b_x \, L_{x,y}\, b_y+b_x\ad \, L_{x,y}\, b_y\ad\right), \qquad K,L \in \R^{|\Vset|\times|\Vset|}
\end{equation}
where $b_x,b_x^\dagger$ are again the bosonic annihilation/creation operators on site $x \in \Vset$ a Lieb-Robinson bound holds.
Writing
\begin{align}
  \label{ident}
  K_{x,y}&=\frac{Q_{x,y}+P_{x,y}}{2} & L_{x,y}&=\frac{Q_{x,y}-P_{x,y}}{2},
\end{align}
such Hamiltonians can be cast into a form reflecting couplings between canonical positions $q_x \coloneqq (b_x + b\ad_x)/\sqrt{2}$ and momenta $p_x \coloneqq \iu (b\ad_x - b_x)/\sqrt{2}$
\begin{equation}
  \label{Hamiltonian}
  \H=\frac{1}{2}\sum_{x,y\in  \Vset}\left(q_x \, Q_{x,y} \, {q}_y+ {p}_x \, P_{x,y}\,  {p}_y \right), \qquad Q,P\in \R^{|\Vset|\times|\Vset|}.
\end{equation}
In this setting, local means that
$K_{x,y} = L_{x,y}= Q_{x,y}=P_{x,y}=0$  for $ \dist(x,y)>R$
for some $R\in \N$.
We write 
$d_{x,y} \coloneqq \dist(x,y)/R$ and define $ \tau \coloneqq \max\{{\norm[\infty]{P\,Q}}^{1/2},{\norm[\infty]{Q\,P}}^{1/2}\}\,|t|$, 
then the following Lieb-Robinson bound is valid:
\begin{theorem}[Lieb-Robinson bounds for quadratic bosonic systems \cite{HarmonicLiebRobinson}]
  Consider a Hamiltonian of the form given in \texteqref{Hamiltonian} then
\begin{align}
  \begin{sesac}
    &\frac{\sqrt{\norm[\infty]{PQ}}}{\norm[\infty]{P}} \norm[\infty]{[{q}_x(t),{q}_y]} \\
    &\frac{\sqrt{\norm[\infty]{PQ}}}{\norm[\infty]{Q}} \norm[\infty]{[{p}_x(t),{p}_y]}    
  \end{sesac}
  &\le
  \frac{\tau^{d_{x,y}+2}\cosh\left(\tau\right)}{d_{x,y}!},\\
  \intertext{and}
  \begin{sesac}
    &\norm[\infty]{[{q}_x(t),{p}_y]} \\
    &\norm[\infty]{[{p}_x(t),{q}_y]}    
  \end{sesac}
  &\le
  \frac{\tau^{d_{x,y}}\cosh\left(\tau\right)}{d_{x,y}!} .
\end{align}
\end{theorem}
That is, for sufficiently large $\dist(x,y)$, one finds a faster-than-exponential decay of commutators between the canonical position and momentum operators.
This gives rise to a ``light cone'' with the Lieb-Robinson velocity
\begin{equation}\label{me2}	
  v=\e\,R\max\{\norm[\infty]{Q\,P}^{1/2},\norm[\infty]{P\,Q}^{1/2}\}.
\end{equation}

Despite the results of Refs.~\cite{1010.4576v1,Nachtergaele09,HarmonicLiebRobinson}, a full proof of a Lieb-Robinson bounds for a natural, interacting, infinite dimensional model, such as the Bose-Hubbard model with finite filling, is to date still missing.

\begin{figure}[tb]
  \centering
  \begin{tikzpicture}[domain=-2:2,samples=41,scale=1.5]
    \begin{scope}[xshift=-2.5cm]
      \draw[thick,->] (-2,0) -- (2,0) node[at end,below] {};
      \draw[thick,->] (0,0) -- (0,2) node[at end,left] {$t$} node[at start,below=0.7cm] {(a)};
      \clip (-2,0) rectangle (2,1.8);
      \fill[niceblue,opacity=0.5] plot (\x,{abs(\x)}) -- cycle;
      \draw[thick] plot (\x,{abs(\x)});
    \end{scope}
    \begin{scope}[xshift=2.5cm]
      \draw[thick,->] (-2,0) -- (2,0) node[at end,below] {};
      \draw[thick,->] (0,0) -- (0,2) node[at end,left] {$t$} node[at start,below=0.7cm] {(b)};
      \clip (-2,0) rectangle (2,1.8);
      \fill[niceblue,opacity=0.5] plot (\x,{exp(abs(\x))-1}) -- cycle;
      \draw[thick] plot (\x,{exp(abs(\x))-1});
    \end{scope}
  \end{tikzpicture}
  \caption{\label{fig:liebrobinsoncones} Schematic depiction of the Lieb-Robinson ``light'' cones in clean systems (a) and the more stringent bounds that can be derived in disordered systems (b). Outside the shaded area 
  causal influences are exponentially suppressed.}
\end{figure}
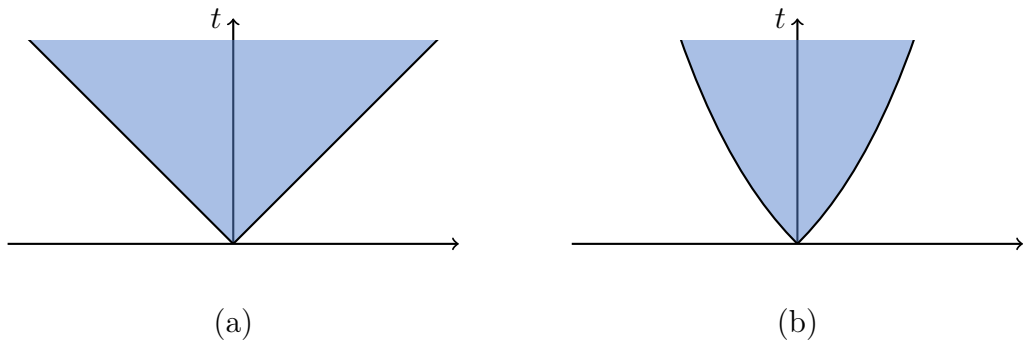

A problem that has recently started to attract an increasing amount of attention is transport in disordered systems. 
For an XY spin chain with disordered interactions and disordered external magnetic field a Lieb-Robinson type bound of the form
\begin{equation}
  \norm[\infty]{[A,B(t)]} \leq C\,n^2\,|t|\,\norm[\infty]A \norm[\infty]B \e^{-\eta\,\dist(A,B)}
\end{equation}
with $n$ the length of the chain and $C,\eta>0$ constants was derived in Ref.~\cite{Burrell2007} (see also Ref.~\cite{Logarithmic} for a similar result).
Notice that the bound is significantly more stringent than the ones we discussed before.
In order for the right hand side to be significantly larger than zero, a time $t$ that scales exponentially with the distance $\dist(A,B)$ is necessary (see Figure~\ref{fig:liebrobinsoncones}).
Hence, in the considered disordered system information propagation is constrained to a region that is not even a cone --- its radius only grows logarithmic with time. There is even evidence that for a given realisation of a disordered XY chain, one obtains a certain type of zero velocity Lieb-Robinson bound with high probability \cite{Hamza2012,MBLMPS}.
We will return to the problem of transport and (many-body) localisation in disordered systems in Section~\ref{sec:mbl}.

\subsection{Time scales for equilibration on average}
\label{sec:timescales}
In this section we summarise what is known about the time scales on which subsystem equilibration to the reduction of the de-phased state happens, i.e., on which time scales small subsystems equilibrate towards their time averaged state.
We will see that it is possible to go beyond what Theorem~\ref{thm:equilibrationonaverage} implies, but that all analytical results known to date that do so have the disadvantage of not being applicable to concrete Hamiltonians, but are only statements about all but few Hamiltonians from certain probability measures.

To discuss such results we will need to refer to and use methods of typicality and measure concentration.
In particular, we will encounter the uniform or Haar measure $\muhaar[U(d)]$ on the unitary group $U(d)$ of dimension $d$ and the probability $\probability_{U \sim \muhaar[U(d)]}(A)$ that a given statement $A$ holds for unitary operators $U \sim \muhaar[U(d)]$ drawn from the Haar measure.
Readers unfamiliar with these constructions might want to refer to Section~\ref{sec:typicality}, where we discuss them in more detail.

We argued in the paragraphs following Theorem~\ref{thm:equilibrationonaverage} that the bounds in \texteqref{eq:equilibrationonaverageforexpectationvalues} and \texteqref{eq:equilibrationonaverageforrestricedpovms} can be expected to become meaningful only if $T$ is of the order of $\landauOmega(d_S^2\,d')$.
As $d'$ usually grows exponentially with the system size the equilibration times implied by Theorem~\ref{thm:equilibrationonaverage} become physically meaningless already for medium sized systems.

There are good reasons to believe that without further assumptions on the Hamiltonian no significantly better general bounds on the subsystem equilibration time can hold. An example of a system that indeed can take exponentially long to equilibrate is a bipartite system in which the subsystem is only coupled to a low dimensional subspace of the Hilbert space of the bath.
It can then take exponentially long before the Hamiltonian on the bath has rotated the state of the bath into this subspace, thereby effectively leaving the subsystem uncoupled for extremely long times (see Ref.~\cite{Malabarba2014} for a related construction).
Such a coupling to a low dimensional subspace is necessarily non-local and hence unphysical.
In Ref.~\cite{Schiulaz2015}, however, it is demonstrated that density inhomogeneities can persist also for exponentially long times even in translation invariant interacting lattice models.
Very slow dynamics is also characteristic for disordered and glassy systems.

Numerical evidence suggests that most natural, locally interacting systems without disorder started in reasonable initial states do not exhibit such extremely long equilibration times,
see for example Refs.~\cite{Rigol08,Venuti09,1108.2703v1,Sirker2013,Fagotti2012,Eckstein2008,Gemmer09,Campos10,Rigol11,1110.4690v1,Rigol07,Rigol2006,1104.0154v1,Torres-Herrera2013,1104.3232v1,Sorg2014}, 
even though surprisingly slowly relaxing local observables can be constructed in some cases \cite{Kim2014a} and also power law approaches to equilibrium can occur \cite{1103.0787v1}.

As it is still unclear how the features of natural many-body models, such as locality of interactions, can be exploited to derive tighter bounds on equilibration time scales, Refs.~\cite{1108.2985v3,1108.0374,1112.5295v1,1107.6035v1,1111.6119} instead consider \emph{random Hamiltonians} and Ref.~\cite{Malabarba2014} certain types of random observables, as well as a class of non-random low rank observables.
We first cover the results of the type derived in Refs.~\cite{1108.2985v3,1108.0374,1112.5295v1,1107.6035v1,1111.6119} but concentrate on Ref.~\cite{1108.0374}, as it goes beyond the rather unrealistic scenario of Hamiltonians with Haar random eigenstates.

As a warm-up, we shall, however, consider exactly the situation of Hamiltonians with Haar random eigenvectors.
First, we define what a \emph{Haar random Hamiltonian} is:
Consider a system with Hilbert space $\mcH$ of dimension $d$ and fix an observable $G \in \Obs(\mcH)$.
Then for $U \sim \muhaar[U(d)]$ the operator
\begin{equation} \label{eq:defrandomhamiltonian}
  \H_G(U) \coloneqq U\,G\,U\ad
\end{equation}
is a \emph{Haar random Hamiltonian}.
Of course, $G$ and $\H_G(U)$ share the same spectrum and eigenvalue multiplicities for any unitary $U$, but the energy eigenstates / spectral projectors of $\H_G(U)$ are Haar random.
Fixing $G$ is thus equivalent to fixing the eigenvalues and degeneracies of the ensemble $\H_G(U),\ U \sim \muhaar[U(d)]$ of Haar random Hamiltonians.

A quantity that will play an important role in the theorems to come is 
\begin{equation} \label{eq:deffouriertransformedcpectrum}
  f_G(t) \coloneqq \frac{1}{d}\,\sum_{k=1}^d \e^{- \iu\,\tilde{E_k}\,t} ,
\end{equation}
where $(\tilde E_k)_{k=1}^d$ is the sequence of eigenvalues with respective multiplicity of $G$ (and hence also of $\H_G(U)$ for any unitary $U$).
The function $f_G$ can be interpreted as the \emph{Fourier transform} of the sequence $(\tilde E_k)_{k=1}^d$ \cite{1108.0374}.

We can now state the first result of Ref.~\cite{1108.0374}, which concerns quantum systems composed of spin-1/2 systems, so-called \emph{qubits}, i.e., quantum systems whose Hilbert space is $\C^2$:
\begin{theorem}[Equilibration under Haar random Hamiltonians {\cite[Result~1]{1108.0374}}] \label{thm:equilibrationunderrandomhamiltonians}
  Consider a bipartite system consisting of $|\Vset|$ many qubits, i.e., $\Vset = S \dunion B$ and $\mcH = \bigotimes_{x\in\Vset} \mcH_{\{x\}}$ with $\mcH_{\{x\}} = \C^2$ for all $x \in \Vset$, starting in a fixed initial state $\rho(0) \in \Qst(\mcH)$.
  Then, for every $G \in \Obs(\mcH)$, every $t \in \R$, and every $\epsilon > 0$ it holds that
  \begin{equation} \label{eq:equilibrationunderrandomhamiltonsbound}
    \probability_{U \sim \muhaar[U(d)]}\left( \tracedistance{\rho^S(t)}{\omega_{H_G(U)}^S} > \frac{\sqrt{d_S}}{2\,\epsilon} 
    \left(
    { |f_G(t)|^4 + \frac{g_G^2}{d^2} + \frac{7}{d_B} }
    \right)^{1/2}
    \right) < \epsilon ,
  \end{equation}
  where $\omega_{H_G(U)}^S \coloneqq \Tr_B(\$_{\H_G(U)}(\rho(0)))$ and $g_G \coloneqq \max_{k\in[d]} |\{l\oftype \tilde E_l = \tilde E_k \}|$ with $(\tilde E_k)_{k=1}^d$ the sequence of eigenvalues with respective multiplicity of $G$.
\end{theorem}
A very similar result is also contained in Ref.~\cite{1108.2985v3}.
Essentially, Theorem~\ref{thm:equilibrationunderrandomhamiltonians} connects the temporal evolution of the trace distance of $\rho^S(t)$ from the equilibrium state $\omega^S_G$ with the temporal evolution of $|f_G(t)|$.
If the bath is large and the Hamiltonian has only few degeneracies, then for most Haar random Hamiltonians the distance $\tracedistance{\rho^S(t)}{\omega_{H_G(U)}^S}$ is small whenever $|f_G(t)|$ is small.
This will make it possible to give bounds on equilibration time scales.

The above result can be extended to a more general ensemble of random Hamiltonians.
More specifically, consider again the setting of a composite system of $N$ qubits and the ensemble $\H_G(U)$, but now with $G \in \Obs(\mcH)$ diagonal in some product basis and $U$ given by a \emph{random circuit} of \emph{circuit depth} $C \in \Z^+$.
Here, a \emph{circuit} is a sequence of so-called \emph{quantum gates}, i.e., unitary quantum channels that each act on only one or two qubits.
The gates can be members of a so called \emph{universal gate set}, i.e., a set of quantum gates such that any unitary can be approximated arbitrarily well by a circuit of gates from this set.
The \emph{circuit depth} of a circuit is the number of gates in the circuit.
Finally, a \emph{random circuit} is a circuit in which the gates have been drawn randomly according to some measure from a universal gate set.
We write $\muC$ for the measure on unitaries induced by random circuits of circuit depth $C$ with gates drawn uniformly at random from some fixed, finite universal gate set.
It is known that $\lim_{C\to\infty} \muC = \muhaar[U(d)]$  and that for large enough $C$ the measure $\muC$ approximates $\muhaar[U(d)]$ in the sense of being an \emph{approximate unitary design} \cite{Brandao2012}.
This holds regardless of which finite universal gate set is used.

For the random circuit ensemble of random Hamiltonians the following statement holds, which generalises Theorem~\ref{thm:equilibrationunderrandomhamiltonians}:
\begin{theorem}[Equilibration under random circuit Hamiltonians {\cite[Result~3]{1108.0374}}] \label{thm:equilibrationundercircuitrandomhamiltonians}
   Consider a bipartite system consisting of $N \coloneqq |\Vset|$ many qubits, i.e., $\Vset = S \dunion B$ and $\mcH = \bigotimes_{x\in\Vset} \mcH_{\{x\}}$ with $\mcH_{\{x\}} = \C^2$ for all $x \in \Vset$, starting in a fixed initial state $\rho(0)$.
  There exists a constant $\alpha \in \R$ that depends only on the universal gate set such that for every $G \in \Obs(\mcH)$ diagonal in a product basis, every $t \in \R$, every circuit depth $C \in \Z^+$, and every $\epsilon > 0$
  \begin{equation} \label{eq:equilibrationundercircuitrandomhamiltonsbound}
    \probability_{U \sim \muC}\left( \tracedistance{\rho^S(t)}{\omega_{H_G(U)}^S} > \frac{\sqrt{d_S}}{2\,\epsilon} 
    \left({ |f_G(t)|^4 + \frac{g_G^2}{d^2} + \frac{7}{d_B} + d^3\, 2^{-\alpha\,C/N}}\right)^{1/2}
    \right) < \epsilon ,
  \end{equation}
  where $\omega_{H_G(U)}^S \coloneqq \Tr_B(\$_{\H_G(U)}(\rho(0)))$ and $g_G \coloneqq \max_{k\in[d]} |\{l\oftype \tilde E_l = \tilde E_K \}|$ with $(\tilde E_k)_{k=1}^d$ the sequence of eigenvalues with respective multiplicity of $G$.
\end{theorem}
As can be seen from \texteqref{eq:equilibrationundercircuitrandomhamiltonsbound}, a slightly super-linear circuit complexity, i.e., $C = C(N) \notin \landauO(N)$, is sufficient to make the additional term in \texteqref{eq:equilibrationundercircuitrandomhamiltonsbound} (compared to \texteqref{eq:equilibrationunderrandomhamiltonsbound}) go to zero for large $N$.

If this is the case, and in addition $N$ is large enough, the bath is much larger than the subsystem, i.e., $d_B \gg d_S$, and $G$ has only few degeneracies, i.e., $g_G \ll d$, then the right hand side of both \texteqref{eq:equilibrationunderrandomhamiltonsbound} and \eqref{eq:equilibrationundercircuitrandomhamiltonsbound} is approximately equal to $|f_G(t)|^2\,\sqrt{d_S}/(2\,\epsilon)$.
Hence, the bounds are non-trivial for reasonably small $\epsilon$ for all $t$ for which $\sqrt{d_s}\,|f(t)|^2 \ll 1$.
For which times $t$ this is the case of course crucially depends on the spectrum that was fixed by fixing $G$.

The spectrum of the Ising model with transverse field, for example, leads to an approximately Gaussian decay of $|f(t)|^2$, implying an estimated equilibration time of the order of $\landauO(N^{-1/2})$ \cite{1108.0374}.
For more general locally interacting Hamiltonians on $D$-dimensional lattices one can show equilibration times of the order of $\landauO(N^{1/(5\,D)-1/2})$ \cite{1112.5295v1}.

This means that given an initial state $\rho(0)$, if $G$ is chosen to be the Hamiltonian of the transverse field Ising model and $U \sim \muhaar[U(d)]$, then the dynamics under the Haar random Hamiltonian $\H_G(U)$, which has the same spectrum as $G$, is, with high probability, such that the time evolution $\rho\oftype\R\to\Qst(\mcH)$ is such that the state of any small subsystem $S$ equilibrates to the reduced state of the de-phased state on that subsystem, during a time of the order of $\landauO(N^{-1/2})$.
This, however, is in contradiction with the intuition that larger systems should take longer to equilibrate, simply because excitations in locally interacting spin systems travel with a finite speed (see also Section~\ref{sec:liebrobinson}).
One would expect that for locally interacting systems of $N$ spins on a $D$ dimensional regular lattice with nearest neighbour or short range interactions, subsystem equilibration should happen on a time scale of the order of $\landauTheta(N^{1/D})$, where $N^{1/D}$ is the \emph{linear size of the system}, for many reasonable initial states.

Another point of criticism is that one can show that for Haar random Hamiltonians the subsystem equilibrium state is the maximally mixed state \cite[Corollary 1]{1112.5295v1} and a similar statement can be shown for the random circuit ensemble of random Hamiltonians.
Systems to which the above results apply can thus never exhibit subsystem equilibration to an interesting, e.g., finite temperature, state.

The reason for both of these problems is that neither the model of Haar random Hamiltonians nor that of Hamiltonians whose diagonalizing unitary is given by a random circuit with high circuit complexity are good models for realistic, locally interacting quantum systems.
Simply put, even though random Hamiltonian ensembles have been successfully used to model certain features of realistic Hamiltonians in the context of \emph{random matrix theory} \cite{1102.0528v1,0412017v2,Guhr1998,1006.1634v1,Gemmer09,Tabor1989,Bohigas1984,Tao2012,mehta90}, the eigenstates of reasonable locally interacting quantum systems are far from Haar random.

We hence turn to the results of Ref.~\cite{Malabarba2014} for concrete Hamiltonians and measurements.
The main result of that work is a bound on the equilibration time of low rank measurements, i.e., POVMs with two outcomes, one of which is a low rank projector.
Such measurements do not correspond to local observables, but rigorous bounds on their equilibration behaviour can be given even for concrete situations:
\begin{theorem}[Fast equilibration of low rank observables \cite{Malabarba2014}] \label{thm:fastequilibrationforlowrankobservables}
  Given a system with Hilbert space $\mcH$ and non-degenerate Hamiltonian $\H \in \Obs(\mcH)$ with spectral decomposition $\H = \sum_{k=1}^{d} E_k\,\ketbra{E_k}{E_k}$.
  For $\rho(0) \in \Qst(\mcH)$ the initial state of the system, let $\omega = \$_H(\rho(0))$ be the de-phased state and define the energy level occupations $p_k \coloneqq \bra{E_k}\rho(0)\ket{E_k}$.
  Let $\POVMs = \{(\Pi,\1-\Pi)\}$ with $\Pi$ a rank $K$ projector, then
  \begin{equation}
    \taverage[T]{\tracedistance[\POVMs]{\rho(t)}{\omega}} \leq C\,(\eta(1/T)\,K)^{1/2}
  \end{equation}
  with $C = 5\,\pi/(4\,\sqrt{1-1/\e})+1$ and for any $\Delta\geq 0$
  \begin{equation}
    \eta(\Delta) \coloneqq \sup_{E \in \R} \sum_{k\suchthat E_k \in [E,E+\Delta]} p_k.
  \end{equation}
\end{theorem}
One can now show that $\eta$ is lower bounded by $1/\deff(\omega)$ but also argue that, up to reasonably large $T$, it holds that $\eta(1/T) \leq C'/T$ with $C'$ a constant that depends on the shape of the energy distribution of the initial state.
In particular, the above theorem predicts an at least power-law like approach to equilibrium of low rank measurements on a time scale proportional to $K/\sigma_E$ with $\sigma_E$ the energy uncertainty in the initial state (for details see Ref.~\cite{Malabarba2014}).

In addition to the above results some \emph{lower} bounds on equilibration time scales exist:
For example, if a state has overlap only with energy eigenstates of the Hamiltonians in an energy interval of width $\Delta E$, then the equilibration time is at least of the order of $\landauOmega(1/\Delta E)$ \cite{Gogolin10-masterthesis} (see also Ref.~\cite{PhysRevE.50.88}).
Similarly, if the Hamiltonian $\H$ of a bipartite system with $\Vset = S \dunion B$ is uncoupled, except for a small 
coupling Hamiltonian $H_I \coloneqq \H - \H_S - \H_B$, then the equilibration time is at least of the order of $\landauOmega(1/\norm[\infty]{\H_I})$ \cite[Section 2.6.3]{Gogolin10-masterthesis}.
Similarly, lower bounds on the equilibration and thermalisation time --- as will be discussed in Section \ref{sec:thermalisation} --- follow from bounds on the rate of change of certain entropies \cite{MasterThesisHutter,Hutter11}.
In Ref.~\cite{Kastner11}, lower bounds on the equilibration time of the type $\landauOmega(N^{1/2})$ have been obtained for a class of spin systems with long range interactions.
For spin systems with short range interactions, Lieb-Robinson bounds (see Section~\ref{sec:liebrobinson}) immediately imply lower bounds on the equilibration time for certain initial states that are of the order of the linear size of the system.
Finally, in systems whose density of states can be approximated by a continuous function the \emph{Riemann-Lebesgue Lemma} \cite{Bochner49} can be used to give upper bounds on equilibration time scales \cite{Yukalov2011}.
Despite the large number of results the full problem still awaits a solution.

\subsection{Fidelity decay}
\label{sec:fidelitydecay}
A scenario in which the equilibration behaviour has been studied in detail and is now particularly well understood is that of \emph{fidelity decay}.
Rather than looking at the expectation value of say a local observable the quantity, whose equilibration is of interest here is the fidelity between the initial state and the time evolved state at time $t$.
For pure initial states $\psi(0) = \ketbra{\psi}{\psi} \in \Qst(\mcH)$ and unitary time evolution under a Hamiltonian $\H \in \Obs(\mcH)$ the fidelity takes the simple form
\begin{equation}
  \fidelity{\psi(0)}{\psi(t)} = \Tr(\psi(0)\,\psi(t)) = \left|\bra\psi \e^{-\iu\,\H\,t} \ket\psi\right|^2 = \Big| \sum_k |\braket\psi{E_k}|^2 \e^{-\iu\,E_k\,t} \Big|^2 , 
\end{equation}
which makes apparent that it can be seen as the square of the Fourier transform of the weighted energy distribution of the initial state.

As $\psi(0)$ can be seen as a low rank (in fact rank one) observable, the result on the equilibration times of low rank observables (Theorem~\ref{thm:fastequilibrationforlowrankobservables}) can be used to bound the time scales on which fidelity decays.
Below this, typically power-law bound a rich variety of different decay behaviours can be observed \cite{Flambaum2001,Santos2014,Vyas2014}.

\section{Investigations of equilibration for specific models}
\label{sec:numericalandanalyticalinvesitgationsofequilibration}
There is a large body of literature studying equilibration dynamics of quantum many-body systems partly with analytical, but mostly with numerical methods.
These works typically focus on a specific model or a subclass of models.
In this section we cover a selection of works in this direction.
We will discuss many more works with a similar scope later in Section~\ref{sec:numericalthermal} once we have introduced the concept of thermalisation.

\subsection{Global quenches}
Often the behaviour of quantum systems after a suddenly altered Hamiltonian, a so called \emph{quench} is considered.
In this much discussed setting, the initial state $\rho(0)$ is, e.g., the ground state of a locally interacting Hamiltonian $\H_0$, and following the sudden quench to a different locally interacting Hamiltonian $\H$, properties of
\begin{equation}
  \rho(t) = \e^{-\iu\,t\,\H} \,\rho(0)\, \e^{\iu\,t\,\H}
\end{equation}
are explored.
The seminal early study \cite{Barouch1970} introduces quenches to the literature and finds a ``non-approach to equilibrium'' in the XY model that is mapped to a quadratic fermionic system.
Refs.~\cite{CalabreseCardy06,Calabrese2007} use field theoretical methods to gain insight into the dynamics of correlation functions after quenches.
If the final Hamiltonian is close to being critical, notions of universality are being identified at long times.
The early work \cite{Kollath07} investigates an out of equilibrium phase diagram of the Bose-Hubbard model, arising from quenches from the superfluid to Mott phase.
Ref.~\cite{Flesch08} also considers out of equilibrium dynamics in the Bose-Hubbard model and discusses signatures of equilibration that can be probed using optical super-lattices.
A similar setting is numerically analysed in Ref.~\cite{1101.2659v1}, which is then taken 
as a benchmark for an experiment performing a dynamical quantum simulation.
Ref.~\cite{Krutitsky2014} numerically investigates quenches inside the Mott phase with a method most suitable for lattices with high coordination number.
Ref.~\cite{Barmettler} studies the relaxation dynamics in XXZ chains following a quench,
Here, a rich phenomenology emerges and both oscillatory and exponential relaxation are being observed.
Counter-intuitively, the relaxation speed increases at a critical point for the anisotropy parameter.
The seminal experimental work \cite{Greiner2002a} also studies the on-equilibrium evolution of coherent states, being superpositions of different particle number states in a three dimensional optical lattice,
observing collapses and revivals.
A very powerful tool in numerical studies is provided by time-dependent variants of the density-matrix renormalisation group (DMRG) approach and related tensor network approaches \cite{Schollwock201196,Hastings2009,Muller-Hermes2012,Hallberg2006}.
An early example being Ref.~\cite{DeChiara2005}, which studies the spreading of entanglement after quenches in Heisenberg spin chains.
Noteworthy are Refs.~\cite{Cramer2008,Flesch08}, which investigate equilibration with such methods in a setting described by the Bose-Hubbard model which can be realised with ultra cold atoms \cite{1101.2659v1}.
Ref.~\cite{Karrasch2012} uses time dependent DMRG methods to study the relaxation dynamics after quenches in the Tomonaga-Luttinger model and in systems of spin-less fermions and finds universal long time behaviour.
Ref.~\cite{Queisser2013} discusses the equilibration in the Bose-Hubbard and Fermi-Hubbard models following a global quench, employing an expansion in large coordination numbers.
Ref.~\cite{1108.2703v1} uses quantum Monte Carlo techniques to investigates the equilibration dynamics after switching off the coupling of the Hubbard-model starting form a thermal state.
Ref.~\cite{Manmana2009} focuses on quenches in 1D spin-less fermions.
Ballistic transport is observed, except if the quench is from a metallic state deep into the insulating phase, in which case local domains form reminiscent of the picture provided by the Kibble-Zurek mechanism \cite{Zurek_review,KZP}.

\subsection{Local and geometric quenches}
By no means are the sudden global quenches the only type of non-equilibrium situation considered in the literature.
\emph{Local quenches} are also frequently investigated \cite{Calabrese2007a,LocalQuenches2,PhysRevE.88.032913,Jurcevic2014,1004.2232v1}, as well as \emph{geometric quenches} \cite{GeometricCaux,GeometricQuench,Rigol08}, in which the system's response to a sudden alteration of its geometry is being studied.
Ref.~\cite{GeometricCaux} considers general geometric quenches between systems integrable by means of the Algebraic Bethe ansatz and how it allows to compute overlaps between eigenstates of the old and new Hamiltonian.
Ref.~\cite{GeometricQuench} investigates the dynamics of entanglement and equilibration after a geometric quench in the anisotropic spin-1/2 Heisenberg chain.
Ref.~\cite{Rigol08} studies the relaxation after a "valve" between to previously isolated systems of hard-core bosons is opened.
Refs.~\cite{Calabrese2007a,LocalQuenches2} develop a quantum field theory approach to describe the growth of entanglement and the dynamics of correlation functions after a quench during which two uncoupled halves, initially in their ground state, of a translation invariant system are joined together.
Refs.~\cite{1004.2232v1,Ponomarev2012} consider a related scenario in which two systems initially at different temperature are joined together and the authors also find equilibration.
Also related is the series of works Refs.~\cite{Bernard2014,Bernard2012,Bhaseen2015,Doyon2015,DeLuca2013} in which properties of the non-equilibrium steady state are studied that can emerge in such a situation if the two systems are infinitely large. In particular, the steady state energy current and its fluctuations, as well as the time dependence of local observables are calculated. 
Refs.~\cite{Bernard2015,Castro-Alvaredo2014} consider 
the non-equilibrium dynamics emerging from bringing two
systems together in a language of conformal and relativistic quantum field theory.
Ref.~\cite{PhysRevE.88.032913} presents a detailed numerical study of the time evolution under various, integrable and non-integrable, translation invariant spin Hamiltonians for several types of initial states, including domain wall states with all spin in the left half up and all in the right half down. This situation can be thought of as a local quench.
A similar setting is analysed in Ref.~\cite{Gritsev10} for XXZ chains, where equilibration is also found, albeit to a state that retains memory on its initial state.
Local quenches and the subsequent (quasi-particle) dynamics can now also be probed experimentally with impressive precision \cite{Jurcevic2014}.

\subsection{Entanglement dynamics}\label{sec:entanglement}
Early on, it has been realised that the light-cone like propagation of excitations following global quenches is accompanied by a growth of entanglement if the initial state has low entanglement or even is a product state \cite{CalabreseQuenchEntanglement,DeChiara2005,Calabrese2007a,Cramer2006,Eisert06,Bravyi06-1}. 
Entanglement is always defined with respect to a separation of the system into distinct, spatially separate subsystems.
Bipartite entanglement of pure states with respect to a decomposition $\Vset = X \union \compl X$ for some subsystem $X \subset \Vset$ can be measured in terms of the \emph{entanglement entropy} defined for any state $\rho$ as
\begin{equation}
  E_X(\rho) = \Svn(\rho^X),
\end{equation}
where $\Svn$ denotes the von Neumann entropy.
In a precise sense, this is the ``unique measure of entanglement'' in this pure bipartite setting \cite{Plenio07} and we call a pure state $\rho$ \emph{uncorrelated} (with respect to a decomposition $\Vset = X \union \compl X$) if $E_X(\rho) = 0$.
Other \emph{Renyi entropies}
\begin{equation}
  E_X^p(\rho) = \frac{1}{1-p}\log_2(\tr( (\rho^X)^p)
\end{equation}
for $p>0$, however, also play a role when it comes to questions of approximations of states with tensor network methods \cite{Approximation,Ge}.
Lieb-Robinson bounds imply an affine upper bound for the entanglement entropy following global quenches:

\begin{observation}[Entanglement growth]
For any locally interacting system with Hamiltonian $\H$ and any $X \subseteq \Vset$ it holds that
\begin{equation}
  E_X(\rho(t)) - E_X(\rho(0)) \in \landauO(t) .
\end{equation}
Conversely, there exist pairs of (translation invariant) locally interacting Hamiltonians $\H$ and pure uncorrelated initial states $\rho(0)$ such that
\begin{equation}
  E_X(\rho(t)) \in \landauOmega(t) .
\end{equation}
Moreover, if $\rho(0)$ is uncorrelated, then for $t$ fixed and any family of subsystems $X \subset \Vset$ of increasing size the entanglement entropy scales only like the boundary of these subsets in the sense that
\begin{equation}
  E_X(\rho(t)) \in \landauO(|X_\partial|) ,
\end{equation}
where $X_\partial$ is the set of elements of the edge set $\Eset$ of the Hamiltonian that overlap with both $X$ and $\compl X$.
\end{observation} 
The first and last statements have been proven in Refs.~\cite{Eisert06,Bravyi06-1} and improved in Ref.~\cite{VanAcoleyen2013}. The second statement follows from Refs.~\cite{Cramer2006,Schuch08}.
The intuition behind these statements is clear: 
Following the ballistic propagation of quasi-particles, at most a linear growth of 
entanglement over any finite cut can be observed.
This indeed follows from suitable Lieb-Robinson bounds.
Similarly, such a bound can be saturated for quadratic models, so it is tight in this sense.
At the same time, for each fixed time $t>0$, the entanglement entropy follows what is called an \emph{area law} \cite{Eisert2008} for the entanglement entropy, in that the entanglement scales at most like the boundary area $|X_\partial|$ of the subset $X$.

There is a large body of literature that corroborates the intuition behind this theorem \cite{Calabrese2007a,DeChiara2005,Eisert06,VanAcoleyen2013,Jurcevic2014,Kollath08}.
The early Ref.~\cite{DeChiara2005}, for example, studies the spreading of 
entanglement after quenches in Heisenberg spin chains, Ref.~\cite{Cramer2006} 
studies quadratic models, Ref.~\cite{Calabrese2007a,CalabreseQuenchEntanglement} specifically develops the quasi-particle picture.

The above notion of entanglement is not the only one that can and has been meaningfully
considered. The correlations present in states arising from out of equilibrium dynamics can also be captured in terms of bipartite entanglement of separated subsystems which jointly still form a small subset of the entire lattice.
Since the state under consideration is then no longer pure, other \emph{measures of entanglement} \cite{Entanglement,Plenio07} have to be employed, such as the \emph{entanglement of formation} \cite{BennettBig,Wootters} or the \emph{negativity} \cite{VidalNegativity,PhD,OldNegativity}.
Such entanglement measures have defining features of being monotone under \emph{local operations and classical communication} (LOCC) and vanish on separable states as discussed in Eq.~\eqref{Separable}).

While it is known from the \emph{monogamy of entanglement} \cite{Seevick10} that at any time most sites of a lattice are not entangled, in the course of entanglement dynamics, remote sites generically get entangled at suitable times \cite{Nano,Sengupta,Detecting}.
The intuition is that ``wave fronts'' of entanglement propagate ballistically through the lattice. Such bipartite entanglement has already been experimentally observed in systems of trapped ions \cite{Jurcevic2014}.

\subsection{Ramps, slow quenches and the dynamics of quantum phase transitions}
Many works discuss also \emph{non-instantaneous ramps} and other instances of so-called \emph{slow quenches}.
In this context, the dynamics under a family of locally interacting Hamiltonians
\begin{equation}
  \H(t) = \H_0 + f(t)\, V
\end{equation}
is usually studied, with $f \oftype [0,\infty[ \to \R$ being a suitably slowly varying function and $\rho(0)$ the ground state of $\H(0)=\H_0$.
Such 
a situation is specifically interesting when at an instance in time $t_0$ the Hamiltonian $\H(t_0)$ undergoes a second order quantum phase transition.
If the change of the Hamiltonian in time is sufficiently slow, far away from the phase transition the adiabatic theorem will be applicable and the state $\rho(t)$ is then well approximated by 
the instantaneous ground state of the Hamiltonian at that given time $t$.
However, in the vicinity of the critical point, the Hamiltonian gap will close down, and no 
slow change of the Hamiltonian will be sufficiently slow such that the adiabatic theorem could still capture the situation at hand.
This setting hence allows to explore the \emph{dynamics of quantum phase transitions}.
This review cannot give justice to this topic, which can be considered a research field in its own right.
We still attempt to give a short sketch of important ideas.

The \emph{Kibble-Zurek mechanism} provides an intuitive understanding of the 
phenomenology of such slow quenches across critical points \cite{Zurek_review,KZP}.
It is specifically well understood for thermal phase transitions, a setting in which it has also been experimentally tested \cite{Kibble_Vortex,Ulm2013,KZM_Ion2}.
For quantum phase transitions similar scaling laws can be derived in the limit of infinitely slow ramps, invoking adiabatic perturbation theory and universality arguments but the situation is more involved \cite{Polkovnikov08-1,Polkovnikov2005,Sen2008,Barankov2008}.
Ref.~\cite{Schutzhold2006}, considers an exponential ramp from the superfluid phase into the insulating one and calculates the time dependence of various experimentally relevant quantities for this case, and Ref.~\cite{Fischer2008} treats further analytically solvable ramps.

Ref.~\cite{Chen_Slow} experimentally probes the Mott-insulator to superfluid transition in the Bose-Hubbard model by slowly decreasing the ratio of the interaction energy to the hopping strength.
Ref.~\cite{Braun2014a} studies the Mott insulator to superfluid quantum phase transition experimentally with ultra cold atoms and compares the findings to extensive numerics for the Bose-Hubbard model, using exact diagonalisation and tensor network techniques.
Also in the Bose-Hubbard model Refs.~\cite{Bernier2012,Bernier2011} study the formation and melting of Mott-insulating domains during ramps with tDMRG methods.
Ref.~\cite{Haque_Crossover} analytically investigates finite time ramps of the inter-mode interaction strength in a Luttinger liquid model.
The series of works Refs.~\cite{Uhlmann2007,Uhlmann2010,Uhlmann2010a} investigates the formation of topological defects after quenches that involve the breaking of a continuous rotational symmetry.
For reviews on this field --- to the extent it is understood to date --- see Refs.~\cite{Zurek_review,Polkovnikov08-1,Campo2014}.

\section{Quantum maximum entropy principles}
\label{sec:aquantummaximumentropyprinciple}
In this section we connect the pure state statistical mechanics framework with the canonical approach to justify the ensembles of statistical physics by means of a maximum entropy principle.
We first show that the apparent equilibrium state in systems that equilibrate on average can always be defined in terms of an entropy maximisation under the constraint that the expectation values of all conserved quantities are held fix.
Then we will discuss the possibly surprising fact that in many cases, in particular following quenches of sufficiently complex quantum systems, equilibrium expectation values of many relevant observables are very well approximated by those in a state that is the entropy maximiser given a much smaller set of constants of motion --- a so-called \emph{generalised Gibbs ensemble} (GGE).

\subsection{A maximum entropy principle based on all constants of motion}
We have seen in Section~\ref{sec:equlibrationintheweaksense} that if the expectation value of an observable or the reduced state of a subsystem equilibrates on average, then they necessarily equilibrate to their expectation value in, or the reduced state of, the time averaged/de-phased state $\omega = \taverage{\rho} = \$_\H(\rho(0))$.
The state $\omega$ hence encodes the information necessary to describe the equilibrium properties of such a system.
It turns out that it is also the maximum entropy state given all constants of motion:

\begin{theorem}[Maximum entropy principle \cite{PhysRevLett.10-6}] \label{thm:maximumentropyprinciple}
  Consider the time evolution $\rho\oftype\R\to\Qst(\mcH)$ of a quantum system with Hilbert space $\mcH$ and Hamiltonian $\H \in \Obs(\mcH)$.
  If the expectation value of an operator $A \in \Bop(\mcH)$ equilibrates on average, then it equilibrates towards its time average, given by
  \begin{equation}
    \taverage{\Tr(A\,\rho)} = \Tr(A\,\omega) ,
  \end{equation}
  where $\omega = \taverage{\rho}$ is the unique quantum state that maximises the von Neumann entropy $\Svn$, given all conserved quantities.
\end{theorem}
\begin{proof}
  That the equilibrium value of the expectation value of $A$ is given by $\Tr(A\,\omega)$ follows directly from the definition of equilibration on average.
  The time averaged state $\omega$ is equal to the de-phased initial state 
  \begin{equation}
    \$_H(\rho(0)) = \sum_{k=1}^{d'} \Pi_k\,\rho(0)\,\Pi_k .
  \end{equation}
  The de-phasing map $\$_\H$ is a so-called \emph{pinching} and the von Neumann entropy is non-decreasing under pinchings \cite[Problem II.5.5]{bhatia} (this is a generalisation of \emph{Schur's theorem}).
  Furthermore, two states $\sigma_1,\sigma_2 \in \Qst(\mcH)$ yield the same expectation values for all conserved quantities, i.e., all $A \in\Obs(\mcH)$ that commute with the Hamiltonian $[A,\H] = 0$, if and only if $\$_\H(\sigma_1) = \$_\H(\sigma_2)$.
  This already shows that $\omega$ has the maximal achievable von Neumann entropy given all conserved quantities (see also figure~\ref{fig:maximumentropyprinciple}).
  It remains to show uniqueness.
  Let $\Basis$ be a basis of the linear span of all $A \in\Obs(\mcH)$ with $[A,\H] = 0$.
  The objective function of the maximisation problem, namely the von Neumann entropy, is a strictly concave function $\Svn\oftype\Qst(\mcH) \to \R$ and it is optimised over all $\sigma \in \Qst(\mcH)$ under the finite number of affine equality constrains $\forall B \in \Basis \itholds \Tr(B\,\sigma) = \Tr(B\,\rho(0))$.
  Under these conditions uniqueness follows from a standard result from convex optimisation \cite{Boyd2004}.
\end{proof}

\begin{figure}[tb]
  \centering
%  \begin{subfigure}{0.3\textwidth}
    \begin{tikzpicture}
      \begin{scope}[scale=0.45,transform shape]
      \path[use as bounding box] (-1.5,-0.3) rectangle (8.7,-8.7);
      \foreach \i/\ri/\phii/\omegai in {1/0.137264/-159.008/1,2/0.0919508/93.4659/1,3/0.204039/-110.476/1,4/0.14146/28.212/2,5/0.116946/145.35/4,6/0.128177/59.1885/4,7/0.113461/-153.794/5,8/0.0667032/-167.439/6}
      \foreach \j/\rj/\phij/\omegaj in {1/0.137264/-159.008/1,2/0.0919508/93.4659/1,3/0.204039/-110.476/1,4/0.14146/28.212/2,5/0.116946/145.35/4,6/0.128177/59.1885/4,7/0.113461/-153.794/5,8/0.0667032/-167.439/6}
      {
        \draw[very thin,-latex] (\i,-\j) -- +(\phii-\phij:\ri*\rj*12); 
        % \draw<4>[thick,-latex] (\i,-\j) -- +({\phii-\phij+(\omegai-\omegaj)*60}:\ri*\rj*12); 
        \ifnum \omegai=\omegaj
        % \draw<5->[thick,-latex] (\i,-\j) -- +(\phii-\phij:\ri*\rj*12);
        \fi
        % \draw<4>[thick] (\phii-\phij:\ri*\rj*12)+(\i,-\j) arc (\phii-\phij:{\phii-\phij+(\omegai-\omegaj)*60}:\ri*\rj*12)+(\i,-\j);
        \node[minimum size=1cm] at (\i,-\j) (n\i\j) {};
      }
      
      \draw[thick] (n11.north west) to[bend right=6] node[midway,anchor=base east] {\Large $\rho(0)=$} (n18.south west);
      \draw[thick] (n88.south east) to [bend right=6] (n81.north east);
      \fill[opacity=0.2,niceblue] (n11.north west) rectangle (n33.south east);
      \fill[opacity=0.2,niceblue] (n44.north west) rectangle (n44.south east);
      \fill[opacity=0.2,niceblue] (n55.north west) rectangle (n66.south east);
      \fill[opacity=0.2,niceblue] (n77.north west) rectangle (n77.south east);
      \fill[opacity=0.2,niceblue] (n88.north west) rectangle (n88.south east);
    \end{scope}
      \node [below=5pt of n48] {(a)};
    \end{tikzpicture}
    \hfill
%    \caption{}
%    \label{fig:maximumentropyprinciplea}
%  \end{subfigure}\hfill
%  \begin{subfigure}{0.3\textwidth}
    \begin{tikzpicture}
      \begin{scope}[scale=0.45,transform shape]
      \path[use as bounding box] (-1.5,-0.3) rectangle (8.7,-8.7);
      \foreach \i/\ri/\phii/\omegai in {1/0.137264/-159.008/1,2/0.0919508/93.4659/1,3/0.204039/-110.476/1,4/0.14146/28.212/2,5/0.116946/145.35/4,6/0.128177/59.1885/4,7/0.113461/-153.794/5,8/0.0667032/-167.439/6}
      \foreach \j/\rj/\phij/\omegaj in {1/0.137264/-159.008/1,2/0.0919508/93.4659/1,3/0.204039/-110.476/1,4/0.14146/28.212/2,5/0.116946/145.35/4,6/0.128177/59.1885/4,7/0.113461/-153.794/5,8/0.0667032/-167.439/6}
      {
        % \draw[very thin,-latex] (\i,-\j) -- +(\phii-\phij:\ri*\rj*12); 
        \draw[very thin,-latex] (\i,-\j) -- +({\phii-\phij+(\omegai-\omegaj)*60}:\ri*\rj*12); 
        \ifnum \omegai=\omegaj
        \draw[very thin,-latex] (\i,-\j) -- +(\phii-\phij:\ri*\rj*12);
        \fi
        \draw[very thin] (\phii-\phij:\ri*\rj*12)+(\i,-\j) arc (\phii-\phij:{\phii-\phij+(\omegai-\omegaj)*60}:\ri*\rj*12)+(\i,-\j);
        \node[minimum size=1cm] at (\i,-\j) (n\i\j) {};
      }
      \draw[thick] (n11.north west) to[bend right=6] node[midway,anchor=base east] {\Large $\rho(t)=$} (n18.south west);
      \draw[thick] (n88.south east) to [bend right=6] (n81.north east);
      \fill[opacity=0.2,niceblue] (n11.north west) rectangle (n33.south east);
      \fill[opacity=0.2,niceblue] (n44.north west) rectangle (n44.south east);
      \fill[opacity=0.2,niceblue] (n55.north west) rectangle (n66.south east);
      \fill[opacity=0.2,niceblue] (n77.north west) rectangle (n77.south east);
      \fill[opacity=0.2,niceblue] (n88.north west) rectangle (n88.south east);
    \end{scope}
      \node [below=5pt of n48] {(b)};
    \end{tikzpicture}
    \hfill
%    \caption{}
%    \label{fig:maximumentropyprincipleb}
%  \end{subfigure}\hfill
%  \begin{subfigure}{0.3\textwidth}
    \begin{tikzpicture}
      \begin{scope}[scale=0.45,transform shape]        
      \path[use as bounding box] (-1.5,-0.3) rectangle (8.7,-8.7);
      \foreach \i/\ri/\phii/\omegai in {1/0.137264/-159.008/1,2/0.0919508/93.4659/1,3/0.204039/-110.476/1,4/0.14146/28.212/2,5/0.116946/145.35/4,6/0.128177/59.1885/4,7/0.113461/-153.794/5,8/0.0667032/-167.439/6}
      \foreach \j/\rj/\phij/\omegaj in {1/0.137264/-159.008/1,2/0.0919508/93.4659/1,3/0.204039/-110.476/1,4/0.14146/28.212/2,5/0.116946/145.35/4,6/0.128177/59.1885/4,7/0.113461/-153.794/5,8/0.0667032/-167.439/6}
      {
        % \draw[very thin,-latex] (\i,-\j) -- +(\phii-\phij:\ri*\rj*12); 
        % \draw[very thin,-latex] (\i,-\j) -- +({\phii-\phij+(\omegai-\omegaj)*60}:\ri*\rj*12); 
        \ifnum \omegai=\omegaj
        \draw[very thin,-latex] (\i,-\j) -- +(\phii-\phij:\ri*\rj*12);
        \fi
        % \draw[very thin] (\phii-\phij:\ri*\rj*12)+(\i,-\j) arc (\phii-\phij:{\phii-\phij+(\omegai-\omegaj)*60}:\ri*\rj*12)+(\i,-\j);
        \node[minimum size=1cm] at (\i,-\j) (n\i\j) {};
      }
      
      \draw[thick] (n11.north west) to[bend right=6] node[midway,anchor=base east] {\Large $\taverage{\rho}=$
      } (n18.south west);
      \draw[thick] (n88.south east) to [bend right=6] (n81.north east);
      \fill[opacity=0.2,niceblue] (n11.north west) rectangle (n33.south east);
      \fill[opacity=0.2,niceblue] (n44.north west) rectangle (n44.south east);
      \fill[opacity=0.2,niceblue] (n55.north west) rectangle (n66.south east);
      \fill[opacity=0.2,niceblue] (n77.north west) rectangle (n77.south east);
      \fill[opacity=0.2,niceblue] (n88.north west) rectangle (n88.south east);
    \end{scope}
      \node [below=5pt of n48] {(c)};
    \end{tikzpicture}
%    \caption{}
%    \label{fig:maximumentropyprinciplec}
%  \end{subfigure}
  \caption{(Reproduction from Ref.~\cite{Gogolin2014}) Dephasing implies a maximum entropy principle.
A quantum system started in an initial state $\rho(0)$ represented in panel (a) in an eigenbasis of its Hamiltonian $\H$ with degenerate subspaces corresponding to the squares, evolves (b) in a way such that time averaging its evolution (c) has the same effect as de-phasing the initial state with respect to $\H$.
The time averaged state $\taverage{\rho}$ is the state that maximises the von Neumann entropy under the constraint that all conserved quantities give the same expectation value as in the initial state $\rho(0)$.}
  \label{fig:maximumentropyprinciple}
\end{figure}
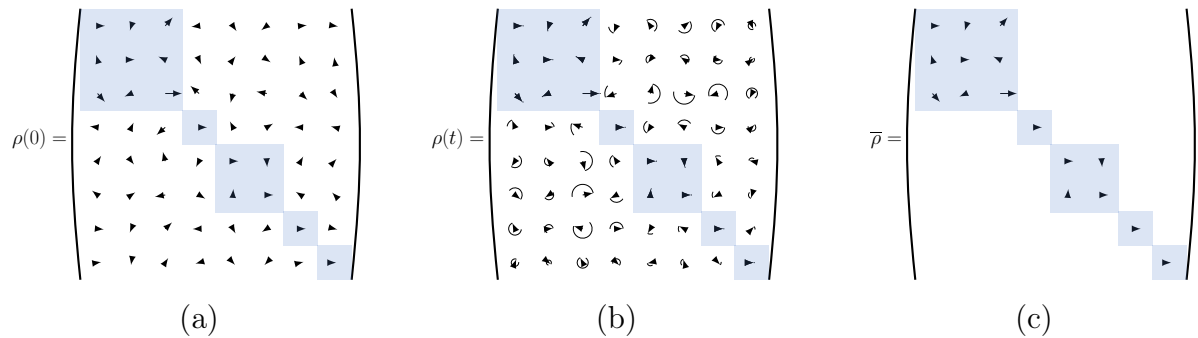
Theorem~\ref{thm:maximumentropyprinciple} is very reminiscent of Jaynes' maximum entropy principle.
It is in any case remarkable that it is not, as in Jaynes' approach, a postulate motivated by a subjective interpretation of probability that is taken as a starting point of a statistical theory, but a consequence of purely quantum mechanical considerations.
The unitary quantum dynamics of closed systems alone gives rise to a maximum entropy principle.

Theorem~\ref{thm:maximumentropyprinciple} is at the same time not the end of the story.
It says that the equilibrium expectation values of all observables that equilibrate on average can be calculated from the state that maximises the von Neumann entropy given \emph{all} conserved quantities (compare also Ref.~\cite{a726f01}).
The number of \emph{all} linearly independent conserved quantities of a composite quantum system, however, increases exponentially with the number of constituents, and finding each of them individually usually again requires resources that scale exponentially with the system size.
The predictive power of Theorem~\ref{thm:maximumentropyprinciple} is hence rather limited.

\subsection{Generalised Gibbs ensembles}
\label{sec:gge}
In light of the insight discussed in the last section an interesting question to ask is \cite{Sirker2013}: ``How many and which conserved quantities are actually relevant? Can one reasonably describe the equilibrium state by maximising entropy holding the expectation values of a much smaller number of possibly even efficiently obtainable conserved quantities fixed?''

For many practically relevant locally interacting Hamiltonians, a number of conserved quantities can be identified that are local in some sense.
In fact, one of the possible definitions of quantum integrability (see also Section~\ref{sec:integrabilityinquantumtheory}) is that there exists a number of conserved quantities scaling linearly in the system size.
When such systems seemingly equilibrate to the time average state $\omega$ under unitary dynamics, this time average can not be expected to be described by a Gibbs ensemble.
The system may, however, still be reasonably expected to equilibrate to the maximum entropy state given these suitably local constants of motion (see Refs.~\cite{Polkovnikov11,Perk1977a,Girardeau1969} and the references therein).
Such a maximum entropy state is usually referred to as a \emph{generalised Gibbs ensemble} (GGE) \cite{Rigol07,Rigol08,PhysRevLett.10-5,CauxEssler,deNardis2,Cassidy11,1104.0154v1,PhysRevLett.111.100401}.

More precisely, a \emph{conserved quantity} is an observable $A \in \Obs(\mcH)$ for which $[\H,A] = 0$.
Moreover, an operator $A \in \Obs(\mcH)$ is called \emph{local} if it is only supported on some $X\subset \Vset$, with $|X|= K$, for some constant $K\in \N$ independent of the system size.

Often, one considers conserved quantities that are \emph{approximately local}.
This notion can be made precise as follows.
For simplicity, we restrict the attention to finite dimensional spin systems.
For a region $X\subset\Vset$ and $l\in\Z^+$ denote by $X_l$ the sets
\begin{equation}
  X_l \coloneqq \{x \in  \Vset \oftype \dist(x,X) < l \},
\end{equation}
of sites of the lattice that contain $X$ as well as all sites within distance at most $l$ from some site in $X$.
For regions $X\subset\Vset$ define the map $\Gamma_X \oftype \Bop(\mcH) \to \Bop(\mcH)$ which acts on operators $A \in \Bop(\mcH)$ as
\begin{equation}
  \Gamma_X(A) \coloneqq  \Tr_{\compl X} (A) \otimes \1_{\compl X} / d^{|\compl X|} \, .
\end{equation}
If $\supp(A) \subseteq X$, then $\Gamma_X(A) = A$, otherwise it die-cuts away everything of $A$ that acts non-trivially outside of $X$, i.e., for any $A \in \Bop(\mcH)$ and $X \subseteq \Vset$ it holds that $\supp(\Gamma_X(A)) \subseteq X$.

We call an operator \emph{$(g,K)$-local} for some function $g \oftype \Z^+ \to \R$ if there is some $X \subset \Vset$ of cardinality $K$ such that
\begin{equation}
  \norm[\infty]{A - \Gamma_{X_l} (A)} \leq \norm[\infty]{ A }\, g(l).
\end{equation}
Often, the function $g$ is taken to be exponentially decaying
\begin{equation}
  g(x) = c_1\, \e^{-c_2 x}
\end{equation}
for some suitable constants $c_1,c_2>0$.
We call operators $A \in \Bop(\mcH)$ that are $(g,K)$-local in this sense \emph{approximately local (with exponential tails)}.
Taking $g$ to be a step function gives the special case of (exactly) $K$-local operators.

Finally we define the notion of a \emph{quasi-local operator}.
Consider a translation invariant spin system, i.e., a system whose vertex set $\Vset$ corresponds to the sites of a regular lattice in a way such that its Hamiltonian $\H \coloneqq \sum_{X \in \, \Eset} \H_X$ is invariant under a set of translations $T_x \oftype \Vset \to \Vset$ of the lattice, indexed by the element $x \in \Vset$ that is mapped to the first element of $\Vset$, in the sense that for all $x \in \Vset$ it holds that $T_x(\H) \coloneqq \sum_{X \in \, \Eset} \H_{T_x(X)} = \H$.
Thereby, and slightly abusing notation, we have implicitly defined the action of a translation $T_x$ on operators $\Bop(\mcH)$ in an obvious way.
An operator $A \in \Bop(\mcH)$ is then called \emph{quasi-local} \cite{Prosen2014} if it can be written in the form
\begin{equation}
  A = \frac{1}{|\Vset|} \sum_{x \in \Vset} T_x(B) , 
\end{equation}
for some operator $B \in \Bop(\mcH)$ that is approximately local with exponential tails.

We now give a precise definition of the \emph{generalised Gibbs ensemble} in a way that seems useful in its own right:
\begin{definition}[Generalised Gibbs ensemble]
Let $K \in \Z^+$ and $g\oftype\Z^+ \to \R$ be a suitably decaying function.
Consider a system with locally interacting Hamiltonian $\H \in \Obs(\mcH)$ and a set $\mcC$ of constants of motion that are either local, approximately local, or quasi-local (as defined above), or the eigenmode occupations if the Hamiltonian is quadratic.
The \emph{generalised Gibbs ensemble} (GGE) of the system for a given initial state $\rho(0) \in \Qst(\mcH)$ is then defined as
\begin{equation}
  \sigma_\mathrm{GGE} \coloneqq \arg\max_{\sigma\in\Qst(\mcH)}  \left\{ \Svn(\sigma) \oftype \forall A\in \mcC \suchthat \Tr(A\,\sigma) = \Tr(A\,\rho(0)) \right\}.
\end{equation}
\end{definition}
Note that the resulting state is always of the form
\begin{equation} \label{eq:ggestate}
  \sigma_\mathrm{GGE} \propto \e^{-\sum_{A \in \mcC} \beta_A\,A}
\end{equation}
with generalised inverse temperatures (Lagrange multipliers) $\beta_A \in \R$, one for each conserved quantity $A \in \mcC_{|\Vset|}$.
It is often clear what suitable sets of constants of motion are, for example in quenches to integrable systems.
In situations in which these constants of motion are ambiguous, the GGE inherits the same ambiguity \cite{Wouters2014,Mestyan2014,Pozsgay2014,Sels2014}.
As pointed out in the previous subsection, if all constants of motion are taken into account, then all finite dimensional systems equilibrate on average to the respective GGE if they equilibrate on average at all \cite{PhysRevLett.10-6}.

There is a large body of evidence that integrable quantum models indeed equilibrate to a suitable generalised Gibbs ensemble in this sense.
The works Refs.~\cite{Rigol07,Rigol08,PhysRevLett.10-5} noticed the significance of the GGE early on.
Refs.~\cite{Iucci2009,Iucci2010,1104.0154v1} discuss the behaviour of the one- and two-point correlation functions after a quench in various models,
and it is found that the relaxation dynamics and equilibrium values can be well understood by means of a GGE.
The validity of the generalised Gibbs ensemble is also studied in Ref.~\cite{Cassidy11}, where in particular, a quench of one-dimensional hard-core bosons in an optical lattice is considered, and in Ref.~\cite{Fagotti2014} for a spin-1/2 Heisenberg XXZ chain with an anisotropy for various initial states,
followed up by Ref.~\cite{deNardis2}. Ref.~\cite{Cazalilla11} develops a general method based on Wick's theorem that allows to show that the GGE correctly captures the equilibrium properties in certain quenches in solvable systems starting in thermal states, such as certain Ising chains, the Luttinger model, 1D hard-core bosons, and XX spin chains, and explains how quasi-particle occupations suffice to construct the GGE in these systems.
Ref.~\cite{Sotiriadis2014} follows up on this by considering situations with interacting pre-quench Hamiltonians. 
In contrast, Ref.~\cite{Kormos2014a} considers a quench from quadratic to infinitely strongly interacting (hard core) bosons and obtains exact results on the time evolution and shows that the equilibrium state is described by a GGE; Ref.~\cite{Mazza2014} follows up on this by analytical investigating how a trapping potential influences the equilibration dynamics and equilibrium properties.
Ref.~\cite{Kormos2013} analyses the GGE in the Lieb-Liniger model, finds potentially observable differences between the GGE and the grand canonical ensemble, and highlights that the GGE can turn out to be ill-defined if the initial expectation values of conserved quantities diverge.
Ref.~\cite{deNardis1} discusses an interaction quench in the Lieb-Liniger model
where the GGE implementation is not well defined and where also the idea of the representative state has been tested for the first time
in a truly interacting model where no use of the conserved charges was made.
Ref.~\cite{CauxEssler} is generally concerned with integrable models.
In this work, the concept of a ``representative Hamiltonian eigenstate'' 
is introduced and it is shown how to construct it efficiently by means of a generalised thermodynamic Bethe ansatz.
For long times, the equilibrium values of local observables after a quench are given by this ``representative Hamiltonian eigenstate''.
A framework for geometric quenches in integrable models based on the algebraic Bethe ansatz is developed in Ref.~\cite{GeometricCaux}.
Refs.~\cite{Pozsgay2014,Wouters2014,Goldstein2014a} study the crucial question of which constants of motion need to be included in the GGE to make it correctly reproduce the post quench equilibrium state.
For certain quenches in an XXZ model and attractive Lieb-Liniger model they find that even if all known local conserved quantities are included the GGE it still fails to reproduce the equilibrium expectation values of even some local observables. 
It is conceivable that this is an indication that the model might have more (quasi-)local conserved quantities.
Finding them all is a non-trivial task \cite{Prosen2013}.
Trying to identify the physically most relevant observables, Ref.~\cite{Sels2014} proposes to rather try to find the best possible approximation to the dephased state $\omega$ with an ansatz of GGE form \eqref{eq:ggestate} with as few observables as possible and exemplifies that this allows for example to capture the dephased state of a locally interacting fermionic system much better than with the standard GGE containing only mode occupations.

In Ref.~\cite{1109.5904v1} the ability of the GGE to capture scenarios with repeated quenches is explored.
Ref.~\cite{Fioretto2010} collects evidence that the GGE can correctly capture the long time limit of the expectation values of local operators in certain integrable models.
Ref.~\cite{Fagotti2012} investigates differences between the infinite time time averaged equilibrium state of a translation invariant finite system and the infinite time limit of the state of the corresponding infinite system and whether their properties are captured by the GGE and which role is played by local conserved quantities common to both the initial and final Hamiltonian of a quench. 
Ref.~\cite{Velenik2014} proposes generalised form factors for the analysis of correlation functions in generalised Gibbs ensembles.

\section{Typicality}
\label{sec:typicality}
We have up to now managed to avoid the introduction of \emph{ensembles}, or as one could say not \emph{put any probabilities by hand}.
However, ensembles and averages with respect to certain postulated probability distribution do play important roles in statistical mechanics.
In this section we review some arguments that can be used in the framework of pure state statistical mechanics to justify their use. These approaches to explain the applicability of statistical mechanics are based on the insight that under certain assumptions most individual instances of a situation lead to a behaviour that is very similarly to the average or \emph{typical} behaviour.

\subsection{Typicality for uniformly random state vectors}
We begin by reviewing the most influential articles on the subject in historic order, starting with the works of Schr{\"o}dinger \cite{Schroedinger1927} and von Neumann \cite{vonneumann1929}.
We will then state, prove, and discuss a general typicality theorem for uniformly random quantum state vectors.
We finish this section with a discussion of typicality in other ensembles and the most common objections against typicality arguments.

The strategy behind justifications for the use of ensembles is to argue that most states drawn according to some reasonable measure from a set of physically reasonable states have approximately the same properties, so that for computations it is practical to work with an average state.
This average state can, for example, turn out to be the state corresponding to a micro-canonical or canonical ensemble.

The use of such \emph{typicality arguments} in the foundations of quantum statistical mechanics has a long history.
First considerations along these lines already appear in a work by Schr{\"o}dinger \cite{Schroedinger1927} from 1927.
After an introduction into (first order) perturbation theory and a discussion of resonance phenomena in quantum mechanics with a focus on energy exchange in weakly interacting systems he goes on discussing what he calls a ``statistical hypothesis''\footnote{German original \cite{Schroedinger1927}: \foreignlanguage{ngerman}{``Statistische Hypothese''.}}.
He aims at describing the long time behaviour of weakly interacting systems hoping to find thermodynamic behaviour.
More specifically, he considers two systems that each have a pair of energy levels with the same gap.
The coupling between them that mixes the levels is assumed to be weak.
As his previous calculation had shown that the time averaged state depends on the initial state, he proposes to make an assumption about the initial energy level populations.
His assumption is that the populations of the levels are proportional to the products of the degrees of degeneracy of the non-interacting levels.
By introducing an entropy like quantity, he argues that if one of the systems is sufficiently large, this implies that when populations of energy levels whose reduced states on the small system are almost identical are combined, then the combined populations satisfy a canonical distribution.
By this, he effectively argues that initial states fulfiling his ``statistical hypothesis'' have reduced states on the small system that are well described by thermal states.

The concept of \emph{typicality} is even more prominent in an article by von Neumann \cite{vonneumann1929} from 1929.
His work has been translated by Tumulka \cite{Tumulka2010} and reviewed and refined by Goldstein, Lebowitz, Mastrodonato, Tumulka, and Zanghi \cite{0907.0108v1}.
Von Neumann sets out to clarify ``how it can be that the known thermodynamic methods of statistical mechanics enable one to make statements about imperfectly (e.g., only macroscopically) known systems that in most cases are correct.''%
\footnote{German original \cite{vonneumann1929}: \foreignlanguage{ngerman}{``[\dots] wie es kommt, dass die bekannten thermodynamischen Methoden der statistischen Mechanik es erm{\"o}glichen, 
{\"u}ber mangelhaft (d.h.\ nur makroskopisch) bekannte Systeme meistens richtige Aussagen zu machen.''}}
He goes on to describe that this means to clarify ``first, how the strange, seemingly irreversible behaviour of entropy emerges, and second, why the statistical properties of the (fictitious) micro-canonical ensemble can be assumed for the imperfectly known (real) systems, and that these questions will be tackled with the methods of quantum mechanics.''%
\footnote{German original \cite{vonneumann1929}: \foreignlanguage{ngerman}{``Insbesondere, wie erstens das 
eigent{\"u}mliche, irreversibel scheinende Verhalten der Entropy zustande kommt, und warum zweitens die statistischen Eigenschaften der (fiktiven) mikrokanonischen Gesamtheit f{\"u}r 
das mangelhaft bekannte (wirkliche) System unterstellt werden d{\"u}rfen.
Und zwar sollen diese Fragen mit den Mitteln der Quantenmechanik angegriffen werden.''}}
He further argues that the phase space of classical systems \cite{Kinchin1949}, a central object in Gibbs' formulation of classical statistical mechanics \cite{Gibbs1902}, should, in the context of quantum mechanics, be replaced by a system of mutually commuting macroscopic observables that approximate the true non-commuting quantum observables.
Each sequence of eigenvalues of all macroscopic observables is associated with a \emph{phase cell}, i.e., the subspace spanned by the state vectors that all give precisely these measurement outcomes for the macroscopic measurements, but which are macroscopically indistinguishable from each other.
Following Ref.~\cite{0907.0108v1}, we denote the projector onto the phase cell characterised by the sequence $\nu$ of macroscopic measurement outcomes by $P_\nu$.
The approximation of the microscopic observables is to be taken coarse enough, such that, for example, the commuting macroscopic position and momentum observables do not get in conflict with Heisenberg's uncertainty relation for the true microscopic position and momentum operators.
One of von Neumann's main results is his ``quantum ergodic theorem''\footnote{German original \cite{vonneumann1929}: \foreignlanguage{ngerman}{``Ergodensatz [...] in der neuen Mechanik''.}}.
Essentially, he is able to show the following (for details see the original article and Theorem~1 in Ref.~\cite{0907.0108v1}):
Fix the dimensions $\rank(P_\nu)$ of the phase cells, if they are all neither too small not too large, then for any fixed Hamiltonian without degeneracies and non-degenerate energy gaps (see also Section~\ref{sec:equlibrationintheweaksense}), most decompositions of the Hilbert space into phase cells with these dimensions have the property that, for all initial states and most times during the evolution, the evolving state of the system and a suitable micro-canonical state are approximately macroscopically indistinguishable.
This property is called ``normal typicality'' by the authors of Ref.~\cite{0907.0108v1}.
The result can actually be slightly generalised (Theorems~2 and 3 in Ref.~\cite{0907.0108v1}) and von Neumann's theorem can be reformulated into a statement about all initial states, all decompositions into phase cells, and most Hamiltonians \cite{0907.0108v1}.

It is worth noting that the notion of typicality in Refs.~\cite{vonneumann1929,0907.0108v1} concerns not the quantum state (vector) but the set of macroscopic observables.
The statement holds for most decompositions of the Hilbert space in phase cells (with certain properties), or most Hamiltonians, but for \emph{all} initial states.
In the following, typicality will mostly concern the quantum state (vector), i.e., we will encounter statements that hold for \emph{most} (initial) state vectors.

Typicality arguments feature prominently in the PhD thesis of Lloyd \cite{slloydthesis} (see also Ref.~\cite{lloyed13}).
Essentially he shows that for any fixed observable, if quantum state vectors are drawn uniformly at random from a subspace of a Hilbert space (we will soon make this more precise), then the mean square deviation of the expectation value of the observable in such a random state from that in the corresponding micro-canonical state is inverse proportional to the dimension of the subspace.

In a similar spirit, the concept of typicality is a cornerstone of the arguments in the book by Gemmer \cite{Gemmer09}.
As a measure of typicality the authors propose the \emph{Hilbert space variance} and derive bounds for the Hilbert space variance of various physically interesting quantities, ranging from expectation values of observables and distances of reduced states to entropies and purities.
As in Ref.~\cite{slloydthesis} and the present work, the aim is to use typicality to justify the methods of statistical mechanics and thermodynamics.

Many of the ideas of the works summarised above have later reappeared in Ref.~\cite{Goldstein06} in which the term \emph{canonical typicality} was coined.
Ref.~\cite{Goldstein06} is intended to be a clarification and extension of the work of Schr{\"o}dinger \cite{Schroedinger1927}, which we discussed earlier, and remarks in his book \cite{Schroedinger1952} on statistical thermodynamics.
After a translation of the classical proof of the canonical ensemble from the micro-canonical one to the quantum setting, the authors argue that the law of large numbers implies that if a state vector is drawn uniformly at random from a high dimensional subspace, its reduced state on a small subsystem will look similar to the reduced state of the micro-canonical state corresponding to that subspace.

Before we go on, we must say more precisely what we mean by drawing a state vector \emph{uniformly at random} from a subspace.
Intuitively it should mean that any state from the subspace is as probable as any other.
Mathematically this is made precise in the notion of \emph{left/right invariant measures} \cite{halmos}.
Haar's theorem \cite{Haar1933} implies that for any finite $d$ there is a unique left and right invariant, countably additive, normalised measure on the unitary group $U(d)$ \cite{halmos}.
We refer to this measure as the \emph{Haar measure} on $U(d)$ and denote it by $\muhaar[U(d)]$.
Left and right invariant means that for any unitary $U \in U(d)$ and any Borel set $\mathscr{B}\subseteq U(d)$
\begin{equation}
  \muhaar[U(d)](\mathscr{B}) = \muhaar[U(d)](U\,\mathscr{B}) = \muhaar[U(d)](\mathscr{B}\,U) ,
\end{equation}
where $U\,\mathscr{B}$ and $\mathscr{B}\,U$ are the left and right translates of $\mathscr{B}$.
In this sense, the Haar measure $\muhaar[U(d)]$ is the uniform measure on $U(d)$.

The Haar measure on the group of unitaries that map a (restricted) subspace $\mcH_R \subseteq \mcH$ of dimension $d_R$ into itself induces in a natural way a uniform measure $\muhaar[\mcH_R]$ on state vectors $\ket\psi\in\mcH_R$.
We call state vectors drawn according to this measure, and also pure quantum states $\ketbra \psi \psi$ drawn according to the natural induced measure, \emph{Haar random} and write $\ket\psi \sim \muhaar[\mcH_R]$.

A practical way to obtain state vectors distributed according to this measure is to fix a basis $(\ket{j})_{j=1}^{d_R}$ for the subspace $\mcH_R$ and then draw the real and imaginary part of $d_R$ complex numbers $(c_j)_{j=1}^{d_R}$ from normal distributions of mean zero and variance one.
The state vector
\begin{equation} \label{eq:normalisedstateintermsofcoefficients}
  \ket{\psi} = \frac{\sum_{j=1}^{d_R} c_j \ket{j}}{( {\sum_{j=1}^{d_R} |c_j|^2} )^{1/2}}
\end{equation}
is then distributed according to $\muhaar[\mcH_R]$, i.e., $\ket\psi \sim \muhaar[H_R]$ \cite{Zyczkowski2001}.
We denote the probability that an assertion $\mathbb{A}(\ket{\psi})$ about a state vector $\ket\psi$ is true if $\ket\psi \sim \muhaar[\mcH_R]$ by $\probability_{\ket\psi \sim \muhaar[\mcH_R]}(\mathbb{A}(\ket\psi))$.

In the framework of \emph{measure theory} \cite{halmos}, \emph{typicality} can be seen as a consequence of the phenomenon of \emph{measure concentration} \cite{ledoux01,CHATTERJEE07}.
In particular a result known as \emph{Levy's lemma}, has been used in Refs.~\cite{Popescu06,Popescu05} to obtain theorems in the spirit of Refs.~\cite{slloydthesis,Goldstein06}, but with stronger bounds on the probabilities to observe large deviations from the (micro)canonical ensemble.
Refs.~\cite{Popescu06,Popescu05} focused mainly on reduced states of small subsystems of states drawn at random from high dimensional subspaces.
Based on the same techniques, in Ref.~\cite{Gogolin10-masterthesis}, similar results have been obtained for the expectation values of individual observables on the full system as well as their variances, and for sets of commuting observables, developing further ideas of Ref.~\cite{vonneumann1929} concerning macroscopic measurements.

Furthermore, an extension to the distinguishability under a restricted set of POVMs is possible.
We summarise these results in a single theorem, which, however, is not optimal in terms of constants and scaling (compare Refs.~\cite{Popescu05,Gogolin10-masterthesis} for details).

\begin{theorem}[Measure concentration for quantum state vectors] \label{thm:measureconcentrationforquantumstatevectors}
  Let $R\subset\R$ and let $\mcH_R \subseteq \mcH$ be the subspace of the Hilbert space $\mcH$ of a system with Hamiltonian $\H \in \Obs(\mcH)$ that is spanned by the eigenstates of $\H$ to energies in $R$ and let $d_R \coloneqq \dim(\mcH_R)$.
  Then for every $\epsilon>0$ it holds that (i) for any operator $A \in \Bop(\mcH)$
  \begin{equation} \label{eq:measureconcentrationforobservables}
    \probability_{\ket\psi\sim\muhaar[\mcH_R]}\left( | \ex A {\ketbra\psi\psi} - \ex A {\rhomc[\H](R)} | \geq \epsilon \right) \leq 2\,\e^{-C\,d_R\,\epsilon^2/\|A\|_\infty^2} ,
  \end{equation}
  and (ii) for any set $\POVMs$ of POVMs
  \begin{equation} \label{eq:measureconcentrationforpovms}
    \probability_{\ket\psi\sim\muhaar[\mcH_R]}\left(\tracedistance[\POVMs]{\ketbra\psi\psi}{\rhomc[\H](R)} \geq \epsilon \right) \leq 2\,h(\POVMs)^2\,\e^{-C\,d_R\,\epsilon^2/h(\POVMs)^2} ,
  \end{equation}
  where $C = 1/(36\,\pi^3)$ and
  \begin{equation} \label{eq:definitionofhinthemeasureconcentrationtheorem}
    h(\POVMs) \coloneqq \min(|{\union \POVMs}|, \dim(\mcH_{\supp(\POVMs)})) .
  \end{equation}
\end{theorem}
\begin{proof}
  \texteqref{eq:measureconcentrationforobservables} is Theorem~2.2.2 from Ref.~\cite{Gogolin10-masterthesis}.
  We now prove \texteqref{eq:measureconcentrationforpovms} for $h(\POVMs)$ equal to the second argument of the $\min$ in \texteqref{eq:definitionofhinthemeasureconcentrationtheorem}.
  Let $S \coloneqq \bigcup_{M \in \union \POVMs} \supp(M)$ and remember that then for all $\rho,\sigma \in \Qst(\mcH)$
  \begin{equation}
    \tracedistance[\POVMs]{\rho}{\sigma} \leq \tracedistance{\smash{\rho^S}}{\smash{\sigma^S}} .
  \end{equation}  
  Then Eq.~(75) in Section~VI.C of Ref.~\cite{Popescu05} yields the result.
  To finish the proof, note that \texteqref{eq:distinguishabilityunderrestrictedsetsofpovms} implies that for any $\rho,\sigma \in \Qst(\mcH)$
  \begin{align} \label{eq:distinguishabilityupperboundforproofofmeasureconcentrationtheorem}
    \tracedistance[\POVMs]\rho\sigma &\coloneqq \sup_{M \in \POVMs} \frac{1}{2}\,\sum_{k=1}^{|M|} |\Tr(M_k\,\rho) - \Tr(M_k\,\sigma)| \\
    &\leq \frac{1}{2}\,\sum_{M \in \union\POVMs } |\Tr(M\,\rho) - \Tr(M\,\sigma)| \\
    &\leq \frac{1}{2}\,|{\union\POVMs}| \sup_{M \in \union\POVMs } |\ex M \rho - \ex M \sigma | .
  \end{align}
  Together with Boole's inequality this yields that for every $\sigma \in \Qst(\mcH)$
  \begin{align}
    &\probability_{\ket\psi\sim\muhaar[\mcH_R]}\left(\tracedistance[\POVMs]{\ketbra\psi\psi}\sigma \geq \epsilon\right) \nonumber \\
    \leq &1 - \probability_{\ket\psi\sim\muhaar[\mcH_R]}\left(\bigcap_{M \in \union\POVMs} |\ex M {\ketbra\psi\psi} - \ex M \sigma |< \frac{2\,\epsilon}{|{\union\POVMs}|} \right) \\
    = &\probability_{\ket\psi\sim\muhaar[\mcH_R]}\left(\bigcup_{M \in \union\POVMs} |\ex M {\ketbra\psi\psi} - \ex M \sigma| \geq \frac{2\,\epsilon}{|{\union\POVMs}|} \right) \\
    \leq &\sum_{M \in \union\POVMs} \probability_{\ket\psi\sim\muhaar[\mcH_R]}\left( |\ex M {\ketbra\psi\psi} - \ex M \sigma| \geq \frac{2\,\epsilon}{|{\union\POVMs}|} \right) .
  \end{align}
  The proof of the result for $h(\POVMs)$ equal to the first argument of the $\min$ in \texteqref{eq:definitionofhinthemeasureconcentrationtheorem} can then be finished by choosing $\sigma = \rhomc[\H](R)$, using \texteqref{eq:measureconcentrationforobservables}, and the fact that for all $M \in \union \POVMs$ it holds that $\norm[\infty]{M} \leq 1$.
  Disregarding a favorable factor of $2$ and using the (highly non-optimal) bound $|{\union\POVMs}| < |{\union\POVMs}|^2$ yields the unified result as stated in the theorem.
\end{proof}

A physically particularly relevant case is when $\supp(\POVMs)$ is contained in some small subsystem $S \supseteq \supp(\POVMs)$ and $R = [E,E+\Delta]$ is some energy interval.
Then the theorem yields a probabilistic bound on the distance $\tracedistance{\ketbra\psi\psi^S}{\rhomc^S[\H]([E,E+\Delta])}$.
If $\ket\psi\sim\muhaar[\mcH_R]$ and the dimension $d_R$ of the micro-canonical subspace $\mcH_R$ to the energies in the interval $[E,E+\Delta]$ fulfils $d_R \gg d_S$, then $\tracedistance{\ketbra\psi\psi^S}{\rhomc^S[\H]([E,E+\Delta])}$ is small with very high probability.
That is, the reduced state on $S$ of a random state from the subspace corresponding to the energy interval $R$ is indistinguishable from the reduction of the corresponding micro-canonical state, with high probability.

The same holds in the more general setting that one has access only to a sufficiently small number of measurements, which in total have a sufficiently small number of different outcomes.
If the total number of different outcomes $|\union \POVMs|$ is much smaller than the dimension of the subspace corresponding to the energy interval $[E,E+\Delta]$, a random state from this subspace is with high probability indistinguishable from the micro-canonical state.

For a family of Hamiltonians of locally interacting quantum systems with increasing system size, if $\Delta$ is kept fix and $E$ is chosen such that $R = [E,E+\Delta]$ is not too close to the boundaries of the spectrum of the Hamiltonian, then $d_R$ typically grows exponentially with the system size $|\Vset|$.
For a locally interacting system with a macroscopic number of constituents one would thus need to be able to distinguish an astronomically large number of different measurement outcomes to have a realistic chance of distinguishing a random state from a micro-canonical state.

Similar methods as those used above were employed in Ref.~\cite{Linden09} to prove that for Haar random pure states from high dimensional subspaces the effective dimension (which we encountered in Section~\ref{sec:equlibrationintheweaksense}) with respect to a fixed Hamiltonian is of the order of the dimension of the subspace, with probability exponentially close to one.
The result can be generalised to certain measures over states that are product with respect to a bipartition $\Vset = S \dunion B$ \cite{Gogolin10-masterthesis}.

\subsection{Typicality for other measures over quantum state vectors}

In addition to the Haar measure, other measures over quantum state vectors have been considered in the literature:
This has been done in order to incorporate meaningful physical constraints into notions of typicality.
Refs.~\cite{Lane,Bender05,Brody07,1003.4982,1104.4625v1,1102.3651v1} 
introduce the \emph{mean energy ensemble}.
Instead of the uniform measure on a subspace corresponding to some energy interval, the mean energy ensemble consists of random state vectors which have a fixed energy expectation value with respect to some given Hamiltonian $\H$.
Under certain conditions on the spectrum of $\H$ it can be shown that the mean energy ensemble exhibits measure concentration \cite{1003.4982}.
In addition to that, it is possible to identify the typical reduced state of states drawn from the mean energy ensemble \cite{1003.4982}, and it can be shown that under certain conditions states from the mean energy ensemble typically have a high effective dimension \cite{Gogolin10-masterthesis}. 

Ref.~\cite{Reimann07} considers an ensemble of quantum state vectors of the form given in \texteqref{eq:normalisedstateintermsofcoefficients}, in which the expansion coefficients $c_j = \braket{j}{\psi}$ have fixed modulus but random phases.
Concentration results, similar in spirit to Theorem~\ref{thm:measureconcentrationforquantumstatevectors}, can be shown for this ensemble that yield typicality whenever sufficiently many energy levels are populated.

Ref.~\cite{Bartsch09} extends the notion of typicality to the dynamics of systems.
Similarly as in the mean energy ensemble, the authors define an ensemble of initial states that share the same expectation value with respect to some given observable and then investigate the time evolution of this expectation value under a Hamiltonian.
The authors find \emph{dynamical typicality}, i.e., that states that initially give similar expectation values also typically lead to a similar dynamical evolution of these expectation values.

Typicality can also be used to speed up numerical calculations.
Instead of sampling over exponentially large sets of states, often drawing just a few representatives can already be sufficient to estimate expectation values \cite{Sugiura12}.
Ref.~\cite{White09} for example introduces the concept of \emph{minimally entangled typical quantum states}, which, given a Hamiltonian $\H$ and inverse temperature $\beta$, constitute an ensemble of pure states whose average corresponds to the thermal state of $\H$ at inverse temperature $\beta$.
The ensemble can be used to more efficiently calculate for example thermal expectation values of observables.
A related approach, which has recently been put forward in Refs.~\cite{Garnerone2010,Garnerone10-1,Garnerone2013,Garnerone2013a}, is to investigate and exploit typicality in the context of so-called \emph{matrix product states}.
The effects of typicality allow for the numerical approximation of thermal expectation values of observables in situations where naive approaches are infeasible \cite{Garnerone2013a}.
In Ref.~\cite{Steinigeweg2013} a method for numerically checking the validity of the \emph{eigenstate thermalisation hypothesis} (see Section~\ref{sec:thermalisationunderassumptionsontheeigenstates}) is proposed that exploits techniques to apply exponentials of operators to random pure states.
Typicality ensures that only few such random states are needed to obtain conclusive results, thereby vastly reducing the computational cost.

Typicality arguments are sometimes dismissed for being ``unphysical'' \cite{Bocchieri1958,Farquhar1957}.
Ref.~\cite{0907.0108v1}, for example, contains a very interesting review of the mostly negative reception of von Neumann's quantum ergodic theorem (see also Section~\ref{sec:typicality}).
Whether the concept of typicality is really superior to other approaches towards the foundations of statistical mechanics and thermodynamics, such as \emph{ergodicity}, the principle of \emph{maximum entropy}, or postulating \emph{ensembles}, is to some extent a matter of personal taste.
However, especially with respect to the latter, typicality has at last one important advantage:
Instead of simply postulating that a certain ensemble yields a reasonable description of a certain physical situation, typicality shows, in a mathematically very well-defined way, when and why details do not matter.
If most states anyway exhibit the same or very similar properties, then this does provide a heuristic, but pretty convincing, argument in favour of the applicability of ensembles.
It is hence an argument supporting a description of large systems with ensembles.

\section{Thermalisation}
\label{sec:thermalisation}
Given the findings presented in the last sections a natural question to ask is:
When do closed quantum systems in pure states that evolve unitarily not only equilibrate, but actually thermalise in the sense that under reasonable restrictions on the experimental capabilities they appear to be \emph{thermalised} or in \emph{thermodynamic equilibrium}?

To make this question meaningful we will define the term \emph{thermalisation} in this section.
Then, in Section~\ref{sec:thermalisationunderassumptionsontheeigenstates} and \ref{sec:thermalisationunderassumptionsontheinitialstate}, we will discuss two general complementary approaches to explain and understand thermalisation in the framework of pure state quantum statistical mechanics in detail.
The first approach is the so-called \emph{eigenstate thermalisation hypothesis} (ETH), the second is based on a quantum version of the classical derivation of the canonical ensemble from the micro-canonical one, augmented with rigorous perturbation theory.
The first approach is based mostly on assumptions on the eigenspaces/eigenstates of the Hamiltonian, while the second one instead requires stronger assumptions on the initial state.
We then turn to a discussion of thermalisation in locally interacting translation invariant systems and a result concerning the equivalence of the canonical- and micro-canonical ensemble in Section~\ref{sec:translationallyinvariant}.
It is possible to interpolate between ETH approach and that based on assumptions on the initial state to some extent.
We say more on that and on alternative notions of thermalisation in Section~\ref{sec:othernotionsandhybridapproaches}.
We finish by surveying numerical investigations of thermalisation and analytical results concerning concrete model Hamiltonians or specific classes of systems in Section~\ref{sec:numericalthermal}.

Throughout this section a focus will be put on \emph{subsystem thermalisation}, i.e., the thermalisation of a small part (subsystem) of a large composite quantum system via the interaction with the rest of the system (bath).
The whole composite system (subsystem and bath together) is thereby assumed to be in a pure state evolving according to the standard (Schr{\"o}dinger-)von-Neumann-equation under some Hamiltonian $\H$.
Let $S \subset \Vset$ be the vertex set of the subsystem and $B = \compl{S}$ that of the bath, then we will call the sum $\H_S + \H_B \eqqcolon \H_0$ of the two restricted Hamiltonians $\H_S$ and $\H_B$ the \emph{non-interacting Hamiltonian} and $\H_I \coloneqq \H - \H_0$ the \emph{interaction Hamiltonian}.
We will say that a Hamiltonian $\H$ is non-interacting if $\H = \H_S + \H_B$.

Whenever the term \emph{bath} is used in the following it refers to this model of thermalisation.
In particular we do not mean quantum systems that are already initially in a thermal state or other models of heat baths.
It is crucial to note that approaches that explain thermalisation in quantum systems by investigating the behaviour of systems coupled to such \emph{thermal baths} cannot solve the fundamental problem of thermalisation, as they leave open the question how the thermal bath became thermal in the first place.

\subsection{What is thermalisation?}
\label{sec:whatisthermalisation}
Whenever a term from one theory is used in a different context, a proper definition is mandatory.
This is particularly true for terms as involved as \emph{thermalisation} and \emph{thermodynamic equilibrium} which, already in classical statistical mechanics, have several different meanings depending on the context.
To take account of the complex nature of the term thermalisation we will not jump directly to a definition.
Instead, we will consider a number of conditions that each capture certain aspects of thermalisation and whose fulfilment, depending on the context, one might or might not find necessary to say that a system has thermalised.

The catalog of properties that we will consider has been chosen with the setting of subsystem thermalisation in mind.
Based on this discussion we will then carefully define what we consider sufficient to call a (sub)system thermalised, leaving open the possibility of defining other, possibly less strict, notions of thermalisation.
In addition to that, we will also define the term \emph{subsystem initial state independence}, a property that we regard as a \emph{necessary} prerequisite for the thermalisation of subsystems, and which we will discuss in more detail in Section~\ref{sec:absenceofthermalisation}.

The aspects of thermalisation that we will use as a guideline for our definition of thermalisation are:
\begin{enumerate}
\item \label{item:thermalisationconditions:equilibration} \emph{Equilibration}:
  Equilibration is generally considered to be a necessary condition for thermalisation.
  In the following we will mostly be concerned with \emph{subsystem equilibration on average} and \emph{apparent equilibration on average} of the whole system under restricted sets of POVMs (see also Section~\ref{sec:notionsofequilibration}).
\item \label{item:thermalisationconditions:subsysteminitialstateindependece} \emph{Subsystem initial state independence}:
  The equilibrium state of a small subsystem should be independent of the initial state of that subsystem.
  If a system exhibits some local exactly conserved quantities then one might still call it thermal and describe its equilibrium state by, for example, a so-called \emph{generalised Gibbs ensemble} \cite{Jaynes,PhysRev.108.17,Rigol07}.
  However, even such a behaviour is often already considered to be non-thermal.
  We will take the more cautious point of view that a system should not be considered thermalising if its equilibrium state depends on details of its own initial state, despite the 
  absence of local exactly conserved quantities.
\item \label{item:thermalisationconditions:bathstateindependence} \emph{Bath state independence}:
  It is generally expected that the equilibrium expectation values of local observables on a small subsystem are almost independent of the details of the initial state of the rest of the system, but should rather only depend on its ``macroscopic properties'', such as the energy density, which one would expect to have an influence on the temperature of the thermalising subsystem.
\item \label{item:thermalisationconditions:diafonalform} \emph{Diagonal form of the subsystem equilibrium state}:
  The equilibrium state of a small subsystem should be approximately diagonal in the energy eigenbasis of a suitably defined ``self-Hamiltonian''.
  If the interaction with the bath makes the state of the subsystem approximately diagonal in some basis one could call this \emph{decoherence}.
\item \label{item:thermalisationconditions:gibbsstate} \emph{Gibbs state}:
  Ultimately, one would like to recover the standard assumption of (classical) statistical physics that the equilibrium state is in some sense close to a Gibbs/thermal state.
\end{enumerate}

In the light of Condition~\refitem{item:thermalisationconditions:equilibration} it seems to be a sensible approach to define thermalisation \emph{on average} or \emph{during an interval} depending on the type of equilibration that goes along with thermalisation.
Conditions~\refitem{item:thermalisationconditions:subsysteminitialstateindependece} and \refitem{item:thermalisationconditions:bathstateindependence} make clear that thermalisation should be defined with respect to sets of initial states. This leads us to the following definition of thermalisation:

\begin{definition}[Thermalisation on average] \label{def:thermlaisationonaverage}
  We say that a system with Hilbert space $\mcH$ and Hamiltonian $\H \in \Obs(\mcH)$ \emph{thermalises} on average with respect to a set $\POVMs$ of POVMs and for a given set of initial states $\Qst_0 \subseteq \Qst(\mcH)$ if for each state $\rho(0) \in \Qst_0$, the system apparently equilibrates on average to an equilibrium state $\omega = \$_\H(\rho(0))$ that is close to a thermal state $\rhog[\tilde \H]\big(\beta(\Tr(\H\,\rho(0)))\big)$ for some Hamiltonian $\tilde \H$ in the sense that for some suitable function $\beta \oftype \R \to \R$ the distinguishability $\tracedistance[\POVMs]{\omega}{\rhog[\tilde \H]\big(\beta(\Tr(\H\,\rho(0)))\big)}$ is sufficiently small.
\end{definition}

Definition~\ref{def:thermlaisationonaverage} implicitly also defines \emph{thermalisation on average of subsystems}.
Just choose $\POVMs$ to be the set of all POVMs with support on a subsystem $S\subset\Vset$ and $\tilde \H = \H_S$.
If on the contrary $\POVMs$ contains POVMs whose support covers the whole system, then $\tilde \H = \H$ is a natural choice.
Moreover, in practice one would probably want that the function $\beta$ has some physically nice properties, like being smooth and or monotonically decreasing.
Thermalisation during intervals can be defined equivalently, but as we will not discuss it here, we keep the definition as simple as possible.

It seems worth emphasising again that the above definition does not say that a system thermalises if and only if the given conditions are met, but only says we call it thermalising if at least the given conditions are met.
It gives a set of conditions that are sufficient for thermalisation.
In addition it shall be noted that for the case of subsystem thermalisation with a small subsystem our definition of thermalisation implies that the equilibrium state of the subsystem must be nearly independent of the subsystems initial state. 
We discuss \emph{subsystem initial state independence} in more detail in Section~\ref{sec:violationofinitialstateindependence}.

An obvious question to ask now is:
What are reasonable sets $\Qst_0$ of initial states?
Particularly important is the \emph{energy distribution} of the initial states, i.e., the sequence $(p_k)_{k=1}^{d'}$ of the energy level populations $p_k \coloneqq \Tr(\Pi_k\,\rho(0))$, as it is conserved under time evolution.
Taking the classical derivation of the canonical ensemble from the micro-canonical one as a guideline, thermalisation can only be expected to happen for initial states whose energy distribution is not too wide, i.e., the energies of the significantly populated levels must be in an interval small compared to $\norm[\infty]\H$.
We will see in Sections~\ref{sec:thermalisationunderassumptionsontheeigenstates} and \ref{sec:thermalisationunderassumptionsontheinitialstate} that such a condition will play an important role in proofs of thermalisation.

In the above definition of thermalisation on average we deliberately left open the question of what ``sufficiently small'' means.
This is ultimately to be decided in the specific situation at hand.
One would probably want that $\tracedistance[\POVMs]{\omega}{\rhog[\tilde H](\beta)}$ somehow suitably decreases with the size of the system.
However, we want to have a definition of thermalisation that is applicable to finite systems.
Moreover, we want to avoid the technicalities of defining thermalisation for sequences of quantum systems of increasing size.

\subsection{Thermalisation under assumptions on the eigenstates}
\label{sec:thermalisationunderassumptionsontheeigenstates}
At the center of the first approach to show thermalisation in quantum systems is the \emph{eigenstate thermalisation hypothesis} (ETH).
There exist various version of the ETH in the literature and we will give a more precise definition below, but a minimal version of the ETH can informally be phrased as follows: ``A Hamiltonian fulfils the ETH if the expectation values of physically relevant observables in its energy eigenstates are approximately smooth functions of their energy.''
As we will see in this section, observables for which a system fulfils the ETH thermalise on average under reasonable conditions.
The ETH is usually said to date back to the two works \cite{PhysRevA.43.20,PhysRevE.50.88}.
As the role of these works is, however, often misunderstood it is worth starting this section with a short historical review:

Already in 1985 Ref.~\cite{Jensen1985} investigated numerically how relatively small quantum systems equilibrate to a state that can be well described by statistical mechanics.
The computational power available at that time made it possible to study a spin-1/2 Ising chain with up to seven sites in a transverse field by means of exact diagonalisation.
Ref.~\cite{Jensen1985} investigates the equilibration behaviour of both global and local observables and compares time averages with micro-canonical and canonical averages.
The authors conclude that ``both integrable and non-integrable quantum systems with as few as seven degrees of freedom can exhibit statistical behaviour for finite times.''
They also describe the reason for the statistical behaviour they observe, which is essentially the mechanism that is today known as the ETH: ``If the expectation values [of an observable in the energy eigenstates] are smooth functions of the energy [\dots], then the short-time average of the observable will be very close to the ensemble average.''
In fact, it seems fair to say that the authors did anticipate large parts of the recent debate on equilibration and thermalisation in closed quantum systems.
The last sentence of the abstract for example reads ``This work clarifies the impact of integrability and conservation laws on statistical behaviour. The relation to quantum chaos is also discussed.'' It is remarkable that Ref.~\cite{Jensen1985} is nevertheless essentially completely ignored by almost the whole recent literature centred around such questions (Refs.~\cite{Yukalov2011,Reimann2013} being notable exceptions).

In Ref.~\cite{PhysRevA.43.20} a mechanism that can lead to the thermalisation of quantum systems is identified, which the author calls \emph{eigenstate thermalisation}.
A quantum and a classical version of a hard sphere gas serve as prototypical examples to illustrate this mechanism.
A central role is played by \emph{Berry's conjecture}.
It states that in certain quantum systems, whose classical counterparts exhibit \emph{classical chaos}, the energy eigenstates to energies in the bulk of the spectrum are superpositions of plain waves with random phases and random Gaussian amplitudes \cite{Berry1977}.
It is argued that in the hard sphere gas, whose classical version is indeed \emph{chaotic}, all energy eigenstates that satisfy Berry's conjecture have a single particle momentum distribution that is thermal.
Finally, thermalisation is explained by the accumulation of relative phases between energy eigenstates due to time evolution.
This \emph{de-phasing} destroys any fine tuned setting of the phases that might have been present in the coherent superposition of energy eigenstates that made up the initial state.
Such a fine tuning is necessary to get an initial state that is out of equilibrium.

Ref.~\cite{PhysRevE.50.88} aims at providing a quantum mechanical justification for the applicability of statistical ensembles.
The main idea is to model interacting composite quantum systems by starting with a non-interacting Hamiltonian that can be well understood, and then modelling generic effects of the interactions by adding a small random Hamiltonian --- very much in the spirit of random matrix theory \cite{Bohigas1984,Tao2012,mehta90}.
Due to the fact that composite quantum systems generically have exponentially dense spectra, i.e., either exponentially small gaps between neighbouring eigenvalues and/or exponentially large degenerate eigenspaces, any extremely small perturbation will typically mix an exponentially large number of energy eigenstates of the non-interacting Hamiltonian.
This smears out their individual properties and should make the expectation values of physical observables in individual energy eigenstates of the perturbed Hamiltonian similar to those in a micro-canonical state with a similar mean energy.

A much more rigorous formulation of the idea behind eigenstate thermalisation can be found in Ref.~\cite{tasaki98} (see also Ref.~\cite{Tasaki97}).
This article considers bipartite systems with $\Vset = S \dunion B$, whose non-interacting part $\H_0 = \H_S + \H_B$ of the Hamiltonian $\H = \H_0 + \H_I$ is non-degenerate.
The interaction Hamiltonian $\H_I$ is assumed to couple only neighbouring energy levels, i.e., it is of the form
\begin{equation}
  \forall k\in[d]\itholds \bra{E^0_k} \H_I \ket{E^0_{k'}} = \lambda/2\,\delta_{|k-k'|,1} 
\end{equation}
for some $\lambda\in\R$ such that $\levelspaceing^{\max}_B \ll \lambda \ll \levelspaceing^{\min}_S$ with $\levelspaceing^{\max}_B$ the maximal spacing between the energy eigenvalues of $\H_B$ and $\levelspaceing^{\min}_S$ the minimal level spacing of $\H_S$.
It is first argued heuristically and then proved, under some additional technical assumptions, that such Hamiltonians indeed exhibit eigenstate thermalisation in the sense that for most $k$ and all observables $A_S$ with $\supp(A_S) \subseteq S$ it holds that $\bra{E_k} A_S \ket{E_k} \approx \Tr(A_S\, g[H_S](\beta(E_k)))$ (see Eq.~(5) and (6) in Ref.~\cite{Tasaki97}).

The \emph{eigenstate thermalisation hypothesis} (ETH) gained wide popularity after the very influential article Ref.~\cite{Rigol08}, which states the ETH as follows:
\begin{conjecture}[Eigenstate thermalisation hypothesis as stated in Ref.~\cite{Rigol08}] \label{conjecture:ethrigolform}
  The expectation value $\bra{E_k} A \ket{E_k}$ of a few-body observable $A$ in an eigenvector 
  $\ket E_k$ of the Hamiltonian, with energy $E_k$, of a large interacting many-body system equals the thermal [\dots] average of $A$ at the mean energy $E_k$.
\end{conjecture} 

It is emphasised that \emph{thermal average} in this context can also mean the micro-canonical average.
Ref.~\cite{Rigol08} studies a system of hard core bosons on a lattice.
It is demonstrated that the observed thermalisation can be explained by the fact that certain physically relevant observables have expectation values in most energy eigenstates that indeed resemble those in a micro-canonical state.
The validity of numerous variants of the ETH has been studied extensively and in great detail at hand of many physically relevant models.
This will be detailed in Section~\ref{sec:numericalthermal}.

A slightly generalised and sharpened version of the ETH that captures the spirit of \emph{eigenstate thermalisation} and applies to degenerate Hamiltonians is the following:
\begin{definition}[Eigenstate thermalisation hypothesis (ETH)] \label{def:eth}
  A Hamiltonian $\H$ fulfils the \emph{eigenstate thermalisation hypothesis} in a set $R \subset \R$ of energies with respect to a set $\POVMs$ of POVMs if and only if all its spectral projectors $\Pi_k$ to energies $E_k \in R$ have the property that there is a sufficiently smooth function $\beta\colon R \to \R^+$ such that for each $k$ with $E_k \in R$ it holds that for all normalised pure states $\psi \in \Qst(\mcH)$ with the property $\psi \leq \Pi_K$ the distinguishability $\tracedistance[\POVMs]{\psi}{\rhog[\H](\beta(E_k))}$ is sufficiently small.
\end{definition}
Again, we have deliberately left open what is meant by ``sufficiently smooth'' and ``sufficiently small''.

It is still open under which precise conditions the ETH holds in this or a similar form.
The rigorous derivations of Ref.~\cite{tasaki98} have so far not been generalised to more reasonable physical interactions.
Methods to analytically check the ETH in ``non-integrable models'' that are interesting in the context of condensed matter theory currently seem to be out of reach.
Very recently in Ref.~\cite{Mueller2013} a statement reminiscent of the ETH was proved under fairly general conditions.
More precisely, Ref.~\cite{Mueller2013} shows \emph{weak local diagonality} (Theorems~4 and 38) of the energy eigenstates of a certain type of Hamiltonian.
In the language used here a slightly simplified version of this statement can be formulated as follows:
\begin{theorem}[Weak local diagonality \cite{Mueller2013}]
  Consider a locally interacting spin system with Hilbert space $\mcH$ and Hamiltonian $\H \in \Obs(\mcH)$ whose interaction graph $(\Vset,\Eset)$ is a hypercubic lattice of spatial dimension $D$ and let $S \subset B \subset \Vset$ be subsystems.
  Then there exist constants $C,c,v > 0$, which depend only on $D$ and the local interaction strength $J \coloneqq \max_{X\in\Eset} \norm[\infty]{\H_X}$ of the Hamiltonian such that for any energy eigenstate $\ket E$ of $\H$ to energy $E$ there exists a state $\rho^B_E \in \Qst(\mcH_B)$ that satisfies for any two energy eigenstates $\ket{E^B_l}, \ket{E^B_m}$ of $\trunc {\H_B} B$ with energies $E^B_l$ and $E^B_m$
  \begin{equation}
    |\bra{E^B_l} \rho^B_E \ket{E^B_m}| \leq \e^{-\dist(S,\compl B)\,(E^B_l-E^B_m)^2/(8\,c\,v^2)}
  \end{equation}
  and at the same time
  \begin{equation}
    \norm[1]{\Tr_{B \setminus S}(\rho^B_E) - \Tr_{\Vset \setminus S}(\ketbra{E}{E}) } \leq C\,A^2\,J\,
    \left({\frac{\dist(S,\compl B)}{4\,c\,v^2}}\right)^{1/2}\,\e^{-c\,\dist(S,\compl B)/2} .
  \end{equation}
\end{theorem}
Essentially the theorem tells us that if $S$ is sufficiently far from the boundary of $B$, then for each energy eigenvector
$\ket{E}$ of $\H$ there exists a state in $\Qst(\mcH_B)$ that is both approximately diagonal in the eigenbasis of $\trunc {\H_B} B$ and locally on $S$ hard to distinguish from $\ketbra{E}{E}$.
If one could improve the result to the effect that it would show local indistinguishability not only from an approximately diagonal state but from a thermal state then it would constitute a proof of an ETH like statement.
However, such a generalisation can almost surely hold only under additional assumptions \cite{Mueller2013}.

The ETH, as defined in Definition~\ref{def:eth}, is sufficient for thermalisation in the following sense:
\begin{observation}[Thermalisation in systems that fulfil the ETH] \label{obs:ethissufficientforthermalisation}
  Systems whose Hamiltonian $\H \in \Obs(\mcH)$ fulfils the ETH, as stated in Definition~\ref{def:eth}, for a set $R \subset \R$ of energies with respect to a set $\POVMs$ of POVMs, thermalise on average with respect to the set $\POVMs$, in the sense of Definition~\ref{def:thermlaisationonaverage}, for all initial states for which the system apparently equilibrates on average with respect the restricted set $\POVMs$ of POVMs (see also Section~\ref{sec:notionsofequilibration}) and whose energy distribution is sufficiently narrow and contained in $R$, i.e., $E_k \notin R \implies \Tr(\Pi_k\,\rho(0)) = 0$.
\end{observation}
The fact that the ETH is sufficient for thermalisation in this or a similar sense is widely known (see, for example, Refs.~\cite{tasaki98,Sirker2013}).
It is worth noting that the strong requirement in Definition~\ref{def:eth} that the distinguishability $\tracedistance[\POVMs]{\psi}{\rhog[\H](\beta(E_k))}$ must be small for all normalised pure states $\psi \leq \Pi_K$ is crucial for the above observation to hold.
At the same time, this requirement obviously becomes harder to satisfy the more degenerate the Hamiltonian is.

If one takes the point of view that one should say that a system thermalises only if it thermalises in the sense of Definition~\ref{def:thermlaisationonaverage} for \emph{all} equilibrating initial states with a sufficiently narrow energy distribution, then fulfilment of the ETH is at the same time essentially also necessary for thermalisation.
We will not make this statement fully rigorous, but the intuition behind it is as follows:
If the ETH is not fulfilled, there should always exist initial states with a narrow energy distribution that only have overlap with energy levels that, for certain observables or POVMs, produce a measurement statistic that is sufficiently far from that of the closest thermal state. This distinguishability from the thermal state will then still be visible in the de-phased state and hence survive de-phasing and equilibration.

Such arguments, and the above mentioned apparent connection between the ETH and \mbox{(non-)}integrability, has lead some authors to proclaim \cite{Kollath07,Cramer2008,Flesch08,Rigol08,Rigol09,Banuls10} that non-integrable systems thermalise and integrable systems do not.
While there is evidence that in many models this is indeed the case, we will see in Section~\ref{sec:absenceofthermalisation} and \ref{sec:integrability} that the situation is in fact more involved.
We will given an overview of the numerical and analytical literature on thermalisation in the context of the ETH in Section~\ref{sec:numericalthermal}. 

We have seen that the ETH as defined in Definition~\ref{def:eth} is by construction essentially sufficient and, in a certain sense, necessary for thermalisation.
The necessary part, however, only holds if one is willing to call a system thermalising only if it thermalises for a given set of POVMs for \emph{all} initial states with a sufficiently narrow energy distribution for which it also apparently equilibrates.
Hence, there is the possibility to show thermalisation in systems that do not fulfil the ETH, if one is willing to restrict the class of allowed initial states.
As we will see in the following this can indeed be done.

\subsection{Subsystem thermalisation under assumptions on the initial state}
\label{sec:thermalisationunderassumptionsontheinitialstate}
In this section we discuss a second approach towards the problem of thermalisation that is independent of the eigenstate thermalisation hypothesis (ETH).
Instead of making strong assumptions concerning the properties of the energy eigenstates of the Hamiltonian we will show thermalisation under stronger assumptions concerning the energy distribution of the initial state.
This alternative and complementing approach is inspired by an argument from classical statistical mechanics, which we will lift to the quantum setting.
The details of this approach were first worked out in Ref.~\cite{Riera2012}.

The first motivation for this work comes from the fact that explaining thermalisation by using the eigenstate thermalisation hypothesis has one important drawback --- that the ETH is indeed a \emph{hypothesis}.
One could make the provocative claim that this leads to the ironic situation that attempts to explain thermalisation by the ETH have the following problem:
They essentially try to explain one phenomenon that is not well understood by another one that is almost as little understood \cite{Singh}.

The second motivation comes from the consideration that demanding thermalisation of \emph{all} initial states with an energy distribution that is only required to be narrow but otherwise allowed to have arbitrary complex structure is asking for too much.

In the light of typicality arguments (Section~\ref{sec:typicality}) it seems plausible to restrict the class of initial states for which one tries to show thermalisation, or even to be content with an argument that shows thermalisation for most states from some measure.
In addition, certain restrictions on the initial states are anyway already necessary to prove equilibration on average in the first place (Section~\ref{sec:equlibrationintheweaksense}), and practical limits on the experimental capabilities can be used to argue that many initial states of macroscopic objects are essentially impossible to prepare \cite{Reimann12,Reimann2012,Reimann08}.

The third motivation comes from the known fact that in some systems the ETH is not fulfilled and this has been linked to the \emph{integrability} of these models, while \emph{non-integrability} is often associated with a fulfilment of the ETH and thermalisation (see for example Refs.~\cite{Rigol07,Rigol08,Rigol11,Larson13,Polkovnikov11,Cassidy11,Gritsev10,Fioretto2010,1006.1634v1,Cazalilla11,1103.0787v1}).
What \emph{\mbox{(non-)}integrability} even means in the context of quantum mechanics is, however, 
far from being settled \cite{1012.3587v1,Benet2003} (see also Section~\ref{sec:integrability}).
It is thus of interest to approach the problem of thermalisation in a way that is independent of the concept of integrability.

As we will see in the following, restricting the class of initial states makes it possible to rigorously prove subsystem thermalisation on average without any reference to the ETH for both spin and fermionic systems.
The overall structure of the argument is depicted in Figure~\ref{fig:structureofthermalisationargument}.
The result that we will derive and discuss in this section can be combined with either the typicality theorems from Section~\ref{sec:typicality} or the dynamical equilibration theorems from Section~\ref{sec:equilibration}.
The former yields a kinematic thermalisation statement (Observation~\ref{obs:thermalisationonofrandomstaates}) that holds for most Haar random states from a certain subspace.
The latter yields a dynamic thermalisation result (Observation~\ref{obs:thermalisationonaverage}) that proves thermalisation on average in the sense of Definition~\ref{def:thermlaisationonaverage} for all initial states from a certain class of states.
It is hence closer to the thermalisation statement obtained under the ETH (Observation~\ref{obs:ethissufficientforthermalisation}), which we discussed in the last section.

In essence, the proofs of the statements presented in this section are translations of the classical derivation of the canonical ensemble for small subsystems of large weakly interacting systems that are described by a micro-canonical ensemble to the quantum setting.
The main difficulty is that in quantum mechanics the interaction between the small subsystem and the bath not only shifts the eigenenergies of the non-interacting Hamiltonian, but, in addition, significantly perturbs the energy eigenstates.
In many previous accounts of the thermalisation problem this issue has been partially overlooked or at least not been addressed rigorously.
Compare for example Refs.~\cite{Popescu05,Popescu06,Goldstein06}.

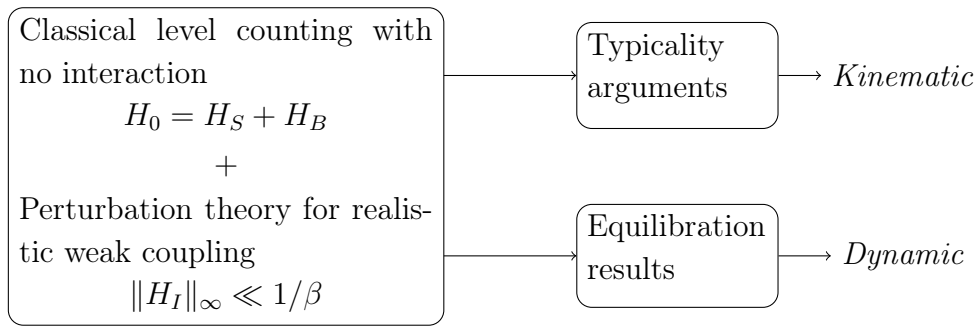
\begin{figure}[bt]
  \centering
  \begin{tikzpicture}
    \node (plus) {$+$};
    \node[above of=plus,node distance=1.2cm] (goldstein) {
      \begin{minipage}{5.5cm}
          \begin{block}{}
            Classical level counting with no interaction\\\centerline{$\H_0 = \H_S + \H_B$}
          \end{block}
        \end{minipage}
      };
    \node[below of=plus,node distance=1.2cm] (perturbation) {
        \begin{minipage}{5.5cm}
          \begin{block}{}
            Perturbation theory for realistic weak coupling\\\centerline{$\|\H_I\|_\infty \ll 1/\beta$}
          \end{block}
        \end{minipage}
      };
    \coordinate (goldsteinne) at (goldstein.north east);
    \draw[rounded corners=0.2cm] (goldstein.north east) rectangle (perturbation.south west);
      \node[right of=goldstein,node distance=6cm,draw,rounded corners=0.2cm] (typicality) {
        \begin{minipage}{2.4cm}
          \begin{block}{}
            Typicality\\arguments
          \end{block}
        \end{minipage}
      };
      \draw[->] (goldsteinne |- goldstein) -- (typicality);
      \node[right of=typicality,node distance=3cm] (kinematic) {\emph{Kinematic}};
      \draw[->]  (typicality) -- (kinematic);
      \node[right of=perturbation,node distance=6cm,draw,rounded corners=0.2cm] (equilibration) {
        \begin{minipage}{2.4cm}
          \begin{block}{}
            Equilibration results
          \end{block}
        \end{minipage}
      };
      \draw[->] (goldsteinne |- perturbation) -- (equilibration);
      \node[right of=equilibration,node distance=3cm] (dynamic) {\emph{Dynamic}};
      \draw[->]  (equilibration) -- (dynamic);
    \end{tikzpicture}
  \caption{(Reproduction from Ref.~\cite{Gogolin2014}) Structure of the proof of thermalisation from Ref.~\cite{Riera2012}.}
  \label{fig:structureofthermalisationargument}
\end{figure}

How does the interaction influence the Hamiltonian?
The eigenvalues of the interacting Hamiltonian are shifted at most by the operator norm of the interaction Hamiltonian with respect to those of the non-interacting Hamiltonian \cite[Theorem III.2.1]{bhatia}.
As long as the interaction is weak, in the sense that its operator norm is small compared to an energy uncertainty or measurement resolution, the change in the eigenvalues will thus be insignificant.

The energy eigenstates, or in the case of a degenerate Hamiltonian the spectral projectors, are much more fragile.
Naive perturbation theory breaks down \cite{Sakurai1995} as soon as the strength of the perturbation is larger than the gaps of the non-interacting Hamiltonian.
The gaps of a locally interacting quantum system are, however, usually exponentially small in the system size.
Indeed, if the non-interacting Hamiltonian $\H_0$ and the interaction Hamiltonian $\H_I$ are not diagonal in the same basis, the energy eigenstates of $\H = \H_0 + \H_I$ will usually be markedly different from those of $\H_0$.

Before we tackle this problem, let us consider the non-interacting case, i.e., a Hamiltonian of the form $\H_0 \coloneqq \H_S + \H_B$. Let $\H_0$ and $\trunc {\H_S} S$ have spectral decompositions $\H_0 = \sum_k^{d'_0} E^0_k\,\Pi^0_k$ and $\trunc {\H_S} S = \sum_l^{d_S'} E^S_l\,\Pi^S_l$, respectively.
Moreover, let $(\ket{\tilde E^S_l})_{l=1}^{d_S}$ and $(\ket{\tilde E^B_m})_{m=1}^{d_B}$ be some orthonormal eigenbases with corresponding eigenvalues $(\tilde E^S_l)_{l=1}^{d_S}$ and $(\tilde E^B_m)_{m=1}^{d_B}$ of $\trunc {\H_S} S$ and $\trunc {\H_B} B$, respectively.
The Hamiltonians $\trunc {H_S} S$, $\trunc {H_B} B$, and $\H_0$ are allowed to have degeneracies, i.e., $l\neq l' \notimplies \tilde E^S_l \neq \tilde E^S_{l'}$ and  $m\neq m' \notimplies \tilde E^B_m \neq  \tilde E^B_{m'}$ and the bases are not unique.
Remember that, on the other hand, by definition, the elements of the sequences $(E^0_k)_{k=1}^{d'_0}$ and $(E^S_l)_{l=1}^{d'}$ are distinct.

We first look at the case of spin systems.
In such systems each of the spectral projectors $\Pi^0_k$ of $\H_0$ is of the form 
\begin{equation}
  \Pi^0_k = \sum_{l,m\suchthat \tilde E^S_l+\tilde E^B_m=E^0_k} \ketbra{\tilde E^S_l}{\tilde E^S_l} \otimes \ketbra{\tilde E^B_m}{\tilde E^B_m} . \label{eq:energyprojectorsofnoninteractinghamiltonianinspinsystems}
\end{equation}
The micro-canonical state $\rhomc[\H_0]([E,E+\Delta])$ to an energy interval $[E,E+\Delta]$ is hence proportional to
\begin{align}
  \rhomc[\H_0]([E,E+\Delta]) &\propto \sum_{k\suchthat E^0_k \in [E,E+\Delta]}\ \sum_{l,m\suchthat \tilde E^S_l+\tilde E^B_m=E^0_k} \ketbra{\tilde E^S_l}{\tilde E^S_l} \otimes \ketbra{\tilde E^B_m}{\tilde E^B_m} . \label{eq:firstequalityinthecountingargumentforspins}\\
  \intertext{Its reduced state $\rhomc^S[\H_0]([E,E+\Delta]) = \Tr_B \rhomc[\H_0]([E,E+\Delta])$ on $S$ therefore satisfies}
  \rhomc^S[\H_0]([E,E+\Delta]) &\propto \sum_{k\suchthat E^0_k \in [E,E+\Delta]}\ \sum_{l,m\suchthat \tilde E^S_l+\tilde  E^B_m=E^0_k} \ketbra{\tilde E^S_l}{\tilde E^S_l} \\
  &= \sum_{k\suchthat E^0_k \in [E,E+\Delta]}\, \sum_{l=1}^{d_S} \ketbra{\tilde E^S_l}{\tilde E^S_l} \, |\{m\oftype \tilde E^S_l+\tilde E^B_m=E^0_k \}| \\
  &= \sum_{k\suchthat E^0_k \in [E,E+\Delta]}\, \sum_{l=1}^{d_S'} \Pi^S_l \, |\{m\oftype E^S_l+\tilde E^B_m=E^0_k \}| \\
  &= \sum_{l=1}^{d_S'} \Pi^S_l \, |\{m\oftype E^S_l+\tilde E^B_m \in [E,E+\Delta] \}| \\
  &= \sum_{l=1}^{d_S'} \Pi^S_l \, \#_\Delta[\trunc{\H_B}B](E-E^S_l) ,
\end{align}
where 
\begin{equation}
  \#_\Delta[\trunc{\H_B}B](E) \coloneqq |\{m\oftype \tilde E^B_m \in [E,E+\Delta] \}| = \rank(\rhomc[\trunc{\H_B}B]([E,E+\Delta]))
\end{equation}
is the \emph{number of orthonormal energy eigenstates} of the bath Hamiltonian $\H_B$ to energies in the interval $[E,E+\Delta]$.

For systems of fermions \texteqref{eq:energyprojectorsofnoninteractinghamiltonianinspinsystems} does not hold, because the Hilbert space of the joint system is not the tensor product of the Hilbert spaces of the subsystems.
However, the following quite lengthy calculation shows an equivalent result also for fermionic systems.
Readers not interested in the details can safely jump directly to Observation~\ref{obs:gibbsstatesasreductionsofmicrocanonicalstates}.

Denote by $f_x,f\ad_x$ the fermionic annihilation and creation operators on $\mcH$ and by $\tilde f_x, \tilde f\ad_x$ with $x \in S$ those acting on $\mcH_S$ and for $x \in B$ those acting on $\mcH_B$.
Furthermore, denote the vacuum state vector
in $\mcH$ by $\ket{0}$ and the projectors in $\Bop(\mcH)$ onto the subspace with no particle in system $S$ or $B$ by $\ketbra{0}{0}_S$, and $\ketbra{0}{0}_B$, respectively.
The projectors $\ketbra{0}{0}_S$, $\ketbra{0}{0}_B$, and $\ketbra{0}{0}$ are all even operators and $\ketbra{0}{0} = \ketbra{0}{0}_S\,\ketbra{0}{0}_B$.
For each $l \in [d_S]$ let $p^{\H_S}_l$ be the representation of the eigenstate $\ket{\tilde E^S_l}$ as a polynomial in the fermionic operators on $\mcH_S$, i.e., $\ket{\tilde E^S_l} = p^{\H_S}_l((\tilde f_s,\tilde f\ad_s)_{s\in S})\,\ket{0}_S$, and likewise for $p^{\H_B}_m$.
Note that the $p^{\H_S}_l$ and the $p^{\H_B}_m$ are either even or odd polynomials as otherwise the projectors $\ketbra{\tilde E^S_l}{\tilde E^S_l}$ and $\ketbra{\tilde E^B_m}{\tilde E^B_m}$ would not be even.
Furthermore, note that commuting two polynomials that are both either even or odd gives a global minus sign only if both polynomials are odd.
As $\H_S$ and $\H_B$ are even operators it is straight forward to verify that the states $\ket{\tilde E^S_l + \tilde E^B_m} \coloneqq p^{\H_S}_l((f_s,f\ad_s)_{s\in S})\,p^{\H_B}_m((f_b,f\ad_b)_{b\in B})\,\ket{0}$ are eigenstates of $\H^0$ to energy $\tilde E^S_l + \tilde E^B_m$.
In fact, they form an orthonormal basis of $\mcH$ in which $\H^0$, $\H_S$, and $\H_B$ are jointly diagonal.
For the sake of brevity we omit the subscripts $_{s\in S}$ and $_{b\in B}$ in the following calculation.
It is again straight forward to verify that for any even operator $A \in \Bop(\mcH)$ with $\supp(A) \subseteq S$ it holds that
\begin{align}
  &\Tr\big(A\,\ketbra{\tilde E^S_l + \tilde E^B_m}{\tilde E^S_l + \tilde E^B_m}\big) \nonumber \\
  = &\Tr\big(A\, p^{\H_S}_l((f_s,f\ad_s))\,p^{\H_B}_m((f_b,f\ad_b))\,\ketbra{0}{0}_S\,\ketbra{0}{0}_B\,p^{\H_B}_m((f_b,f\ad_b))\ad\,p^{\H_S}_l((f_s,f_s))\ad \big) \\
  = &\Tr\big(A\, p^{\H_S}_l((f_s,f\ad_s))\,\ketbra{0}{0}_S\,p^{\H_S}_l((f_s,f\ad_s))\ad\,p^{\H_B}_m((f_b,f\ad_b))\,\ketbra{0}{0}_B\,p^{\H_B}_m((f_b,f\ad_b))\ad \big) \\ 
  = &\Tr\big(A\,\ketbra{\tilde E^S_l}{\tilde E^S_l} \big) .
\end{align}
The last step can be shown by explicitly writing out the trace in the Fock basis and inserting an identity between the operators that are supported on $S$ and those supported on $B$.

Now, note that any operator $A \in \Bop(\mcH)$ with $\supp(A) \subseteq S$ can be written as a sum of an even and odd part and that only the even part can contribute to an expectation value of the form $\Tr(A\,\ketbra{\tilde E^S_l + \tilde E^B_m}{\tilde E^S_l + \tilde E^B_m})$.
The above calculation is hence sufficient to show that (remember the definition of the partial trace in \texteqref{eq:partialtrace})
\begin{equation}
  \forall l\in[d_S],m\in[d_B]\itholds \Tr_B(\ketbra{\tilde E^S_l + \tilde E^B_m}{\tilde E^S_l + \tilde E^B_m}) = \ketbra{\tilde E^S_l}{\tilde E^S_l} .
\end{equation}
Finally, realizing that
\begin{equation}
  \begin{split}
    \rhomc[\H_0]&([E,E+\Delta])\\
    &= \sum_{k\suchthat E^0_k \in [E,E+\Delta]}\ \sum_{l,m\suchthat \tilde E^S_l+\tilde E^B_m=E^0_k} \Tr_B(\ketbra{\tilde E^S_l + \tilde E^B_m}{\tilde E^S_l + \tilde E^B_m}) 
  \end{split}
\end{equation}
yields an expression equivalent to \texteqref{eq:firstequalityinthecountingargumentforspins} and the proof then proceeds analogously.
We summarise the result of the above calculation in the following observation:
\begin{observation}[Gibbs states as reductions of micro-canonical states of the non-interacting Hamiltonians] \label{obs:gibbsstatesasreductionsofmicrocanonicalstates}
  Let $[E,E+\Delta]$ be an energy interval and $\H_0 = \H_S + \H_B$ a non-interacting Hamiltonian of a bipartite quantum system of spins or fermions with $\Vset = S \dunion B$.
If for some $\beta \in \R$ it holds that
  \begin{equation} \label{eq:expoenntialdensityofstates}
    \#_\Delta[\trunc{\H_B}B](E) \propto \e^{-\beta\,E} ,
  \end{equation}
  then $\rhomc^S[\H_0]([E,E+\Delta])$ takes the well known form of a thermal state, i.e.,
  \begin{equation} \label{eq:noninteractingreducedmicrocanonicalstateisequaltocanonicalstate}
    \rhomc^S[\H_0]([E,E+\Delta]) \propto \sum_{l=1}^{d_S'} \Pi^S_l \, \e^{-\beta\,E^S_l} \propto \rhog[\trunc{\H_S}S](\beta) = \rhog^S[\H_0](\beta).
  \end{equation}
\end{observation}

Note how $\beta$, which was introduced in \texteqref{eq:expoenntialdensityofstates} simply as a parameter describing the shape of the number of states, ends up being the inverse temperature of the thermal state $\rhog[\trunc{\H_S}S](\beta)$ of the subsystem $S$.
Similar calculations (at least for spin systems) can be found for example in Refs.~\cite{Goldstein06,tasaki98,Reimann07,Gemmer09} and in many textbooks on statistical mechanics.

For finite dimensional baths the proportionality $\#_\Delta[\trunc{\H_B}B](E) \propto \e^{-\beta\,E}$ can never be exactly fulfilled simply because $\#_\Delta[\trunc{\H_B}B](E)$ is not continuous.
A detailed analysis \cite[Appendix A]{Riera2012} shows that if the logarithm of the number of states $\ln(\#_\Delta[\trunc{\H_B}B](E))$ can be sufficiently well approximated by a twice differentiable function whose second derivative is small compared to the width of the relevant energy range $[E-\norm[\infty]{\H_S},E+\norm[\infty]{\H_S}]$, then \texteqref{eq:noninteractingreducedmicrocanonicalstateisequaltocanonicalstate} is fulfilled approximately.
The first derivative of this approximation ends up being the inverse temperature of the thermal state, the second derivative enters the error bound.

It is widely known that natural locally interacting Hamiltonians $\H$ with bounded local terms ``generically'' have an approximately Gaussian number of states $\#_\Delta[\H](E)$ if the system size is sufficiently large \cite[Section 12.2]{Gemmer09} (see also Ref.~\cite{Hartmann2005,Keating} for some rigorous results).
It is more common to refer to the \emph{density of states} in this case, which is essentially the limit of $\#_\Delta[\H](E)/\Delta$ for $\Delta$ small and increasing system size.
If the bath Hamiltonian $\H_B$ is taken to be such a model with a nearly Gaussian density and number of states, the approximation by a twice differentiable function is possible and the distance $\tracedistance{\rhomc^S[\H_0]([E,E+\Delta])}{\rhog[\trunc{\H_S}S](\beta)}$ can be bounded \cite[Appendix B]{Riera2012} and is usually exponentially small in the size of the bath.
In the following we will call locally interacting systems that have this property ``generic''.

The value of $\beta$ for which $\tracedistance{\rhomc^S[\H_0]([E,E+\Delta])}{\rhog[\trunc{\H_S}S](\beta)}$ is small depends on $E$.
If $\#_\Delta[\trunc{\H_B}B]$ is indeed close to a Gaussian, then $\ln(\#_\Delta[\trunc{\H_B}B])$ can be well approximated by an inverted parabola.
Its first derivative, which is essentially the optimal $\beta$, is large for low values of $E$, thus associating them with low temperatures.
For values of $E$ in the center of the spectrum it goes to zero, corresponding to infinite temperature, and becomes negative for even higher values of $E$.

In conclusion, we can say that the reduction on $S$ of a micro-canonical state to an energy interval $[E,E+\Delta]$ of a system that is a composite system with $\Vset = S \dunion B$ and without any interaction between $S$ and $B$, whose Hamiltonian $\H_B$ on $B$ is a ``generic'' many-body Hamiltonian, will typically be exponentially close to a Gibbs state of $\H_S$ with an inverse temperature $\beta$ that depends in a reasonable way on $E$.
This works for all values of $E$ that are neither too low nor too high.
At the edges of the spectrum the number of states of the bath will be too low to allow for a good approximation of the number of states by a twice differentiable function.
In addition, $\Delta$ must be both small compared to $\norm[\infty]{\H}$ and large compared to the largest gaps in the spectrum of $\H$ in the relevant energy range.

Now we consider the influence of an interaction between $S$ and $B$.
The challenge posed by the fact that such an interaction will typically markedly perturb the energy eigenstates can be overcome by a perturbation theorem based on a result of Ref.~\cite{bhatia} (see also Refs.~\cite{Chandler70,bhatia83}) for projectors that are sums of spectral projectors.
\begin{theorem}[Stability of sums of spectral projectors (implied by Theorem~1 of Ref.~\cite{Riera2012})] \label{thm:stabiityofsumsofspectralprojectors}
  Given an energy interval $[E,E+\Delta]$ and two Hamiltonians $\H,\H' \in \Obs(\mcH)$ with spectral decompositions $H = \sum_k E_k\,\Pi_k$ and $H' = \sum_k E'_k\,\Pi'_k$.
  Let $P$ and $P'$ be projectors that are sums of the spectral projectors $\Pi_k$ and $\Pi'_k$ to energies in $[E,E+\Delta]$ of $\H$ and $\H'$, respectively, i.e.,
  \begin{align} \label{eq:defofPandPprime}
    P &\coloneqq \sum_{k\suchthat E_k\in[E,E+\Delta]} \Pi_k & &\text{and}& P' &\coloneqq  \sum_{k\suchthat E'_k\in[E,E+\Delta]} \Pi'_k .
  \end{align}
  Then for every $\epsilon>0$
  \begin{equation}
    \norm[1]{P-P'} \leq \big(\rank(P) + \rank(P')\big)\, \frac{\norm[\infty]{\H-\H'}}{\epsilon} + \rank(P_\epsilon) + \rank(P'_\epsilon)
  \end{equation}
  where
  \begin{align}
    P_\epsilon &\coloneqq \sum_{k\suchthat E_k\in[E,E+\epsilon]\union[E+\Delta-\epsilon,E+\Delta]} \Pi_k \\
    \intertext{and}
    P'_\epsilon &\coloneqq  \sum_{k\suchthat E'_k\in[E,E+\epsilon]\union[E+\Delta-\epsilon,E+\Delta]} \Pi'_k .
  \end{align}
\end{theorem}

The rather technical theorem stated above has immediate consequences for the stability of micro-canonical states:
\begin{corollary}[Stability of micro-canonical states \cite{Riera2012}] \label{corr:stabilityofmicrocanonicalstates}
  Given an energy interval $[E,E+\Delta]$ and two Hamiltonians $\H,\H' \in \Obs(\mcH)$ with spectral decompositions $H = \sum_k E_k\, \Pi_k$ and $H' = \sum_k E'_k\, \Pi'_k$ it holds that for every $\epsilon>0$
  \begin{equation} \label{eq:stabilityofmicrocanonicalstates}
    \tracedistance{\rhomc[\H]([E,E+\Delta])}{\rhomc[\H']([E,E+\Delta])} \leq \frac{\norm[\infty]{\H-\H'}}{\epsilon} + \frac{\Delta\Omega+\Omega_\epsilon}{2\,\Omega_{\max}} ,
  \end{equation}
  where $\Omega_{\min/\max} \coloneqq \min / \max \big(\rank(\rhomc[\H]([E,E+\Delta])),\rank(\rhomc[\H']([E,E+\Delta]))\big)$, $\Omega \coloneqq \Omega_{\max} - \Omega_{\min}$, and 
  \begin{equation}
    \begin{split}
      \Omega_\epsilon \coloneqq &\rank(\rhomc[\H]([E,E+\epsilon]\union[E+\Delta-\epsilon,E+\Delta])) \\
      + &\rank(\rhomc[\H']([E,E+\epsilon]\union[E+\Delta-\epsilon,E+\Delta])) .      
    \end{split}
  \end{equation}
\end{corollary}
\begin{proof}
  By the triangle inequality 
  \begin{equation}
    \tracedistance{\rhomc[\H]([E,E+\Delta])}{\rhomc[\H']([E,E+\Delta])} \leq \frac{\norm[1]{P-P'}+\Delta\Omega}{2\,\Omega_{\max}}
  \end{equation}
  with $P,P'$ defined as in \texteqref{eq:defofPandPprime}.
  Theorem~\ref{thm:stabiityofsumsofspectralprojectors} finishes the proof.
\end{proof}

What is the meaning of the corollary?
The statement is non-trivial if $\norm[\infty]{\H-\H'} \ll \Delta$.
Then one can expect that there exists an $\epsilon$ with the property that $\norm[\infty]{\H-\H'}\ll\epsilon\ll\Delta$, such that both $\norm[\infty]{\H-\H'}/\epsilon \ll 1$ and $(\Delta\Omega+\Omega_\epsilon)/(2\,\Omega_{\max}) \ll 1$.
Under the assumption of an approximately uniform density of states one finds that $\Omega_\epsilon/(2\,\Omega_{\max}) \approx 2\,\epsilon/\Delta$ and $\Delta\Omega/(2\,\Omega_{\max}) \lessapprox \norm[\infty]{\H-\H'}/\Delta$ such that the optimal choice for $\epsilon$ is approximately $\epsilon \approx \sqrt{\norm[\infty]{\H-\H'} \Delta/2}$, which yields
\begin{equation}
  \tracedistance{\rhomc[\H](I)}{\rhomc[\H'](I)} \lessapprox 4 \left({\frac{\norm[\infty]{\H-\H'}}{\Delta}}\right)^{1/2} .
\end{equation}

While the above example provides some intuition for how powerful Theorem~\ref{thm:stabiityofsumsofspectralprojectors} and Corollary~\ref{corr:stabilityofmicrocanonicalstates} are, the case of a uniform density of states is not the relevant situation if one is interested in showing thermalisation.
As we have seen in the beginning of this section, for $\rhomc^S[\H_0]([E,E+\Delta])$ to become approximately thermal it is necessary that the number of states of the bath grows exponentially with $E$.
What happens in this case?

First, notice that the two terms in the right hand side of \texteqref{eq:stabilityofmicrocanonicalstates} are non-negative and hence must both be small individually for the inequality to become non-trivial.
For the interesting case $\H = \H_0 + \H_I$ and $\H' = \H_0$ this implies that it is necessary that $\norm[\infty]{\H_I} \ll \epsilon$, so that the first term can become small.
For the second term we restrict our attention to $\Omega_\epsilon/(2\,\Omega_{\max})$ as $\Delta\Omega$ can reasonably be assumed to be smaller than $\Omega_\epsilon$.

If to good approximation 
\begin{equation} \label{eq:approxexponentialdensityofstates}
  \#_\Delta[\H_0](E) \approx \#_\Delta[\H](E) \propto \e^{-\beta\,E} ,
\end{equation}
then \cite[Appendix H]{Riera2012}
\begin{equation}
  \frac{\Omega_\epsilon}{2\,\Omega_{\max}} \gtrapprox \frac{1-\e^{-\beta\,\epsilon}}{2\,(1-\e^{-\beta\,\Delta})} .
\end{equation}
That is, for Corollary~\ref{corr:stabilityofmicrocanonicalstates} to be non-trivial it must be possible to chose an $\epsilon$ such that
\begin{equation}
  \beta\,\norm[\infty]{\H_I} \ll \beta\,\epsilon \ll 1 .
\end{equation}
At the same time, if \texteqref{eq:approxexponentialdensityofstates} is fulfilled, then also \cite[Appendix H]{Riera2012}
\begin{equation}
  \frac{\Omega_\epsilon}{2\,\Omega_{\max}} \lessapprox \frac{\beta\,\epsilon}{1-\e^{-\beta\,\Delta}} .
\end{equation}
Under the reasonable assumption that $\Delta\Omega/(2\,\Omega_{\max}) \ll 1$ the choice $\epsilon = \sqrt{\norm[\infty]{\H_I}/\beta}$ yields
\begin{equation}
  \tracedistance{\rhomc[\H](I)}{\rhomc[\H_0](I)} \lessapprox 2 \frac{\sqrt{\beta\,\norm[\infty]{\H_I}}}{1-\e^{-\beta\,\Delta}} ,
\end{equation}
which gives a non-trivial upper bound as long as
\begin{equation} \label{eq:weakinteractioncondition}
  \norm[\infty]{\H_I} \ll 1/\beta \ll \Delta .
\end{equation}
Concluding, we can say that for reasonable bath Hamiltonians $\H_B$, and if the coupling is weak enough and $\Delta$ large enough such that \texteqref{eq:weakinteractioncondition} is fulfilled, then one can expect that
\begin{equation} \label{eq:expectedorderupperboundondeviationfromlocalthermal}
  \tracedistance{\rhomc^S[\H]([E,E+\Delta])}{\rhog[\trunc{\H_S}S](\beta)} \in 
  \landauO\left(\left({\beta\,\norm[\infty]{H_I}}\right)^{1/2} \right) ,
\end{equation}
i.e., that the reduced state on subsystem $S$ of the micro-canonical state is close to a Gibbs state of the restricted Hamiltonian truncated to $S$.
Corollary~\ref{corr:stabilityofmicrocanonicalstates} and the above discussion quantify the errors in the approximate equalities Eq.~(7) in Ref.~\cite{Popescu05} and Eq.~(18) in Ref.~\cite{Goldstein06}.

For the rest of this section we consider a bipartite quantum system with $\Vset = S \dunion B$ of spins of fermions with Hamiltonian $\H$.
Let $\H_0 \coloneqq \H_S + \H_B$ and $\H_I \coloneqq \H - \H_0$.
We are now in a position to state the \emph{kinematic} version of the thermalisation result, which follows from the above discussion of Corollary~\ref{corr:stabilityofmicrocanonicalstates} and Theorem~\ref{thm:measureconcentrationforquantumstatevectors}.
\begin{observation}[Most Haar random states are locally thermal \cite{Riera2012}] \label{obs:thermalisationonofrandomstaates}
  Let $R \coloneqq [E,E+\Delta]$ be an energy interval and $\mcH_R \subseteq \mcH$ the subspace spanned by all eigenstates of $\H$ to energies in $R$ with dimension $d_R \coloneqq \dim(\mcH_R)$.
  If the bath has a ``generic'' locally interacting Hamiltonian with the property that for energies in $[E,E+\Delta]$ the logarithm of the number of states $\ln \#_\Delta[\trunc{\H_B}B]$ can be well approximated by an affine function with slope $\beta$ and if moreover $\Delta$ is sufficiently large and the interaction sufficiently weak such that
  \begin{equation}
    \norm[\infty]{\H_I} \ll 1/\beta \ll \Delta ,
  \end{equation}
  and the interval $R$ is sufficiently far from the edges of the spectrum, then for every $\epsilon > 0$
  \begin{equation} \label{eq:proababilityofrandomstatebeeingnotthermal}
    \begin{split}
      \probability_{\ket\psi\sim\muhaar[\mcH_R]}\left(\tracedistance{\ketbra\psi\psi^S}{\rhog[\trunc{\H_S}S](\beta)} \geq \epsilon + \delta(\H_B) + \landauO\left(\left({\beta\,\norm[\infty]{H_I}}\right)^{1/2}\right) \right)\\
      \leq 2\,d_S^2\,\e^{-C\,d_R\,{\epsilon}^2/d_S^2} ,
    \end{split}
  \end{equation}
  where $C = 1/(36\,\pi^3)$ and $\delta(\H_B)$ decreases fast with the size of the bath.
\end{observation}

To state the \emph{dynamic} result we introduce the notion of \emph{rectangular states} \cite{Riera2012}.
We call a state $\rho \in \Qst(\mcH)$ of a quantum system with Hilbert space $\mcH$ and Hamiltonian $\H \in \Obs(\mcH)$ \emph{rectangular} with respect to an energy interval $[E,E+\Delta] \subset \R$ if de-phasing with respect to $\H$ yields the micro-canonical state corresponding to $[E,E+\Delta]$.
For example, if $\H$ has no degeneracies, then a state is a rectangular state if, when expressed in the eigenbasis of $\H$, it has non-zero matrix elements only in some diagonal block and the same value for each entry on the diagonal in this block.
The class of rectangular states is not a very large class of states, but generally comprises a lot of pure states and usually also states that are \emph{out of equilibrium}, in the sense that their reductions on a small subsystem are well distinguishable from a thermal state and at the same time have a sufficiently widespread energy distribution such that Theorem~\ref{thm:equilibrationonaverage} can be used to guarantee equilibration on average.
Nevertheless, all these states have a tendency to thermalise dynamically:
\begin{observation}[Thermalisation on average \cite{Riera2012}] \label{obs:thermalisationonaverage}
  Let $R \coloneqq [E,E+\Delta]$ be an energy interval.
  Let the bath have a ``generic'' locally interacting Hamiltonian with the property that in an energy interval $[E,E+\Delta]$ the logarithm of the number of states $\ln(\#_\Delta[\trunc{\H_B}B])$ can be well approximated by an affine function with slope $\beta$.
  If $\Delta$ is sufficiently large and the interaction sufficiently weak such that
  \begin{equation}
    \norm[\infty]{\H_I} \ll 1/\beta \ll \Delta ,
  \end{equation}
  and the interval $R$ is sufficiently far from the edges of the spectrum, then the time evolution is such that the subsystem $S$ thermalises on average, in the sense of Definition~\ref{def:thermlaisationonaverage}, for any initial state $\rho(0) \in \Qst(\mcH)$ that is rectangular with respect to $R$ in the sense that
  \begin{equation} 
    \taverage[T]{\tracedistance{\rho^S(t)}{\rhog[\trunc{\H_S}S](\beta)}} \leq 
    \frac{1}{2}
    \left(
    {N(\epsilon)\,d_S^2\,g((p_k)_{k=1}^{d'}) }
    \right)^{1/2}
    + \delta(\H_B) + \landauO\left(\left({\beta\,\norm[\infty]{H_I}}\right)^{1/2}\right) ,
  \end{equation}
  where $\delta(\H_B)$ decreases fast with the size of the bath, and, as in Theorem~\ref{thm:equilibrationonaverage},
  \begin{align}
    N(\epsilon) &\coloneqq \sup_{E \in \R} |\{(k,l) \in [d']^2\suchthat k\neq l \land E_k - E_l \in [E,E+\epsilon] \}| \\
    g((p_k)_{k=1}^{d'}) &\coloneqq \min(\sum_{k=1}^{d'} p_k^2, 3  \maxprime_k p_k ) ,
  \end{align}
  with $(p_k)_{k=1}^{d'}$ the energy populations, i.e., $p_k \coloneqq \Tr(\Pi_k\,\rho(0))$, and $\maxprime_k p_k$ the second largest element in $(p_k)_{k=1}^{d'}$.
\end{observation}

The class of rectangular states seems fairly unnatural on first sight, however, condition of being rectangular can be slightly weakened.
For small deviations from a rectangular state Observation~\ref{obs:thermalisationonaverage} still essentially holds, just an additional error must be taken into account.
If the deviation from rectangular is in a sense uncorrelated with the relevant properties of the energy eigenstates, then even relatively large deviations should be tolerable as the errors will not accumulate but rather cancel each other out.
In the worst case, however, the deviation from rectangular could be highly correlated with the expectation value of, say, a local observable.
Then, even small deviations from rectangular can lead to noticeable deviations of the equilibrium state from a thermal state. that this can indeed happen in natural models for natural initial states \cite{PhysRevLett.10-6}.
In this sense the condition of being rectangular is necessary for thermalisation if no conditions on the energy eigenstates are to be imposed.

A comment on the notion of weak coupling used here is in order:
The condition that is needed for the above results to be non-trivial is (compare \texteqref{eq:expectedorderupperboundondeviationfromlocalthermal})
\begin{equation} \label{eq:weakcouplingcondition} 
  {\beta\,\norm[\infty]{H_I}} \ll 1 .
\end{equation}
This is a significant improvement over the condition that would be necessary to guarantee that naive perturbation theory on the level of individual energy eigenstates is applicable (namely that $\norm[\infty]{H_I}$ is much smaller than the gaps of $\H_0$).
While the gaps of $\H_0$ become exponentially small with the system size $\beta$ can be expected to be an intensive quantity, i.e., to be independent of the system size. It may be worth noting that conditions similar to this
have been considered in very practical contexts, say, when studying the thermalisation of two weakly coupled finite metallic grains \cite{1206.2408v1}.

In the case of a 1D system with short range interactions and if $S$ is a set of consecutive sites $\norm[\infty]{H_I}$ is also intensive.
In this case, \texteqref{eq:weakcouplingcondition} is a physically natural condition to call the coupling \emph{weak}.
In the analogous situation in higher dimensional lattices, for example a system with nearest neighbour interactions on a 2D square lattice and $S$ the sites inside a ball around the origin, $\norm[\infty]{H_I}$, however, scales with the surface of the region $S$, making the above bounds useless already for medium sized $|S|$.
Thus, the above results are not entirely satisfactory.

The reason for this is essentially that the trace distance is a very sensitive metric.
If $\rhomc^S[\H]([E,E+\Delta])$ and $\rhog[\trunc{\H_S}S](\beta)$ for the optimal $\beta$ only differ slightly on each of the sites along the boundary of $S$, then their trace distance (at least as long as it is sufficiently far from one) will be approximately proportional to the surface of $S$.
In consequence, the unfavourable scaling of the given error bounds is expected.

\subsection{Thermalisation in translation invariant systems and equivalence of ensembles}
\label{sec:translationallyinvariant}
Thermalisation and the related question of the \emph{equivalence of ensembles} have recently also been investigated in the more concrete setting of \emph{(translation invariant) locally interacting systems} on cubic lattices \cite{Mueller2013}.
The additional structure can be used to go beyond the results discussed in Sections~\ref{sec:thermalisationunderassumptionsontheinitialstate}.
In this section we discuss the main results of Ref.~\cite{Mueller2013} and the generalisations achieved in Ref.~\cite{BrandaoNew}.

More concretely, Refs.~\cite{Mueller2013,BrandaoNew} consider systems with \emph{$k$-local Hamiltonians} on cubic lattices.
A Hamiltonian $\H$ is \emph{$k$-local} if for some spatial dimension $D\in\Z^+$ and linear size $n\in\Z^+$ the vertex set of the system is $\Vset = [n]^D$, the edge set $\Eset$ contains only subsystems $X$ of diameter at most $k$ measured in the graph distance of the lattice, and the corresponding local terms $\H$ have norm bounded by one, i.e., $\norm[\infty]{\H_X} = 1$.
Furthermore, a Hamiltonians $\H$ is called translation invariant if for any two subsystems $X,X' \subset \Vset$ that differ only by a translation on the lattice it holds that $H_X = H_{X'}$.

In Ref.~\cite{Mueller2013} a family of translation invariant systems of increasing size is considered \emph{thermalising} if they equilibrates on average to a state that, in the limit of infinite system size, becomes indistinguishable from thermal states of that stem.
This is a very natural notion of thermalisation in the translation invariant setting.
Ref.~\cite{Mueller2013} contains theorems very reminiscent to both the \emph{kinematic thermalisation} result (Observation~\ref{obs:thermalisationonofrandomstaates}) and the dynamical result (with a similar conditions on the initial state) on \emph{thermalisation on average} (Observation~\ref{obs:thermalisationonaverage}) for this notions of thermalisation.
The results of Ref.~\cite{Mueller2013} are applicable in situations with strong coupling between subsystem and bath, i.e., $\norm[\infty]{\H_I} > 1/\beta$ but are only asymptotic statements and work only for temperatures around which the translation invariant system has a ``unique phase'' (see Ref.~\cite{Reed1980,Mueller2013} for more details) in the limit of infinite system size.
To understand what a ``unique phase'' is note that in the limit of infinite system size a translation invariant state $\rho$ is given by a series of subsystems states $\rho_X$ which for all $X\subseteq X' \subset \Vset$ fulfil the consistency conditions $\rho_X = \Tr_{\compl X} \rho_{X'}$. 
It is then instructive to \emph{define} a translation invariant state $\rho$ of the infinite system to be \emph{thermal} if it minimises the free energy density
\begin{equation}
  f(\rho) \coloneqq \lim_{|X| \to \infty} \frac{\Tr(\H_X\,\rho_X) - S(\rho_X) / \beta}{|X|} .
\end{equation}
Whenever $|\Vset|$ is finite,  this definition is consistent with our definition of a thermal state from Eq.~\eqref{eq:defthermalstate} and moreover the thermal state is unique.
In infinite systems, however this is not the case any more and one hence says that a system has a ``unique phase'' around some inverse temperature $\beta$ if for all inverse temperatures close to $\beta$ the system has a unique thermal state in the above sense.
At low temperatures this condition is often violated (for example in the 2D Ising model below the Curie temperature).
In contrast, at high temperatures the existence of a unique phase is always ensured (see Section~\ref{sec:propertiesofthermalstatesofcompositesystems}).

The kinematic and dynamics thermalisation results of Ref.~\cite{Mueller2013} rest on a \emph{equivalence of ensembles} theorem.
Two ensembles, for example the canonical and micro-canonical ensemble, are said to be \emph{equivalent} here if their corresponding states become indistinguishable on small subsystems when the total system size is increased.
More concretely: given a locally interacting spin system with Hamiltonian $\H \in \Obs(\mcH)$ and an inverse temperature $\beta$, under which conditions does there exist a suitable energy interval $[E,E+\Delta]$ such that for all sufficiently small subsystems $S \subset \Vset$ the distinguishability $\tracedistance{\rhog^S[\H](\beta)}{\rhomc^S[\H]([E,E+\Delta])}$ is small.
Unfortunately Ref.~\cite{Mueller2013} does not give concrete finite size bounds on the distinguishability but only makes statements about the asymptotic behaviour.

This however was recently achieved in Ref.~\cite{BrandaoNew}, together with a generalisation to systems without translation invariance.
Before we can explain this result in more detail we need to introduce the notion of \emph{$(\xi,z)$-exponentially decaying correlations} that is closely connected to the notion of a ``unique state'' encountered before.
Remember the definition of the covariance from Eq.~\eqref{eq:covariance}.
A state $\rho$ of a system on a lattice is said to have \emph{$(\xi,z)$-exponentially decaying correlations} if for some constants $\xi,z \in \R^+$ and any two observables $A,B \in \Obs(\mcH)$
\begin{equation}
  \cov_\rho(A,B) \leq \norm[\infty]{A} \norm[\infty]{B} N^z\,\e^{-\dist(A,B)/\xi} ,
\end{equation}
where $\dist$ is again the graph distance of the lattice.

A simplified version of the main result of Ref.~\cite{BrandaoNew} can then be phrased as follows:
\begin{theorem}[Equivalence of ensembles {\cite[Theorem~1]{BrandaoNew}}]
  Fix a spatial dimension $D \in \Z^+$, a locality parameter $k \in \Z^+$, a linear region size $l\in\Z^+$, an inverse temperature $\beta$, and $\xi,z \in \R^+$.
  For $n\in\Z^+$ consider an infinite family of spin systems with vertex sets $\Vset = [n]^D$, Hilbert spaces $\mcH_\Vset$, and $k$-local Hamiltonians $\H_\Vset \Obs(\mcH)$.
  If the family of thermal states $\rhog[\H_\Vset](\beta)$ has $(\xi,z)$-exponentially decaying correlations, then for the family $\rhomc[\H_\Vset]([E_\Vset-\Delta_\Vset/2,E_\Vset+\Delta_\Vset/2])$ of micro-canonical states with
  \begin{align}
    E_\Vset &\coloneqq \Tr(\H_\Vset\,\rhog[\H_\Vset](\beta))\\
    \intertext{and}
    \Delta_\Vset &\coloneqq \left( \frac{1}{N} \Tr\big(\H_\Vset^2\,\rhog[\H_\Vset](\beta)\big) - \Tr\big(\H_\Vset\,\rhog[\H_\Vset](\beta)\big)^2 \right)^{1/2} 
  \end{align}
  it holds that
  \begin{equation}
    \lim_{n\to\infty} \tracedistance{\rhomc^{X_\Vset}[\H_\Vset]([E_\Vset,E_\Vset+\Delta_\Vset])}{\rhog^{X_\Vset}[\H_\Vset](\beta)} = 0 
  \end{equation}
  for any family $X_\Vset \subset \Vset$ of subsystems whose diameter grows at most as fast as $n^{1/(d+1)}$.
\end{theorem}
The main virtue of Ref.~\cite{BrandaoNew} is that it actually gives a concrete finite size bound on the average distance between the canonical and the micro-canonical state on hyper-cubic subsystems.
In short, Ref.~\cite{BrandaoNew} shows that and how the canonical and micro-canonical states become indistinguishable on any sufficiently small subsystem when the total system size increases given that $\beta$ is such that the thermal state of the total system has exponentially decaying correlations.
We see later in Section~\ref{sec:propertiesofthermalstatesofcompositesystems} that at sufficiently high temperatures the necessary correlation decay can always be ensured.

\subsection{Hybrid approaches and other notions of thermalisation}
\label{sec:othernotionsandhybridapproaches}
We have seen in the last sections that both approaches to explain thermalisation, the eigenstate thermalisation hypothesis and thermalisation under assumptions on the initial state, have their advantages and drawbacks.
They can be understood as extreme scenarios.
In most cases where thermalisation of closed quantum systems happens it is probably due to a mixture of the two effects.
An interpolation between the two previously discussed approaches is provided by the \emph{eigenstate randomisation hypothesis} (ERH) \cite{Ikeda11}.
The ERH is a weaker condition than the ETH.
Instead of demanding that for certain observables the expectation values of all individual energy eigenstates with nearby energies give approximately the same expectation value (compare Conjecture~\ref{conjecture:ethrigolform} and Definition~\ref{def:eth}), the ERH requires only that the variance of certain coarse-grainings of the sequence of expectation values of an observable in the energy eigenstates becomes sufficiently small.
This, together with a condition on the smoothness of the energy distribution of the initial state that is milder than what we required when we introduced the class of rectangular states, is sufficient to prove a thermalisation result that is similar in spirit to Observation~\ref{obs:thermalisationonaverage} \cite{Ikeda11}.
Again, numerical evidence for the validity of the ERH in certain models has been collected \cite{Ikeda11}.

It seems worth repeating that the notion of thermalisation used here is surely not the only reasonable one. For example Ref.~\cite{0907.0108v1} works in the setting of macroscopic commuting observables of von Neumann, which we discussed briefly in Section~\ref{sec:typicality}.
A system is declared to be in thermal equilibrium if there is a phase cell that is much larger than all others and the state of the system is almost completely contained in the subspace corresponding to this cell.

Many other definitions of thermalisation or thermal equilibrium in quantum many-body systems are possible.
For example, in the context of the ETH it is sometimes said that a system is \emph{thermal} if the expectation values of a given observable in the energy eigenstates of a system are, up to small fluctuations, smooth functions of the energy (compare for example Ref.~\cite{Beugeling2013}).
The validity of fluctuation-dissipation theorems has also been considered as a condition for thermalisation \cite{Foini2011,Foini2012}.

The notion of \emph{relative thermalisation} \cite{DelRio2014} focuses on yet another aspect of thermalisation.
Rather than being concerned with the closeness of an equilibrium state to a thermal state of some kind it stresses that a system can be considered truly thermal only if it is not correlated with any other relevant system, as otherwise phenomena such as anomalous heat flow, which go against the predictions of thermodynamic, can occur.
In Ref.~\cite{DelRio2014} a subsystem $S$ is called approximately thermal relative to a reference system $R$, if the joint state $\rho_{SR}$ is close in trace norm to a state of the form $\pi_S \otimes \rho_R$ with $\pi_S$ being a suitable micro-canonical state.
Decoupling techniques can be used to show that whenever certain entropic inequalities are fulfilled then most joint evolutions of $S$, $R$, and an environment lead to approximate relative thermalisation \cite{DelRio2014}.

\subsection{Investigations of thermalisation in concrete models}
\label{sec:numericalthermal}
A large body of literature is concerned with investigations of thermalisation in specific quantum many-body models.
Many of those studies are directly concerned with testing a variant of the eigenstate thermalisation hypothesis (ETH) at the level of individual eigenstates.
The various \emph{eigenstate thermalisation hypotheses} differ in whether they conjecture closeness to a micro-canonical or a canonical average and concerning the type of observables they supposedly apply to.
\emph{Few body} and \emph{(approximately) local} observables are the two most frequently encountered choices.

The ETH gained wide popularity after the series of influential works \cite{Rigol08,Rigol09,Santos10}.
They identify the ETH as the mechanism for thermalisation and study its breakdown close to integrability in systems of hardcore bosons by means of exact diagonalisation.
Similar conclusions are reached in Ref.~\cite{1102.0528v1} for fermionic systems and the ETH is compared with other signatures of quantum chaos.
Ref.~\cite{Beugeling2013} represents a sound and detailed study of the validity of the ETH in systems with a tunable integrability breaking term by means of finite size scaling and varying the strength of the integrability breaking term.
Ref.~\cite{Steinigeweg2013} discusses the validity of the ETH in a simple model making use of a numerical technique that does not rely on exact diagonalisation.
Ref.~\cite{Steinigeweg2013a} presents a detailed study of the fluctuations of diagonal and off-diagonal matrix elements in the energy eigenbasis of certain physical observables in Heisenberg spin chains that confirms that in the non-integrable case the ETH is fulfilled.
Ref.~\cite{1103.0787v1} finds a breakdown of thermalisation and the ETH in a non-integrable model of spin-less fermions with a power law like random hopping term if the decay exponent is sufficiently large.
In Ref.~\cite{1108.0928v1} the ETH is connected with von Neumann's quantum ergodic theorem and it is confirmed that after a quench from a model that fulfils the ETH to one that does not (for the system being integrable), a system can still behave thermal. 
Ref.~\cite{Ikeda2013a} performs a finite size scaling analysis of the validity of the ETH in the (integrable) Lieb-Liniger model and demonstrates that a weaker version of the ETH still holds that is sufficient to guarantee apparent thermalisation for initial states that occupy sufficiently many energy eigenstates.

On top of that, a large body of literature exists that investigates all sorts of aspects of thermalisation and how various properties of the Hamiltonian and initial state influence it --- in fact, this has a long history \cite{1201.0578v1,Jensen1985}:
Refs.~\cite{Kollath07,Flesch08,Cramer2008,Ronzheimer2013,Sorg2014} numerically and experimentally study transport and thermalisation in the (non-integrable) Bose-Hubbard model.
Ref.~\cite{Moeckel2008} focuses on the (fermionic) Hubbard model at small interaction strength.
Using flow techniques the temporal evolution investigated and is found to go trough three distinct regimes.
After an initial build-up of correlations the system exhibits an intermediate, non-equilibrium, 
pre-thermalised, quasi-steady state and then eventually becomes indistinguishable from being thermalised.
Similar pre-thermalisation effects -- building upon the theoretical understanding discussed in Ref.~\cite{Berges2004} -- have been observed in Ref.~\cite{Queisser2013} in Bose- and Fermi-Hubbard models, in Ref.~\cite{Kehrein2013} in systems of spin-less fermions, and in Ref.~\cite{Gambassi} in instances of non-integrable quantum spin chains.
Similar pre-thermalisation effects were also found in systems evolving under stochastically changing Hamiltonians \cite{Marino2012,Marino2014}.
Ref.~\cite{Karzig2010} studies the energy relaxation and thermalisation of hot electrons in quantum wires.
Ref.~\cite{Banuls10} investigates the influence of the initial state on the time scales on which thermalisation happens in a non-integrable model.
Ref.~\cite{Sirker2013} looks at local and non-local conservation laws and how they influence the non-equilibrium dynamics and thermalisation.
The equilibration and thermalisation after a quench to a coupled Hamiltonian of two identical uncoupled systems initially in thermal states at different temperatures is studied in \cite{Ponomarev2012} and thermalisation to a state close to a joint thermal state is found. 
Ref.~\cite{Pagel2013} investigates conditions for equilibration and thermalisation (albeit in the sense of convergence in the limit $t\to\infty$) in the well studied model of a central harmonic oscillator linearly coupled to an infinite number of other oscillators starting from a non-thermal product initial state.

The bottom line of this large amount of investigations is as follows:
The energy eigenstates in the bulk of the spectrum, i.e., those to energies that are neither too low nor too high, of sufficiently large and sufficiently complicated composite quantum systems seem to generically fulfil some variant of the ETH for certain physically meaningful local or few body observables.
Equilibration of local and few body observables is a very common phenomenon shared by almost all reasonable locally interacting many-body models for wide classes of initial states.
This in turn implies that those systems which fulfil a suitable variant of the ETH also almost always dynamically thermalise after being started in a non-equilibrium initial state, like for example after a quench.

Many studies moreover conclude that the fulfilment of the ETH is related to \emph{non-integrability} or \emph{chaos} \cite{1103.0787v1,1102.0528v1,Polkovnikov11,1108.0928v1,Neuenhahn10,Larson13,1201.0186v1,Beugeling2013,1112.3424v1.pd,Beugeling2013,Singh}.
Moreover, it is often suggested that systems fulfil the ETH and thermalise if and only if they are \emph{non-integrable} \cite{Rigol08,Rigol09,Rigol11,Biroli09,Znidaric09}, disordered systems being an important exception \cite{PhysRevLett.10-6} (see also Section~\ref{sec:mbl}).
What precisely the term \emph{non-integrable} means in the context of many-body quantum mechanics and especially in systems without a well-defined classical limit and the relation between \emph{\mbox{(non-)}integrability} and \emph{(exact) solvability} are, however, still the subject of a lively debate \cite{1111.3375v1,1012.3587v1,Braak11,Fine2013}.
We will come back to this issue in Section~\ref{sec:integrability}.

\section{Absence of thermalisation and many-body localisation}
\label{sec:absenceofthermalisation}
In the past section we have identified and discussed conditions under which locally interacting many-body systems exhibit thermodynamic behaviour like equilibration and thermalisation.
Complementing these considerations, in this section we will identify and discuss scenarios in which thermalisation is prevented.
In particular we will be concerned with situations in which a system fails to thermalise locally because small subsystems retain memory of their initial conditions.
Quite intuitively the presence or absence of thermalisation is intimately linked to the transport properties of a system.
After all, for thermalisation to happen stating from a non-equilibrium initial condition, some equalisation of initial imbalances in, for example, the spatial distribution of energy or particles must happen.  
We will see that the concept of entanglement, to what extend it is present in the eigenstates of a Hamiltonian and how it spreads through the system during time evolution, will be of great use to gain insights into such transport processes.

We start by formulating what we mean by \emph{absence of thermalisation} and in particular define \emph{violation of subsystem initial state independence}.
The main part of this section will be dedicated to the discussion of physical situations in which one naturally expects such an absence of thermalisation to happen:
systems with \emph{static disorder} in the Hamiltonian.
This will lead us to the intriguing phenomenon of \emph{(many-body) localisation}, a type of localisation in which disorder and interactions interplay in a subtle fashion. 
In fact, one of the currently discussed definitions for many-body localisation in quantum systems takes the absence of thermalisation as its defining feature \cite{PhysRevB.82.17,Oganesyan2007,Huse2014}.
We make an attempt to survey the newly emerging debate concerning this phenomenon.
After a brief introduction to Anderson localisation we discuss properties that can be expected from a many-body localised phase and collect different notions of many-body localisation.

\subsection{Violation of subsystem initial state independence}
\label{sec:violationofinitialstateindependence}
We start with defining \emph{subsystem initial state independence}.
Roughly speaking, a system fulfils subsystem initial state independence for a certain set of initial states if changing only the subsystem part of an initial state from that set does not noticeably influence the equilibrium state of the subsystem.
This can be put as follows:

\begin{definition}[Subsystem initial state independence] \label{def:subsysteminitialstateindependence}
  We say that a composite system with Hilbert space $\mcH$ and Hamiltonian $\H \in \Obs(\mcH)$ satisfies \emph{subsystem initial state independence} for subsystem $S$ on average with respect to a given set of initial states $\Qst_0 \subseteq \Qst(\mcH)$ if for all $\rho(0) \in \Qst_0$ the equilibrium state on $S$ is sufficiently independent of its initial state in the sense that for every quantum channel $\Chann \in\Qch(\mcH)$ with support $\supp(\Chann) \subseteq S$ the states $\rho(0)$ and $\Chann(\rho(0))$ have the property that $\tracedistance{\Tr_{\compl{S}}[\$_\H(\rho(0))]}{\Tr_{\compl{S}}[\$_\H(\Chann(\rho(0)))]}$ is sufficiently small.
\end{definition}

If a system does not exhibit any local exactly conserved quantities, \emph{subsystem initial state independence}, as defined in Definition~\ref{def:subsysteminitialstateindependence}, with respect to a sufficiently large set of initial states $\Qst_0 \subset \Qst$, can rightfully be considered a necessary condition for thermalisation of small subsystems, regardless of which precise definition of thermalisation is adopted.

As was shown in Ref.~\cite{PhysRevLett.10-6} subsystem initial state independence after a quench can be provably violated if the Hamiltonian exhibits a \emph{lack of entanglement in the eigenbasis}.
The central quantity in the argument is the \emph{effective entanglement in the eigenbasis}.
Given a bipartite spin system with $\Vset = S \dunion B$, Hilbert space $\mcH$, and Hamiltonian $\H \in \Obs(\mcH)$ with spectral decomposition $\H = \sum_{k=1}^{d'} E_k\,\Pi_k$ we define for any pure state $\psi = \ketbra\psi\psi \in \Qst$ the \emph{effective entanglement in the eigenbasis} as
\begin{equation}
  R_{S|B}(\psi) \coloneqq \sum_{k=1}^{d'} p_k\,\tracedistance{\Tr_B(\Pi_k\,\psi\,\Pi_k)/p_k}{\psi^S} ,
\end{equation}
with $p_k \coloneqq \Tr(\Pi_k\,\psi)$ the energy level populations.
If the Hamiltonian is non-degenerate it takes the simpler form
\begin{equation} \label{eq:effectiveentranglementintheeigenbasisnondegenerate}
  R_{S|B}(\psi) = \sum_{k=1}^{d} p_k\,\tracedistance{\Tr_B(\ketbra{E_k}{E_k})}{\psi^S} .
\end{equation}
The name \emph{effective entanglement in the eigenbasis} is justified by a result of Ref.~\cite{PhysRevLett.10-6}, which bounds $R_{S|B}$ by a quantity that is closely related to the \emph{geometric measure of entanglement} \cite{Shimony95,Barnum2001,Wei2003}.
If the eigenstates of $\H$ are little entangled, and $\psi$ is a suitably chosen product state, then $R_{S|B}(\psi)$ is small.
In fact one can show \cite{PhysRevLett.10-6,Gogolin2014} that there exist many initial states that are perfectly distinguishable on the subsystem but that have both the properties needed to ensure equilibration on average of small subsystems according to Theorem~\ref{thm:equilibrationonaverage} and a small $R_{S|B}$ if $\H$ is non-degenerate and its eigenbasis is only little entangled. The type of system that are naturally expected to show such a behaviour, as will
be discussed in the subsequent subsection, are many-body localising systems.
The effective entanglement in the eigenbasis can be used to bound how much closer the reduced states on $S$ of two different initial states can move during equilibration on average in the following sense:

\begin{theorem}[Distinguishability of de-phased states {\cite[Theorem~1]{PhysRevLett.10-6}}] \label{thm:absenceofthermalisaton}
  Consider a bipartite spin system with $\Vset = S \dunion B$, Hilbert space $\mcH$ and Hamiltonian $\H \in \Obs(\mcH)$.
  For $j \in \{1,2\}$ let $\psi_j(0) = \psi^S_j(0) \otimes \psi^B_j(0) \in \Qst(\mcH)$ be two initial product states and set $\omega^{S(j)} \coloneqq \Tr_B(\$_\H(\psi_j(0)))$ then
  \begin{equation}
    \tracedistance{\omega^{S(1)}}{\omega^{S(2)}} \geq \tracedistance{\psi^S_1(0)}{\psi^S_2(0)} - R_{S|B}(\psi_1(0)) - R_{S|B}(\psi_2(0)) .      
  \end{equation}
\end{theorem}
If the state of the subsystem $S$ equilibrates on average during the evolution under $\H$ for the two initial states, then the de-phased states $\omega^{S(j)} = \Tr_B(\$_\H(\psi_{j}(0)))$ are the respective equilibrium states.
The theorem shows that if $R(\psi_1(0))$ and $R(\psi_2(0))$ are both small, then the subsystem equilibrium states $\omega^{S(1)}$ and $\omega^{S(2)}$ cannot be much less distinguishable than the initial states $\psi^S_1(0)$ and $\psi^S_2(0)$.
We summarise this in the following observation:
\begin{observation}[Absence of initial state independence] \label{obs:absenceofthermalisation}
  Consider a bipartite spin system with $\Vset = S \dunion B$ and Hilbert space $\mcH$.
  Let $\mcH_R \subseteq \mcH$ be a subspace of dimension $d_R \coloneqq \dim(\mcH_R)$.
  If $d_R$ is large and the Hamiltonian $\H \in \Obs(\mcH)$ has not too many degenerate energy gaps (see Theorem~\ref{thm:equilibrationonaverage} for details) and an orthonormal basis $(\ket{j})_{j=1}^{d_S}$ for $\mcH_S$ exists for which $\delta \coloneqq \max_{k\in[d]} \delta_k$, with
  \begin{equation} \label{eq:geometricentanglementdelta}
    \delta_k \coloneqq \min_{j\in [d_S]} \tracedistance{\Tr_B \ketbra{E_k}{E_K}}{\ketbra{j}{j}} ,
  \end{equation}
  is small, then for every $j,j' \in [d_S]$ there exist many initial states of the bath $\psi^B(0) \in \Qst(\mcH_B)$ such that according to Theorem~\ref{thm:equilibrationonaverage} both $\ketbra{j}{j} \otimes \psi^B(0)$ and $\ketbra{j'}{j'} \otimes \psi^B(0)$ lead to subsystem equilibration on average, but despite them having exactly the same initial state on the bath, the corresponding subsystem equilibrium states $\omega^{S(j)}$ and $\omega^{S(j')}$remain well distinguishable for most times during the evolution, in the sense that their trace distance $\tracedistance{\omega^{S(j)}}{\omega^{S(j')}}$ is significantly larger than zero whenever $j \neq j'$, because of Theorem~\ref{thm:absenceofthermalisaton}.
\end{observation}

A statement complementing Observation~\ref{obs:absenceofthermalisation} can be found in Ref.~\cite[Section B]{Linden09} (see also Ref.~\cite{MasterThesisHutter} for a generalisation to mixed initial states and situations with initial correlations to reference system).
There, it is shown that if the energy eigenstates of a non-degenerate Hamiltonian does contain a lot of entanglement, then subsystem initial state independence can be guaranteed.

In a very similar spirit as above, absence of initial state independence has also been studied later in Ref.~\cite{PhysRevE.82.01}, which gives a condition that is necessary for subsystem initial state independence.
The article mostly studies a simplified version of this condition, which essentially demands that the reductions of most eigenvectors of the Hamiltonian must be sufficiently close to the maximally mixed state.

More recently, initial state independence was studied in Ref.~\cite{Hutter11,MasterThesisHutter}).
By using the \emph{decoupling method} \cite{Dupuis2010,1109.4348v1,Szehr2012} and the formalism of so-called \emph{smooth min and max entropies} \cite{Koenig08,Ciganovic2013}.
The authors show that it can be decided from just looking at one particular initial state whether a system satisfies initial state independence for most initial states.
Moreover, they give sufficient and necessary entropic conditions for initial state independence of most initial states.
The authors consider both subsystem initial state independence and bath initial state independence, i.e., the independence of the equilibrium state of the subsystem from the initial state of the bath.
The results concerning the absence of subsystem initial state independence of Ref.~\cite{Hutter11}, when compared to those of Ref.~\cite{PhysRevLett.10-6} discussed above, have the advantage that they apply to specific points in time instead of time averaged states and that the subsystem does not need to be small.
On the other hand they only hold for most/typical initial states.

There exist several articles, including Refs.~\cite{1011.0781v1,Cazalilla11,Biroli09,Znidaric09}, that numerically and analytically study related effects.
Ref.~\cite{Biroli09} finds that the existence of few energy eigenstates that violate the \emph{eigenstate thermalisation hypothesis} (see also Section~\ref{sec:thermalisationunderassumptionsontheeigenstates} and in particular Definition~\ref{def:eth}) can lead to absence of thermalisation.
Ref.~\cite{Znidaric09} goes beyond the closed system setting and considers thermalisation and its absence in systems that are coupled to thermal baths and finds that certain integrable models do not thermalise.
Ref.~\cite{1011.0781v1} studies quenches in a homogeneous XY quantum spin chain with transverse field starting in ground, excited, and thermal states.
The authors find that after certain quenches local observables fail to thermalise and relate this behaviour to criticality.
Ref.~\cite{Cazalilla11} investigates equilibration and thermalisation in exactly solvable models and finds that in such models correlation functions can retain memory of the initial conditions.

\subsection{Anderson localisation}
\label{sec:anderson}
With the aim to develop a better understanding of particle and spin transport in materials with impurities, Anderson in his 1958 article \cite{Anderson1958} proposed a simple model for quantum mechanical particles in a lattice with a random potential and showed that the randomness can lead to a complete suppression of diffusion or transport.
This phenomenon became known as \emph{Anderson localisation}.
More concretely, the model studied by Anderson is a tight-binding model on a cubic lattice of dimension $D$ with a single particle hopping on the lattice sites. 
The Hilbert space is $\mcH = l^2(\Z^D)$ spanned by vectors $\ket{x}$ interpreted as the state with the particle at position $x \in \Z^D$.
The random Hamiltonian of the Anderson model reads
\begin{equation}\label{eq:andersonHamiltonian}
  \H (V)= \sum_{x,y\in \Z^D \oftype |x-y|=1} \ketbra{x}{y} 
  + \lambda\sum_{x \in \Z^D} V_{x}\, \ketbra{x}{x},
\end{equation}
where $\lambda>0$ and $V$ a family of random numbers $V_x$ drawn i.i.d.\ from a suitable distribution $\mu$.
The first term describes hopping between nearest neighbours in the lattice (Ref.~\cite{Anderson1958} actually also considers more general longer range hopping), while the second represents a random on-site potential.
For reviews on the Anderson model from the perspective of mathematical physics, see Refs.~\cite{Stolz,Abrahams2010,Hundertmark}.
For a general overview written on the occasion of the 50th anniversary of the phenomenon see Ref.~\cite{Lagendijk2009}.
For a good book also providing significant historical context, see Ref.~\cite{Abrahams2010}.
To simplify the discussion we concentrate on one spatial dimension $D=1$, and assume that the distribution $\mu$ is absolutely continuous with a bounded density of compact support.
Because of the existence of the rigorous mathematical literature, we moreover take the liberty to brush over some subtleties.

{The Anderson model exhibits ``localisation''.
This is true in at least two different senses of the term \cite{Stolz}:
First, the random Hamiltonian \eqref{eq:andersonHamiltonian} almost surely exhibits \emph{spectral localisation}, meaning that it has a pure point spectrum (that densely fills all non-trivial open intervals contained in its almost sure spectrum) and that the associated eigenfunctions are exponentially decaying.
The latter means that $H(V)$ has a complete countable set of eigenvectors $\{|E_k\rangle\}_k$ obeying
\begin{equation}
  \exists\+ C>0,\xi>0 \suchthat \forall \ket{E_k}\  \exists x_0 \suchthat \forall x \itholds |\braket{x}{E_k} | \leq C\,\e^{-|x-x_0|/\xi}.
\end{equation} 
Here $\xi>0$ is called the localisation length scale.
That is to say: Almost all Hamiltonian eigenvectors are exponentially clustering.
For systems in in more than one dimension a similar statement holds either at sufficiently high disorder, or for energies sufficiently close to band edges.
Note also that these results can be extended to finite systems, localisation then holding with high probability instead of almost surely.

Second, the model exhibits almost surely \emph{dynamical localisation}.
This can be captured as follows: The random Hamiltonian $H(V)$ is said to exhibits \emph{dynamical localisation} in an open interval $I$ if 
\begin{equation}
  \exists\+ C>0,\xi>0 \suchthat \forall x,y\in\Z^D \itholds \expectation\big(\sup_{t\in\R} \bra{x}\, e^{-\iu\,H(V)\,t}\, \Pi_I \, \ket{y} \big) < C\,\e^{-|x-y|/\xi}
\end{equation}
where $\Pi_I$ is the spectral projector corresponding to the interval $I$.
Dynamical localisation implies a complete absence of transport.
In particular it implies that all moments of the ``distance from the origin'' operator $|X|$, which acts like $|X|\,\ket{x} = |x|\,\ket{x}$, are bounded uniformly in time, i.e., that
\begin{equation}
  \forall p>0,\ x,y \in \Z^D \itholds \sup_{t\in\R} \norm{ |X|^q\, e^{-\iu\,H(V)\,t} \,\Pi_I \, \ket{\psi} } < \infty
\end{equation}
for $\ket\psi$ any state vector with compact support almost surely.
Despite the hopping term in the Hamiltonian, which in the absence of the disordered potential allows the particle to move through the lattice, in the Anderson model the particle ``gets stuck''.
The probability of finding it on a site different from its starting point decays exponentially with the distance uniformly for all times.
Dynamical localisation implies spectral localisation by the RAGE theorem, but the converse is not necessarily true.

The above discussion immediately carries over to, for instance, disordered quadratic fermionic systems in which the quasi-particles do not interact.
In one spatial dimension, the corresponding Hamiltonian reads
\begin{equation}
  \H(V) =\sum_{x\in \Z} \left(f_x^\dagger f_{x+1}+ f_{x+1}^\dagger f_x\right)+ \sum_{x\in \Z} V_x\, f_x^\dagger f_x,
\end{equation}
As the fermions do not interact in such quadratic models, each of them behaves as in Anderson's model and conductivity is completely lost.
The Hamiltonian can also be readily related to local spin models by virtue of the \emph{Jordan Wigner transformation} and Anderson's conclusion can be argued to still holds for interacting particles when the density is very low \cite{Anderson1958}.
}

\subsection{Many-body localisation}
\label{sec:mbl}
An intriguing, and in large parts still unsettled, question is whether and in what precise sense localisation survives in systems with interactions and significant particle densities.
This issue has already been raised by Anderson \cite{Anderson1958,Fleishman1980}.
One expects that in models with sufficiently strong disorder some characteristics of Anderson localisation should survive in the presence of interactions.
This new phase of matter is commonly referred to as the \emph{many-body localised} (MBL) phase.
In which sense and under which conditions this is in fact true is the subject of ongoing investigations.
In fact, there is no complete consensus yet as to what precisely constitutes many-body localisation in the first place. In the following we collect and compare different points of view (see also Ref.~\cite{Nandkishore2014}):

\begin{enumerate}
\item \emph{Suppression of transport and localisation in Fock space}: \label{item:localisationinfockspace}
  The influential Ref.~\cite{Basko2006a} gives significant evidence that indeed, localisation in the dynamical sense \cite{Hamza2012} is maintained in the presence of interactions, by invoking a combinatoric argument and perturbation theory: For sufficiently high disorder and sufficiently low temperature (and absence of a coupling to an external heat bath) it is demonstrated that the conductivity in a disordered fermionic lattice system is exactly zero.
  The argument makes use of the concept of \emph{localisation in Fock space} introduced in Ref.~\cite{PhysRevLett.78} (see also Ref.~\cite{Gornyi2005}):
  Consider a fermionic system whose Hamiltonian is a sum of a quadratic Hamiltonian $\H_0$ and an interaction term $\H_1$.
  A many-body state is called \emph{localised} if it is a superposition of only few of the (quasi-particle) eigenstates of $\H_0$.
  If the relevant eigenstates of $\H_0+\H_1$ are localised, i.e., all below a certain critical energy (called a \emph{mobility edge}), in this sense then below a critical temperature the system exhibits zero conductivity.

\item \emph{Absence of thermalisation}: \label{item:absenceofthermalisation}
  Closely related to the characteristic suppression of transport in localised systems is the absence of thermalisation due to a violation of initial state independence (see also Section~\ref{sec:violationofinitialstateindependence}).
  This is a natural expectation, since one does not expect the \emph{eigenstate thermalisation hypothesis} (ETH) to be valid within the MBL phase.
  Ref.~\cite{PhysRevB.82.17} for example studies a disordered Heisenberg chain and finds a violation of the ETH and interprets this as one of the defining features of MBL.
  As was shown in Ref.~\cite{PhysRevLett.10-6} violation of initial state independence in disordered systems can be understood as a consequence of a lack of entanglement in the eigenbasis.

\item \emph{Clustering of correlations}: \label{item:clusteringofcorrelations}
  Another definition puts the \emph{clustering of correlations of eigenvectors} into the centre of attention. 
  For quadratic models, it is expected that all Hamiltonian eigenvectors satisfy an \emph{area law} \cite{Eisert2008} for the entanglement entropy. 
  This means that in one dimension the von Neumann entropy of the reduced state of any energy eigenstate on any subsystem is upper bounded by a constant independent of the size of the subsystem.
  A similar feature has also been suggested as a possible definition for MBL \cite{Bauer2013}: 
  One then calls a system many-body localising if, not necessarily all but at least many (in a suitable sense), eigenstates satisfy an area law.
  A proof of a uniform area law (in expectation) was recently given for the case of the XY chain with disordered transverse magnetic field in Ref.~\cite{Abdulrahman2015}.
  Numerically, there is strong evidence that this is indeed the case in disordered interacting models, at least below a \emph{mobility edge} \cite{Luitz2014}, so an energy scale that separates the MBL from the ``ergodic'' regime.
  A connection with the dynamical aspect of localisation \cite{Hamza2012} was recently established in Ref.~\cite{MBLMPS}, where it was shown that invoking different readings of dynamical localisation, it follows that either all or many energy eigenvectors follow an area law.

\item \emph{Logarithmic growth of entanglement}: \label{item:logarithmicgrowthofentanglement}
  A yet different feature of MBL that has been suggested as a defining property is the \emph{logarithmic growth of entanglement} in time.
  While the entanglement of generic local Hamiltonian models is expected to grow linearly in time (see also Section~\ref{sec:liebrobinson}), quadratic models show a saturation of entanglement entropies.
  This is provably so, as a consequence of the complete suppression of transport.
  In interacting disordered models a slow --- logarithmic in time --- but unbounded growth of entanglement has been numerically observed \cite{Znidaric2008a,Pollmann}.
  This feature is perfectly compatible with individual eigenstates exhibiting little entanglement.
  In fact, maybe counter-intuitively, an unbounded growth of entanglement already follows from localised Hamiltonian eigenstates together with a generic spectrum \cite{MBLTransport}.

\item \emph{Approximately local constants of motion}: \label{item:localconstantsofmotion}
  Another discussed possibly defining feature of MBL is the presence of an \emph{extensive number of exactly or approximately local constants of motion} \cite{Chandran2014a}, with the feature that the Hamiltonian can be expressed entirely as a sum of polynomials in these quantities \cite{MBLTransport,Logarithmic,VidalMBL}.

  If indeed such local constants $\{A_j\}$ of motion can be found, $(g,K)$-local in the above
  sense for a suitable function $g$, violation of subsystem initial state independence and absence of thermalisation follow immediately: Since 
  \begin{equation}
    \tr(A_j \rho(t)) = \tr(A_j \rho(0)) 
  \end{equation}
  is true for all times $t$, the system can possibly only equilibrate to a state that has the same values for these conserved quantities (see Section~\ref{sec:gge}).

  Other, quite sophisticated features also follow from the presence of such approximately local constants of motion.
  For example, one can derive a Lieb-Robinson type bound with a causal ``cone'' that grows only logarithmically in time \cite{Logarithmic} (see also Section~\ref{sec:liebrobinson}).
  From such a bound one can derive that the entanglement entropy can grow at most logarithmic in time \cite{Logarithmic,Eisert06,Bravyi06-1}. 
  An similar bound has been also obtained in Ref.~\cite{Burrell2007} for a disordered XY spin chain and Ref.~\cite{Hamza2012} improves upon this by giving a \emph{zero velocity Lieb-Robinson bound} in disorder average for this model.

  A disadvantage of that definition is that it is far from clear how to construct or identify such approximately local constants of motion in the first place.
  In the disordered Ising model \cite{Imbrie2014} and the XXZ spin chain \cite{Chandran2014a} this is indeed possible, but no general strategy has yet been found \cite{VidalMBL}.
  Several of these defining features have also been connected and made plausible using real space renormalisation group approaches \cite{Vosk2014}.

\item \emph{Poissonian level statistics}: \label{item:poissonianlevelstatistics}
  Properties of the energy level statistics of Hamiltonians have proved to be useful indicators for quantum chaos and integrability.
  It is hence natural to investigate the influence of disorder on the level statistics.
  A key quantity in this context is the distribution of gaps between consecutive energy eigenvalues.
  For quadratic models, this distribution typically is a Poissonian one.
  For interacting models, it is generally expected to follow a Wigner-Dyson type distribution \cite{Oganesyan2007,Luitz2014,Jacquod1997,Georgeot1998,Haake10}.

  For typical many-body localised models, there is strong numerical evidence that the distribution is again close to Poissonian \cite{Luitz2014,Jacquod1997,PhysRevLett.78,PhysRevB.82.17,Oganesyan2007}.
  This can be quantified by the ratio of consecutive level spacings
  \begin{equation} 
    r_j= \frac{\min (\delta_j,\delta_{j+1}) }
    {
      \max (\delta_j,\delta_{j+1})
    },
  \end{equation}
  with $\delta_j= E_j-E_{j-1}$ being the gap between consecutive energy levels. 
  In the non-localised phase one can expect from Wigner's surmise leading to the Gaussian orthogonal (GOE) or unitary (GUE) ensemble that a disorder average of $r_j$ yields a value close to $r_\mathrm{GOE} \approx 0.5307$ or $r_\mathrm{GUE} \approx 0.5996$, while for a Poisson distribution, that one expect in the MBL phase, one obtains on average $r_\mathrm{Poisson}= 2 \log 2 - 1 \approx 0.3863$. 

  An extensive numerical analysis of this ratio of consecutive level spacings has been performed in Ref.~\cite{Luitz2014} for the random field Heisenberg model on a ring.
  Also finds excellent agreement of the position of the cross over in the consecutive level spacings statistics with that of an area law / volume law crossover of the entanglement entropy and a crossover in the scaling of the participation entropies (a quantity closely related to the inverse participation ratio and the effective dimension discussed in Section~\ref{sec:equlibrationintheweaksense}).
  This work also calculates these quantities in an energy resolved fashion and finds that it can happen that for low energies a system shows strong signatures of a Poissonian distribution, while for higher energies, it resembles a Gaussian orthogonal ensemble consistent with the existence of a mobility edge in interacting systems.

\item \emph{Power-law approach to equilibrium}: \label{item:survivalprobability}
  Ref.~\cite{Serbyn2014a} identified a power law (as opposed to exponential) approach to equilibrium of local observables as a characteristic feature of the MBL phase.
  In addition the MBL phase has also been found to exhibit a slow power law like decay of the disorder average of the survival probability, i.e., the fidelity with the initial state, at long times \cite{Torres-Herrera2015} (see also Section~\ref{sec:fidelitydecay}).
\end{enumerate}

Each of the definitions above only capture part of the intricate phenomenon of many-body localisation.
In particular it is far from clear whether disorder is really necessary to realise all of the above qualifying features of many-body localisation.
In fact, drawing intuition from classical glassy systems it is possible to design clean spin systems that show many of the features one would expect from a system with a MBL transition \cite{Garrahan}.
In fact many of the static properties discussed above may also occur in certain (nearly) integrable models without any disorder.
Concerning the dynamical features of MBL, Ref.~\cite{Yao2014} for example demonstrates that the slow growth of entanglement entropies \refitem{item:logarithmicgrowthofentanglement} can also exist in clean systems.
The same holds for long lived metastable states that break a symmetry of the system  \cite{Yao2014,Carleo2012,Schiulaz2015,Garrahan} (see also the effect of pre-thermalisation discussed in Section~\ref{sec:numericalthermal}).
It would be specifically intriguing to see rigorously whether fully translation invariant models can exhibit dynamical localisation in the sense of property \refitem{item:absenceofthermalisation} even for infinite time, similarly as this is possible for interacting disordered models \cite{PhysRevLett.10-6}.

In several physical architectures, Anderson and many-body localisation has already been experimentally observed.
Ref.~\cite{Billy2008} discusses an experimental observation of exponential localisation of a Bose-Einstein condensate in a random potential generated with a laser speckle pattern.
The recent Ref.~\cite{SchneiderMBL} experimentally probes the many-body localisation transition in a system of ultra cold fermions in a disordered optical lattice by measuring the imbalance between the occupation of even and odd sites starting from a situation where only even sites are occupied, resembling the experimental situation of Ref.~\cite{1101.2659v1}. 
For sufficiently strong disorder the imbalance is found to no longer decay to zero even for long times, reflecting the absence of thermalisation and the violation of a subsystem initial state independence very much in the spirit of Ref.~\cite{PhysRevLett.10-6}.

\section{Integrability}
\label{sec:integrability}
In this section we discuss a concept that has recently started playing an important role in the debate on equilibration and thermalisation in closed quantum systems --- the concept of \emph{integrability}.
It is often suggested or claimed that non-integrable systems thermalise, while integrable ones do not.
This wisdom has become folklore knowledge that is often invoked in discussions and talks on the topic (compare also Refs.~\cite{Rigol08,Rigol09,Rigol11,Biroli09,Znidaric09,1103.0787v1,1102.0528v1,Polkovnikov11,1108.0928v1,Neuenhahn10,Larson13,1201.0186v1,Beugeling2013,1112.3424v1.pd}).
In the following, we will briefly review the current state of affairs concerning the usage of the term \emph{(quantum) integrability} in the context of equilibration and thermalisation in closed quantum systems, comment on the concept of integrability and investigate to which extend the circumstantial evidence concerning the connection between \mbox{(non-)}integrability and thermalisation can be substantiated.

To that end we will first recapitulate the definition of integrability in classical mechanics and then discuss obstacles for a generalisation of the concept of integrability to the quantum setting.
This assessment is largely based on the previous works Refs.~\cite{Weigert1992,1012.3587v1,PhysRevLett.10-6}.
We finish with some speculations on the connection of quantum \mbox{(non-)}integrability and computational complexity.

\subsection{In classical mechanics}
\label{sec:integrabilityinclassicalmechanics}
In classical mechanics \cite{Arnold78} \emph{(Liouville) integrability} is a very well-defined concept.
Consider a classical system with $n \in \Z^+$ degrees of freedom, each associated with a \emph{coordinate} $q_k$ and a corresponding \emph{momentum} $p_k$.
Then, in the Hamiltonian formalism, the $2\,n$ \emph{canonical coordinates} $(q_k)_{k=1}^n$ and $(p_k)_{k=1}^n$ span the \emph{phase space} $\Cst $ of the system \cite{Kinchin1949}.
We assume that the \emph{Hamiltonian function} $\CH\oftype\Cst\to\R$, i.e., the energy functional, of the system is time independent.
It then governs the time evolution of the system via \emph{Hamilton's equations} \cite{Arnold78}:
\begin{align}
  &\forall k\in[n]\itholds &\dot p_k &= - \frac{\del\CH}{\del q_k} &  \dot q_k &= \frac{\del\CH}{\del p_k}
\end{align}
The dot indicates the derivative with respect to time of the corresponding quantity, i.e., $\dot q_k$ is the temporal change of $q_k$.
Integrating these differential equations yields the \emph{phase flow} $g_\CH^t\oftype\Cst \to \Cst$, which maps the initial phase space vector of a system at time $0$ to that at time $t \in R$.
Define for any two functions $F,G\oftype\Cst\to\R$ their \emph{Poisson bracket} $(F,G)\oftype\Cst\to\R$ as
\begin{equation}
  (F,G) \coloneqq \lim_{t\to0} \frac{\ddel}{\ddel t}\, F \circ g_G^t ,
\end{equation}
where $\circ$ denotes function \emph{composition}.
It turns out that $(\argdot,\argdot)$ is bilinear and skew-symmetric \cite{Arnold78}.
A function $F\oftype\Cst\to\R$ is called a \emph{first integral of motion} under the evolution induced by $\CH$ if $(F,\CH) = 0$.
More generally, if for $F,G\oftype\Cst\to\R$ it holds that $(F,G) = 0$, then $F$ and $G$ are said to be \emph{in involution}.
We can now define Liouville integrability:

\begin{definition}[Liouville integrability \cite{Arnold78}]
  A classical system with $n$ degrees of freedom is called (Liouville) integrable if it entails a sequence $(F_k)_{k=1}^{n}$ of $n$ independent first integrals of motion that are pairwise in involution.
\end{definition}
\emph{Liouville's theorem for integrable systems} shows that Liouville integrable systems can be solved, i.e., the time evolution can be explicitly calculated, in a systematic way by \emph{quadratures}, i.e., by direct integration of differential equations:
\begin{theorem}[Corollary of Liouville's theorem for integrable systems \cite{Arnold78}]
  If a system is Liouville integrable, its time evolution can be solved by quadratures.
\end{theorem}
In more detail: Liouville's theorem for integrable systems essentially ensures that, given the initial values of all canonical coordinates, the time evolution of an integrable system is confined to a \emph{smooth} \emph{submanifold} of the phase space that is \emph{diffeomorphic} to an $n$-dimensional \emph{torus}.
The time evolution is quasi-periodic and can be described in terms of the so-called \emph{action angle coordinates} $(\varphi_k)_{k=1}^n$ that parametrise the torus.

The action angle coordinates can be explicitly constructed from the sequence $(F_k)_{k=1}^{n}$ of $n$ independent first integrals of motion and the values fixed for them.
Fixing different values for the $n$ first integrals of motion results in different tori.
In the coordinate system of the action angle variables the equations of motion are given by $2\,n$ simple ordinary differential equations of the form $\dot F_k = 0$ and $\dot \varphi_k = w_k$, with $w_k \in \R$ being constants that depend on the values that were fixed for the $n$ first integrals of motion.

If a Liouville integrable system is perturbed, then the time evolution is generally not confined to a torus anymore and cannot be derived in a systematic way.
For small perturbations the Kolmogorov-Arnold-Moser (KAM) theorem ensures, under a so-called \emph{non-resonance condition}, that most tori are only deformed and the time evolution on them is then still quasi-periodic \cite{Moradi2001,Tabor1989,Poeschel03}.

In summary we have:
Integrability in classical systems implies \emph{systematic solvability} and thereby yields a \emph{qualitative classification} of classical systems.
Liouville integrable systems are not \emph{ergodic} (see Section~\ref{sec:canonicalapproaches}) in the sense that their phase space trajectory does not explore the whole phase space, but is confined to a portion of it.
Whether or not this implies that integrable systems cannot thermalise depends on the definition of thermalisation, but the motion of the system is quasi-periodic and hence no convergence of the state of the system in the limit $t\to\infty$ is possible.
Non-integrability in classical systems is \emph{not} sufficient for ergodicity or chaos and hence also not sufficient for notions of mixing or thermalisation based on these concepts.
Still, the concept of Liouville integrability yields a classification of systems with strong implications for their physical behaviour.

\subsection{In quantum mechanics}
\label{sec:integrabilityinquantumtheory}
Ideally, a notion of quantum integrability should yield a classification that divides quantum systems into two classes, \emph{integrable} ones and \emph{non-integrable} ones, with markedly different physical properties.
In addition it should, in some sense, be a generalisation of Liouville integrability.
However, if one tries to generalize the concept of Liouville integrability to quantum systems in a straight forward manner, one immediately encounters problems (see also Ref.~\cite{1111.3375v1}):

Consider a quantum system with $d$ dimensional Hilbert space $\mcH$ and Hamiltonian $\H \in \Obs(\mcH)$.
An orthonormal eigenbasis $(\ket{\tilde E_k})_{k=1}^d$ of $\H$, with corresponding eigenvalues $(\tilde E_k)_{k=1}^d$, can always be constructed in a systematic way by diagonalising the Hamiltonian.
The time evolution of an arbitrary initial state vector $\ket \psi \in \mcH$ is then given by
\begin{equation}\label{eq:timeevolutioninexpliciteform}
  t \mapsto \ket{\psi(t)} \coloneqq \sum_{k=1}^d |\braket{\tilde E_k}{\psi}|\,\e^{\iu\,\tilde\varphi_k(t)}\,\ket{\tilde E_k} ,
\end{equation}
with $\tilde\varphi_k(t) \coloneqq \arg(\braket{\tilde E_k}{\psi})-\,\tilde E_k\,t$, where $\arg$ is the \emph{argument} function, i.e., for every $c \in \C, c = |c|\,\e^{\iu\,\arg(c)}$.
The overlaps $\braket{\tilde E_k}{\psi}$ can also be calculated systematically, so the time evolution of a (finite dimensional) quantum system can always be obtained in a systematic way for any Hamiltonian and any initial state.

The analogy to the situation of Liouville integrable systems is striking:
The dimension $d$ plays the role of the number $n$ of degrees of freedom of the system in the classical case.
The linear functionals $|\bra{\tilde E_k} \,\argdot\, | \oftype \mcH \to \R$, induced by the eigenvectors of $\H$, are analogous to the first integrals of motion in Liouville's theorem on integrable systems, and the time independent moduli of the overlaps $|\braket{\tilde E_k}{\psi}| = |\braket{\tilde E_k}{\psi(t)}|$ play the role of the values fixed for these constants of motion.
Finally, the functions $\tilde \varphi_k$ in the right hand side of \texteqref{eq:timeevolutioninexpliciteform} satisfy differential equations analogous to those of the action angle variables, namely $\dot{\tilde{\varphi}}_k = \tilde E_k$, and the time evolution indeed happens on a $d$-torus.
As in the classical case, the specific torus to which the evolution is confined depends on the values fixed for the conserved quantities.
It seems that the dynamics of quantum systems is far less rich than that of classical systems.
This constitutes a major obstacle for a good definition of quantum \mbox{(non-)}integrability.

Before going on, it is reasonable to give a set of conditions that a good notion of \mbox{(non-)}integrability for quantum systems should satisfy.
It seems reasonable to demand \cite{1012.3587v1} that a definition of quantum integrability should:
\begin{enumerate}%[label={(Condition~\arabic*)},ref={\arabic*},leftmargin=*]
\item \label{item:reasonablenotionofquantumintegrabilitycondition4} have implications for the physical behaviour,
\item \label{item:reasonablenotionofquantumintegrabilitycondition1} be applicable to a large class of quantum systems,
\item \label{item:reasonablenotionofquantumintegrabilitycondition2} be unambiguous,
\item \label{item:reasonablenotionofquantumintegrabilitycondition3} be decidable for concrete models.
\end{enumerate}

Unfortunately almost none of the existing frequently used notions of quantum integrability seems to fulfil all these criteria.
The following is a list of some of the definitions of quantum integrability that have been introduced, together with exemplary references in which the corresponding definition appears or is used (see also Refs.~\cite{1012.3587v1,PhysRevLett.10-6,sutherland04,Weigert1992}).
A system is \emph{quantum integrable}:
\begin{enumerate}%[leftmargin=*]
\item \label{item:npotionsofintegrability_conservedquantity} If it exhibits $n$ physically meaningful mutually commuting conserved quantities that are in some sense independent \cite{Rigol07,Braak11,1109.5904v1,Barthel08,Hawkins2008,Jensen1985} (see also Ref.~\cite{Weigert1992} and the references therein) or depend linearly on some parameter of the Hamiltonian \cite{1111.3375v1}.
\item \label{item:npotionsofintegrability_betheansatz} If it is integrable by the Bethe ansatz \cite{sutherland04,Ikeda2013a,Beugeling2013}.
\item \label{item:npotionsofintegrability_nondiffractive} If it exhibits nondiffractive scattering \cite{sutherland04}.
\item \label{item:npotionsofintegrability_classicallimit} If it has a classical limit that is integrable \cite{Castagnino2006}.
\item \label{item:npotionsofintegrability_levelstatistics} If its level statistics follows a Poisson law and is non-integrable if it is of Wigner-Dyson type \cite{Casati1985,1103.0787v1,PhysRevB.82.17,Znidari2013,Atas12,Tabor1989,Bohigas1984,Fine2013,Jensen1985}.
\item \label{item:npotionsofintegrability_nolevelrepulsion} If it does not exhibit level repulsion \cite{Stepanov2008,Berry1977a}.
\item \label{item:npotionsofintegrability_quantumnumbers} If (many of) its eigenfunctions can be labeled in a certain way with quantum numbers \cite{Braak11,Berry1977a}.
\item \label{item:npotionsofintegrability_exaxclysolvable} If it is exactly solvable in any way \cite{Beugeling2013,Fendley1995,Braak11,Jensen1985}.
\end{enumerate}
In the first definition both \emph{physically meaningful} and $n$ can have very different meanings.
It can, for example, in the case of composite systems, refer to local operators.
The number $n$ is usually taken to be equal to the number of degrees of freedom of the model or the number of constituents in the case of composite systems.
In Ref.~\cite{1111.3375v1} $n$ can be any number between zero and $d-1$, and models are then classified according to this number $n$.
Similarly, \emph{independent} can have several meanings, \emph{linearly independent} and \emph{algebraically independent} being popular choices.
Usually all quadratic systems and systems such as the Hydrogen atom fall in this category.
Models that are integrable according to this definition are often also integrable according to one of the other definitions given above (especially Definitions~\refitem{item:npotionsofintegrability_betheansatz}, \refitem{item:npotionsofintegrability_quantumnumbers}, and \refitem{item:npotionsofintegrability_exaxclysolvable}).
Many of the definitions of integrability of this type suffer from the severe problem that if the definition is taken seriously, all quantum systems classify as integrable and hence it violates Condition~\refitem{item:reasonablenotionofquantumintegrabilitycondition4} (see the discussion above and Refs.~\cite{Weigert1992,1111.3375v1} for a critic of such notions of integrability).

Definitions~\refitem{item:npotionsofintegrability_betheansatz}, \refitem{item:npotionsofintegrability_nondiffractive}, and \refitem{item:npotionsofintegrability_classicallimit} are only applicable to restricted classes of models and hence violate Condition~\refitem{item:reasonablenotionofquantumintegrabilitycondition1} in the above list.
The same holds, although arguably in a weaker sense, for Definition~\refitem{item:npotionsofintegrability_nolevelrepulsion} and the version of Definition~\refitem{item:npotionsofintegrability_conservedquantity} of from Ref.~\cite{1111.3375v1}, which are only applicable to systems which have a natural tuning parameter.

Definitions~\refitem{item:npotionsofintegrability_levelstatistics} and \refitem{item:npotionsofintegrability_nolevelrepulsion} suffer from the problem that also certain models that are usually regarded as integrable can have spectra that would classify them as non-integrable \cite{Benet2003,Berry1977a}.
In fact, it is trivial to construct such examples.
In a composite systems of, say, spin-1/2 systems, one can simply take a Hamiltonian that is diagonal in the usual Pauli-$Z$ product basis and which hence should clearly be classified as integrable and set its spectrum to be that of some non-integrable model.
Moreover, natural tunable models are known that exhibit thermodynamic behaviour in both the regime that would be classified as integrable and the one that would be classified as non-integrable according to this definition \cite{Jensen1985}.
Hence, these definitions violate Condition~\refitem{item:reasonablenotionofquantumintegrabilitycondition2} and \refitem{item:reasonablenotionofquantumintegrabilitycondition4}.

Especially Definitions~\refitem{item:npotionsofintegrability_conservedquantity} and \refitem{item:npotionsofintegrability_exaxclysolvable} suffer from the problem that it might simply be a lack of imagination that prevents one from finding a relevant conserved quantity or from solving a given model and thus violate Condition~\refitem{item:reasonablenotionofquantumintegrabilitycondition3}.
This is well illustrated by the recent (partial) solution of the Rabi model, which was long thought to be non-integrable (see Ref.~\cite{Braak11} and the references therein).

In conclusion, it seems fair to say that the question of how to define integrability in quantum mechanics is still to some extent open and even more so for quantum non-integrability.
At the same time a number of very useful and promising indicators of and proposals for a definition of integrability exist (see also Refs.~\cite{1111.3375v1,1012.3587v1} for more background information and recent proposals).
Still, general claims that ``non-integrable quantum systems thermalise'' seem unjustified at present.

\section{Decay of correlations and stability of thermal states}
\label{sec:propertiesofthermalstatesofcompositesystems}
In this section we will somewhat depart from the pure state quantum statistical mechanics approach, as we will take the canonical ensemble for thermal states for granted and turn to a study of structural properties of such thermal states.
This will bring us to the seemingly innocent question: What is the meaning of temperature on very small scales and in which sense is temperature really \emph{intensive}, as is paradigmatically claimed in thermodynamics?
The problem with assigning locally a temperature to a small subsystem of a global system in a thermal state is the following: 
Interactions between the subsystem and its environment generate correlations that can lead to noticeable deviations of the state of the subsystem from a thermal state.
Given only a subsystem state, there is no canonical way to assign a temperature to the subsystem.
We shall call this the \emph{locality of temperature problem}.

This problem has been addressed in Refs.~\cite{Hartmann2003,Hartmann2004,Hartmann2004a,0908.3157v3}, and more recently extensively studied in Ref.~\cite{Kliesch2014}.
There, three theorems are proven:
A \emph{truncation formula}, which allows to express the influence of sets of locally interacting Hamiltonian terms on the expectation value of an observable in the thermal state of a locally interacting quantum system in terms of a correlation measure.
A \emph{clustering of correlations result}, which shows that above a universal critical temperature this correlation measure exhibits an exponential decay.
And finally, a result that ensures \emph{local stability} of thermal states above a universal critical temperature and thereby partially solves the locality of temperature problem.

\subsection{Locality of temperature}
These results build upon and significantly go beyond previous results on clustering of correlations in classical systems \cite{Rue99_stat_mech_book,Rue64}, for quantum gases \cite{Gin65}, i.e., translational invariant Hamiltonians in the continuum, and cubic lattices \cite{BratteliRobinson1,Greenberg1969,Park-Yoo}.
For the latter systems the existence and uniqueness of thermal states in the thermodynamic limit at high temperatures is proven and analyticity of correlations can be derived.
Moreover, in the regime of high temperatures, $n$-point correlation functions have been shown to cluster for spin gases \cite{Rue64,Gin65} and translational invariant bosonic lattices \cite{Park-Yoo}.

To begin the more detailed discussion, we introduce a quantity that measures correlations.
We define for any $\tau \in [0,1]$, any two operators $A,B \in \Bop(\mcH)$, and any quantum state $\rho \in \Qst(\mcH)$ the \emph{generalised covariance}
\begin{equation} \label{eq:generalisedcovariance}
 \cov_\rho^\tau(A,B) 
 \coloneqq \Tr(\rho^\tau A\, \rho^{1-\tau} B) - \Tr(\rho\, A) \Tr(\rho \, B) \, .
\end{equation}
The choice $\tau = 1$ gives the usual covariance\footnote{For the fine print see \cite{Kliesch2014}.}.
As a side remark, the quantity $\cov_\rho^\tau(A,B)$ also appears in studies of one dimensional models \cite{Cardy1996a}, where it can be written as a different times correlation function in terms of the transfer matrix of the system.

The reason for introducing the general definition here is that it naturally appears in the \emph{truncation formula}.
Before we can state it we need one last piece of notation.
For any subsystem $X \subset \Vset$ let $X_\partial \subset \Eset$ the set of edges that overlap with both $X$ and its complement, i.e.,
\begin{equation}
  X_\partial \coloneqq \{ Y \in \Eset\oftype Y \intersection X \neq \emptyset \land Y \intersection \compl X \neq \emptyset \} .
\end{equation}
We extend this notation to operators $A \in \Bop(\mcH)$ and define
\begin{equation}
  A_\partial \coloneqq \{ Y \in \Eset\oftype Y \intersection \supp(A) \neq \emptyset \land Y \intersection \compl\supp(A) \} .
\end{equation}
\begin{theorem}[Truncation formula {\cite[Corollary~1 and 4]{Kliesch2014}}] \label{thm:truncationformula}
  Consider a spin or fermionic system with Hilbert space $\mcH$ and let $\H \in \Obs(\mcH)$ be a locally interacting Hamiltonian with edge set $\Eset$.
  Let $B \subset \Eset$ and define for $s \in [0,1]$ the interpolating Hamiltonian $\H(s) \coloneqq \H - (1-s) \sum_{X \in B_\partial} \H_X$.
  Then, for any operator $A \in \Bop(\mcH)$ with $\supp(A) \subset B$ it holds that
  \begin{equation}\label{eq:truncation_error_in_terms_of_cov}
    \Tr\bigl(A\,\rhog[\H_B](\beta)\bigr)-\Tr\bigl(A\,\rhog[\H](\beta)\bigr) = \beta\,\int_0^1  \int_0^1 \cov_{\rhog[\H(s)](\beta)}^\tau(A, \sum_{X \in B_\partial} \H_X)\,\dd\tau\,\dd s .
  \end{equation}
\end{theorem}
The left hand side of \texteqref{eq:truncation_error_in_terms_of_cov} is the difference between the expectation value of $A$ in the thermal states of the Hamiltonian $\H_B$ with only the terms contained in the region $B$ and the full Hamiltonian $\H$.
The truncation formula quantifies how the expectation value of $A$ changes when the terms that couple $B$ to the rest of the system are added or removed, hence the name, and tells us that this change can be expressed exactly in terms of the generalised covariance.

It is important to note that \texteqref{eq:truncation_error_in_terms_of_cov} is an \emph{equality}.
The generalised covariance exactly captures the response of expectation values in the thermal state to truncations of the Hamiltonian.
The truncation formula tells us that the response of the expectation value is small if and only if the right hand side of \texteqref{eq:truncation_error_in_terms_of_cov}, which is an average over the generalised covariance times $\beta$, is small.
In other words:
\begin{observation}[Locality of temperature \cite{Kliesch2014}] 
  Temperature can be defined locally on a given length scale if and only if the averaged generalised covariance is small compared to $1/\beta$ on that length scale.
\end{observation}

We will see shortly that if locally interacting spin or fermionic lattice systems are at a sufficiently high temperature, then the generalised covariance $\cov^\tau_{\rhog(\beta)}(A,B)$ between any to operators $A,B \in \Bop(\mcH)$ decays exponentially with the graph distance $\dist(A,B)$ between their supports.

\subsection{Clustering of correlations in high temperature thermal states}

The following theorem applies to all Hamiltonians whose interaction (hyper)graph has a finite \emph{growth constant}.
To explain what this means we need some additional notation.
A subset $F \subset \Eset$ of the edge set \emph{connects} $X$ and $Y$ if $F$ contains all elements of some sequence of pairwise overlapping edges such that the first overlaps with $X$ and the last overlaps with $Y$ and similarly for sites $x,y \in V$.
A subset $F \subset \Eset$ of the edge set $\Eset$ that connects all pairs of its elements is called \emph{connected} and connected subsets $F$ are also called \emph{animals} \cite{Miranda2011,Penrose1994}.
The size $|F|$ of an animal $F$ is the number of edges it contains.
It turns out that for many interesting (hyper)graphs the number of animals of a given size that contain a given edge grows exponentially with the size, but not faster.
That is, they have a finite \emph{growth constant}.
More precisely, the growth constant of a (hyper)graph $\mcG = (\Vset,\Eset)$ is the smallest constant $\animalc$ satisfying
\begin{equation}\label{eq:animal_bound}
  \forall k\in \Z^+\itholds \sup_{X \in \Eset} |\{ F\subset\Eset \text{ connected}\oftype X \in F \land |F|=k \}| \leq \animalc^k .
\end{equation}
For example, the growth constant $\animalc$ of the interaction graph of nearest neighbour Hamiltonians on $D$ dimensional cubic lattices can be bounded by $2\,D\,\e$ (see Lemma~2 in Ref.~\cite{Miranda2011}).
Moreover, there is a finite growth constant $\animalc$ for any regular lattice \cite{Penrose1994}, and there exist upper bounds on the growth constants of so-called spread-out graphs \cite{Miranda2011} that make it possible to bound the growth constant of the interaction hypergraphs of all $l$-local $k$-body Hamiltonians on regular lattices \cite{Kliesch2014}.
Where $l$-local $k$-body on a regular lattice means that $\Vset$ can be mapped onto the sites of a regular lattice such that $\Eset$ contains only subsystems which consist of at most $k$ sites that are all contained in a ball (measured in the graph distance of the regular lattice) of diameter $l$.
Apart from all $l$-local $k$-body Hamiltonians on regular lattices this also makes the following results indirectly applicable to systems with exponentially decaying interactions (such Hamiltonians can be exponentially well approximated by $l$-local $k$-body Hamiltonians) but not to Hamiltonians with algebraically decaying interactions, such as for example Coulomb or dipole interactions.
We can now state the clustering of correlations result:

\begin{theorem}[Clustering of correlations at high temperature {\cite[Theorem~2 and 4]{Kliesch2014}}] \label{thm:clustering}
  Consider a locally interacting system of spins or fermions with Hilbert space $\mcH$ and Hamiltonian $\H \in \Obs(\mcH)$ with \emph{local interaction strength} $J \coloneqq \max_{X\in\Eset} \norm[\infty]{\H_X}$ and interaction (hyper)graph $\mcG = (\Vset,\Eset)$ with growth constant $\animalc$.
  Define the \emph{critical temperature}
  \begin{equation}\label{eq:crit_beta_def}
    \beta^\ast \coloneqq \ln\big((1+\sqrt{1+4/\animalc})/2\big)/(2\,J) 
  \end{equation}
  and the \emph{thermal correlation length}
  \begin{equation} \label{eq:correlation_length_def}
    \xi(\beta) \coloneqq \left|1/\ln\left(\animalc\, \e^{2\,|\beta|\,J}(\e^{2\,|\beta|\,J}-1)\right)\right|  \, .
  \end{equation}
  Then, for every $|\beta|<\beta^\ast$, parameter $\tau \in [0,1]$, and every two operators $A,B \in \Bop(\mcH)$ with 
  $\dist(A, B) \geq \xi(\beta)\, \left|\ln\left(\ln(3)\, (1-\e^{-1/\xi(\beta)})/\min(|A_\partial |,|B_\partial |)\right)\right|$ ,
  \begin{equation} \label{eq:clustering} 
    |\cov^\tau_{\rhog(\beta)} (A, B)| \leq \frac{4\,\min(|A_\partial |,|B_\partial |) \norm{A}_\infty\, \norm{B}_\infty}{\ln(3)\, (1-\e^{-1/\xi(\beta)})}\,\e^{-\dist(A,B)/\xi(\beta)} .
  \end{equation}
\end{theorem}

The above theorem implies that in thermal states above the critical temperature the correlations between any two $A,B \in \Bop(\mcH)$ decay exponentially with their distance $\dist(A,B)$.
Importantly, the critical temperature \eqref{eq:crit_beta_def} is independent of global properties of $\H$ but only depends on the local interaction strength $J$ and the growth constant $\animalc$ of its interaction (hyper)graph.

In the context of this work, the most important implication of Theorem~\ref{thm:clustering} is the following result, which proves stability of thermal states above the critical temperature against local perturbations.
More precisely, it shows that changing the Hamiltonian of a locally interacting quantum system only outside of a subsystem $B$ has only limited influence on how thermal states to temperatures above the critical temperature look like in the interior $S \subset B$ of $B$ if the distance between $S$ and $B_\partial$ is large enough:

\begin{theorem}[Universal locality at high temperatures {\cite[Corollary~2 and 5]{Kliesch2014}}]\label{thm:intensivity}
  Let $\H$ be a Hamiltonian satisfying the conditions of Theorem~\ref{thm:clustering},
  let $\beta^\ast$ and $\xi(\beta)$ be defined as in \texteqref{eq:crit_beta_def} and \texteqref{eq:correlation_length_def},
  let $|\beta|< \beta^\ast$, and let $S \subset B \subseteq \Vset$ be subsystems with 
  $\dist(S, B_\partial) \geq \xi(\beta)\, \left|\ln\left(\ln(3)\, (1-\e^{-1/\xi(\beta)})/|S_\partial|\right)\right|$.
  Then
  \begin{equation}\label{eq:stabilityofthermalstatesbound}
    \tracedistance{\rhog^{S}[\H](\beta)}{\rhog^{S}[H_B](\beta)} \leq \frac{ v\, |\beta|\, J }{1-\e^{-1/\xi(\beta)}}\,\e^{- \dist(S, B_\partial) /\xi(\beta)} ,
  \end{equation}
  where $v \coloneqq 4\, |S_\partial|\,|B_\partial|/\ln(3)$.
\end{theorem}

If the conditions of the above theorem are met and the interior subsystem $S$ is sufficiently far from the boundary $B_\partial$ of $B$ such that $\dist(S, B_\partial)$ is large and hence the right hand side of \texteqref{eq:stabilityofthermalstatesbound} small, then the reduced state $\rhog^{S}[\H](\beta)$ on $S$ of the thermal state of $\H$ is almost independent of the terms of the Hamiltonian $\H$ that are not in the restricted Hamiltonian $\H_B$.

Theorem~\ref{thm:intensivity} is not unexpected, but it is nevertheless remarkable that it can be shown in this generality for systems of both locally interacting spins and fermions.
Even more so, because, as we have seen in the discussion of equilibration (Section~\ref{sec:equilibration}) and especially in the section on equilibration time scales, a major obstacle for improving the statements we were able to make is that it seems to be hard to use the structure of natural many-body Hamiltonians, namely that interactions are usually few body and often short range.
Theorem~\ref{thm:intensivity} is an instance of a result whose proof heavily relies on the locality structure of locally interacting Hamiltonians and is able to exploit their structure.

It is interesting to plug in the numbers of a specific model to see how physical the derived critical temperature is.
As a concrete example consider the ferromagnetic two dimensional isotropic Ising Model without external field.
The critical temperature of Theorem~\ref{thm:clustering} and \ref{thm:intensivity} is 
$1/(\beta^\ast\,J) = 2/\ln((1+\sqrt{1+1/\e})/2) \approx 24.58$, whereas the \emph{Curie temperature}, i.e., the temperature at which the phase transition between the paramagnetic and the ferromagnetic phase happens is known to be $1/(\beta_c\,J) = 2/\ln(1+\sqrt{2}) \approx 2.27$ \cite{Bhattacharjee1995}.
To put this into perspective however, it is worth noting that the above theorem still improves upon previously known bounds like for example that implied by Ref.~\cite{Greenberg1969}, which yields $1/(\beta^\ast_c\,J) = 124$ and that it is a \emph{universal upper bound} independent of details of the particular model.
Given how difficult it is to calculate or even bound critical temperatures in lattice models (both classical and quantum) and that good bounds are known only for very few models the existence of such a non-trivial and universal upper bound is quite remarkable.

Besides being of fundamental interest, Theorem~\ref{thm:intensivity} has some obvious computational implications:
It implies that for all $|\beta| < \beta^\ast$ reduced states of thermal states can be approximated with a computational cost independent of the system size and sub-exponential in the reciprocal approximation error (polynomially for systems in one dimension) \cite{Kliesch2014}.
The proof of Theorem~\ref{thm:clustering} is based on a \emph{cluster expansion} (see Lemma~6 in Ref.~\cite{Kliesch2014}) previously used in Ref.~\cite{Hastings06} to show that thermal states above a critical temperature can be approximated by so-called \emph{matrix product operators} (MPOs).
The subtleties of this approximation are often misunderstood.
For details see the appendix of Ref.~\cite{Kliesch2014}.

\section{Conclusions}
\label{sec:conclusions}
In this review, we have elaborated on a question that is at the heart of the foundations of quantum statistical mechanics:
This is the question of how pure states evolving unitarily according to the Schr{\"o}dinger equation can give rise to a wealth of phenomena that can rightfully be called \emph{thermodynamic}.
Individual observables and entire subsystems have a tendency to evolve towards equilibrium values/states and then stay close to them for most times during the evolution or extended time intervals.
It turns out that the equilibrium properties can be captured by suitable maximum entropy principles implied by quantum mechanical dynamics alone.
If a part of the system can be naturally identified as a bath and its complement as a distinguished subsystem, a weak interaction naturally leads to decoherence in the energy eigenbasis, and under additional conditions even equilibration to a thermal state can be guaranteed.
We have also discussed properties of thermal states in lattice systems and in particular elaborated on precise conditions under which correlations decay exponentially.
We have also reviewed systems where an absence of thermalisation is anticipated and the role played by many-body localisation played in this context.

Notions of information propagation as well as entanglement and correlation dynamics play key roles in processes of equilibration and thermalisation.
Complementing these dynamical approaches, the immensely large dimension of the Hilbert space of composite quantum systems can also justify the applicability of statistical ensembles via typicality arguments.

It goes without saying that we have only touched the tip of the iceberg:
Many key questions had to be left unmentioned, despite the considerable length of the article.
This is particularly regrettable with respect to the exciting experimental developments that have taken place over the recent years and now allow us to probe the questions at hand under remarkably precise conditions.
The high degree of control offered by such experiments makes it possible to use them as quantum simulators assessing features quantitatively that are outside the scope of present analytical or numerical approaches.

At the same time, one reason for why many questions can not be satisfactory discussed here is that many key problems actually remain wide open, despite the enormous progress surveyed here.
The question of what time scales are to be expected in equilibration is just as open as is a full understanding of thermalisation.
And here the present review reveals its most important purpose: To serve as an invitation to this exciting field of research.

\section{Acknowledgements}
\label{sec:acknowledgements}
We would like to sincerely thank numerous colleagues for dis- cussions over the years and the EU (RAQUEL, SIQS, AQuS), the ERC (TAQ, OSYRIS, FP7-PEOPLE-2013-COFUND), the BMBF, the Studienstiftung des Deutschen Volkes, MPQ- ICFO, the Spanish Ministry Project FOQUS (FIS2013- 46768-P), the Generalitat de Catalunya (SGR 874 and 875), the Spanish MINECO (Severo Ochoa grant SEV-2015-0522), the Fundaci\'{o} Privada Cellex, MPQ-ICFO, the EU's Marie Sk\l{}odowska-Curie IF program (GA: 700140), and the COST Action MP1209 for support.

%\appendix % As the style doesn't handle appendices properly we do it manually
\section*{Appendix}
\addcontentsline{toc}{section}{Appendix}
\setcounter{equation}{0}
\setcounter{section}{0}
\setcounter{subsection}{0}
\setcounter{subsubsection}{0}
\renewcommand\thesection{\Alph{section}}
\renewcommand\theequation{A\arabic{equation}}

\section{Remarks on the foundations of statistical mechanics}
\label{chap:remarksonthefoundationsofstatmech}
In this appendix we briefly sketch the most influential canonical approaches towards the foundations of thermodynamics and statistical mechanics.
We will roughly follow the historical development, but emphasise more the problems and shortcomings of the respective approaches rather then their undeniable success and ingenuity.

Contrary to the rest of this work, this appendix is rather superficial.
The main justification for the brevity is the existence of several comprehensive works on the topic, in particular the review by Uffink \cite{UffinkFinal} and the book by Sklar \cite{Sklar1995}, but also Refs.~\cite{RevModPhys.27.289,Ehrenfest2002,Penrose1979} and Chapter~4 in Ref.~\cite{Gemmer09}.
Adding yet another work to this list simply seems superfluous and a detailed review of the history of statistical mechanics is beyond the scope of this work.
Also, we will brush over many of the more subtle issues of the classical approaches, such as the interpretation of probability and the problem of comparing discrete and continuous measures.

The intention of this appendix is to partially answer the legitimate question of a person already familiar with thermodynamics and statistical physics: ``Why should I care about pure state quantum statistical mechanics? Weren't all the foundational questions already solved in the works from the 19th and early 20th century?''
As we will see, despite the numerous attempts and the great amount of work that has been put into establishing a convincing justification for the methods of statistical mechanics it has ``not yet developed a set of generally accepted formal axioms'' \cite{UffinkFinal}, or, as E.\ T.\ Jaynes \cite{Jaynes} puts it: ``There is no line of argument proceding from the laws of microscopic mechanics to macroscopic phenomena that is generally regarded by physicists as convincing in all respects.''

\subsection{Canonical approaches}
\label{sec:canonicalapproaches}
Thermodynamics was originally developed as a purely phenomenological theory.
Prototypical for this era are the laws of \emph{Boyle--Mariotte} and \emph{Gay--Lussac} that state empirically observed relations between the volume, pressure, and temperature of gases.

The more widespread acceptance of the \emph{atomistic hypothesis} in the 18th century opened up the way for a microscopic understanding of such empirical facts.
The works of Clausius \cite{Clausius1857}, Maxwell \cite{Maxwell1860,Maxwell1860a}, Boltzmann \cite{Boltzmann1872}, and Gibbs \cite{Gibbs1902} in the second half of the 19th and the beginning of the 20th century are often perceived as the inception of statistical mechanics (see also Refs.~\cite{Boltzmann1896,UffinkFinal,Sklar1995}).
In this section we review some of these early attempts to develop a deeper understanding of thermodynamics based on microscopic considerations.

\subsubsection{Boltzmann and the H-Theorem}
\label{sec:boltzmannshtheorem}
One of Boltzmann's arguably most important contributions to the development of statistical mechanics is his derivation of what is known today as the \emph{Boltzmann equation} and his \emph{H-theorem} \cite{Boltzmann1872} (see also the first chapter of Boltzmann's book ``Vorlesungen {\"u}ber Gastheorie. Bd. 1.'' \cite{Boltzmann1896} as well as Ref.~\cite[Chapter~4]{Gemmer09} and Ref.~\cite{Sklar1995}).

\begin{figure}[bt]
  \centering
  \begin{itemize}  \setlength{\itemsep}{0.8cm} 
  \item[(a)] Time reversal objection (Loschmidt)
    \begin{center}
      \begin{tikzpicture}[scale=0.8,baseline=(current bounding box.base)]
        \node (p1) at (-0.9, 0.65) {};
        \node (p2) at (-0.3, 0.4 ) {};
        \node (p3) at (-0.2, 0.1 ) {};
        \node (p4) at (-0.8,-0.3 ) {};
        \node (p5) at (-0.4,-0.7 ) {};
        \node (q1) at (-0.4, 0.5 ) {};
        \node (q2) at (-0.6, 0.2 ) {};
        \node (q3) at (-0.4, 0.3 ) {};
        \node (q4) at (-0.2,-0.0 ) {};
        \node (q5) at (-0.3,-0.2 ) {};
        \draw[black,fill=white,rounded corners] (-1,-1) rectangle (1,1);
        \foreach \x in {1,2,...,5} {\draw[->,fill=gray] (p\x) circle (2pt);};
        \draw[dashed,gray] (0,-1) -- (0,1);        
      \end{tikzpicture}
      $\overset{t}{\longrightarrow}$
      \begin{tikzpicture}[scale=0.8,baseline=(current bounding box.base)]
        \node (p1) at (-0.7, 0.7) {};
        \node (p2) at ( 0.1, 0.4 ) {};
        \node (p3) at (-0.2, 0.1 ) {};
        \node (p4) at ( 0.6,-0.3 ) {};
        \node (p5) at (-0.4,-0.7 ) {};
        \node (q1) at (-0.4, 0.5 ) {};
        \node (q2) at (-0.6, 0.2 ) {};
        \node (q3) at (-0.4, 0.3 ) {};
        \node (q4) at ( 0.2,-0.0 ) {};
        \node (q5) at (-0.3,-0.3 ) {};
        \draw[black,fill=white,rounded corners] (-1,-1) rectangle (1,1);
        \foreach \x in {1,2,...,5} {\draw[->,fill=gray] (p\x) circle (2pt) -- ++(q\x);};
      \end{tikzpicture}
      \quad $\vec{v} \rightarrow -\vec{v}$ \quad
      \begin{tikzpicture}[scale=0.8,baseline=(current bounding box.base)]
        \node (p1) at (-0.7, 0.7) {};
        \node (p2) at ( 0.1, 0.4 ) {};
        \node (p3) at (-0.2, 0.1 ) {};
        \node (p4) at ( 0.6,-0.3 ) {};
        \node (p5) at (-0.4,-0.7 ) {};
        \node (q1) at (-1.0, 0.9 ) {};
        \node (q2) at ( 0.7, 0.6 ) {};
        \node (q3) at (-0.0,-0.1 ) {};
        \node (q4) at ( 1.0,-0.6 ) {};
        \node (q5) at (-0.5,-1.1 ) {};
        \draw[black,fill=white,rounded corners] (-1,-1) rectangle (1,1);
        \foreach \x in {1,2,...,5} {\draw[->,fill=gray] (p\x) circle (2pt) -- ++(q\x);};
      \end{tikzpicture}
      $\overset{t}{\longrightarrow}$
      \begin{tikzpicture}[scale=0.8,baseline=(current bounding box.base)]
        \node (p1) at (-0.9, 0.65) {};
        \node (p2) at (-0.3, 0.4 ) {};
        \node (p3) at (-0.2, 0.1 ) {};
        \node (p4) at (-0.8,-0.3 ) {};
        \node (p5) at (-0.4,-0.7 ) {};
        \draw[black,fill=white,rounded corners] (-1,-1) rectangle (1,1);
        \foreach \x in {1,2,...,5} {\draw[->,fill=gray] (p\x) circle (2pt);};
      \end{tikzpicture}
    \end{center}
  \item[(b)] Recurrence objection (Poincar\'{e}, Zermelo)
    \begin{center}
      \begin{tikzpicture}[scale=0.8,baseline=(current bounding box.base)]
        \node (p1) at (-0.9, 0.65) {};
        \node (p2) at (-0.3, 0.4 ) {};
        \node (p3) at (-0.2, 0.1 ) {};
        \node (p4) at (-0.8,-0.3 ) {};
        \node (p5) at (-0.4,-0.7 ) {};
        \node (q1) at (-0.4, 0.5 ) {};
        \node (q2) at (-0.6, 0.2 ) {};
        \node (q3) at (-0.4, 0.3 ) {};
        \node (q4) at (-0.2,-0.0 ) {};
        \node (q5) at (-0.3,-0.2 ) {};
        \draw[black,fill=white,rounded corners] (-1,-1) rectangle (1,1);
        \foreach \x in {1,2,...,5} {\draw[->,fill=gray] (p\x) circle (2pt);};
        \draw[dashed,gray] (0,-1) -- (0,1);        
      \end{tikzpicture}
      $\overset{t}{\longrightarrow} \dots \overset{t}{\longrightarrow}$
      \begin{tikzpicture}[scale=0.8,baseline=(current bounding box.base)]
        \node (p1) at (-0.7, 0.7) {};
        \node (p2) at ( 0.1, 0.4 ) {};
        \node (p3) at (-0.2, 0.1 ) {};
        \node (p4) at ( 0.6,-0.3 ) {};
        \node (p5) at (-0.4,-0.7 ) {};
        \draw[black,fill=white,rounded corners] (-1,-1) rectangle (1,1);
        \foreach \x in {1,2,...,5} {\draw[->,fill=gray] (p\x) circle (2pt);};
      \end{tikzpicture}      
      $\overset{t}{\longrightarrow} \dots \overset{t}{\longrightarrow}$
      \begin{tikzpicture}[scale=0.8,baseline=(current bounding box.base)]
        \node (p1) at (-0.9, 0.65) {};
        \node (p2) at (-0.3, 0.4 ) {};
        \node (p3) at (-0.2, 0.1 ) {};
        \node (p4) at (-0.8,-0.3 ) {};
        \node (p5) at (-0.4,-0.7 ) {};
        \node (q1) at (-0.4, 0.5 ) {};
        \node (q2) at (-0.6, 0.2 ) {};
        \node (q3) at (-0.4, 0.3 ) {};
        \node (q4) at (-0.2,-0.0 ) {};
        \node (q5) at (-0.3,-0.2 ) {};
        \draw[black,fill=white,rounded corners] (-1,-1) rectangle (1,1);
        \foreach \x in {1,2,...,5} {\draw[->,fill=gray] (p\x) circle (2pt);};
      \end{tikzpicture}
    \end{center}    
  \end{itemize}
  \caption{(Reproduction from Ref.~\cite{Gogolin2014}) The \emph{time reversal objection}, also known as \emph{Loschmidt's paradox} \cite{Loschmidt1877}, but actually first published by William Thomson \cite{ThomsonLordKelvin1874}, states that it should not be possible to deduce time reversal asymmetric statements like the H-theorem, implied by the Boltzmann equation, from an underlying time reversal invariant theory.
More explicitly, it argues that for any process that brings a system into an equilibrium state starting from a non-equilibrium situation, there exists an equally physically allowed reverse process that takes the system out of equilibrium.
The initial state for that process is obtained from the equilibrium state by reversing all velocities (see Panel~(a)).
The \emph{recurrence objection}, which is based on the \emph{Poincar\'{e} recurrence theorem} but was made explicit by Zermelo \cite{Zermelo1896}, states that Boltzmann's H-theorem is in conflict with Hamiltonian dynamics, because it can be proven on very general grounds that all finite systems are recurrent, i.e., return arbitrarily close to their initial state after possibly very long times (see Panel~(b)).}
  \label{fig:timereversalandrecurrenceobjection}
\end{figure}

In his 1872 article \cite{Boltzmann1872} Boltzmann aims at showing that the \emph{Maxwell-Boltzmann distribution} is the equilibrium distribution of the speed of gas particles and that a gas with an initially different distribution must inevitably approach it.
He tries to do this on the grounds of microscopic considerations and starts off from the prototypical model of the hard sphere gas.
He takes for granted that in equilibrium the distribution of the particles should be ``uniform'' and that their speed distribution should be independent of the direction of movement.
He assumes that the number of particles is large and introduces a continuously differentiable function he calls ``distribution of state''\footnote{German original \cite{Boltzmann1872}: \foreignlanguage{ngerman}{``Zustandsverteilung''.}}, which is meant to approximate the (discrete) distribution of the speed of the particles.
He then derives a differential equation for the temporal evolution of this function, known today as the \emph{Boltzmann equation}.
He also defines an entropy for the ``distribution of state'' and shows that it increases monotonically in time under the dynamics given by the Boltzmann equation, a statement he calls \emph{H-Theorem}, after the letter $H$ used for denoting the entropy.

During the derivation he makes several approximations.
Essential is his ``Sto{\ss}zahl Ansatz'', later dubbed the ``hypothesis of molecular disorder'' in Ref.~\cite{Boltzmann1896}, which explicitly breaks the time reversal invariance of classical mechanics.
This breaking of the time reversal symmetry is responsible for the temporal increase of entropy.
Naturally this assumption has been much criticised.
Famous are the \emph{time reversal objection} of William Thomson and Loschmidt and the \emph{recurrence objection} due to Poincar\'{e} and Zermelo \cite{Sklar1995} (see Figure~\ref{fig:timereversalandrecurrenceobjection}).
The bottom line of this debate, also later acknowledged by Boltzmann \cite{Boltzmann1896a}, is that any statement that implies the convergence of a finite system to a fixed equilibrium state/distribution in the limit of time going to infinity is incompatible with a time reversal invariant or recurrent microscopic theory.
This is important for the notions of equilibration we discuss in Section~\ref{sec:equilibration}.

\subsubsection{Gibbs'  ensemble approach}
\label{sec:gibbsensembleapproach}
For many, Gibbs' book ``Elementary principles in statistical mechanics'' \cite{Gibbs1902} from 1902 marks the birth of modern statistical mechanics \cite{UffinkFinal}.
Central in Gibbs' approach is the concept of an \emph{ensemble}, which he describes as follows:
``We may imagine a great number of systems of the same nature, but differing in the configurations and velocities which they have at a given instant [\dots] we may set the problem, not to follow a particular system through its succession of configurations, but to determine how the whole number of systems will be distributed among the various conceivable configurations and velocities at any required time [\dots]''

In fact, the book then is not so much concerned with (non-equilibrium) dynamics, but rather with the calculation of statistical equilibrium averages.
Gibbs considers systems whose phase space is, as in Hamiltonian mechanics, spanned by canonical coordinates and introduces the \emph{micro-canonical}, \emph{canonical}, and \emph{grand canonical} ensemble for such systems.
He assumes that the number of states is high enough such that a description with a, as he calls it, ``structure function'', a kind of density of states, is possible.
He shows how various thermodynamic relations for quantities such as temperature and entropy can be reproduced from his ensembles, if these quantities are properly defined in terms of the structure function.

Gibbs is mostly concerned with defining recipes for the description of systems in equilibrium.
He gives little insight into why the ensembles he proposes capture the physics of thermodynamic equilibrium or how and why systems equilibrate in the first place \cite{UffinkFinal}.
Instead of addressing such foundational questions he is ``contented with the more modest aim of deducing some of the more obvious propositions relating to the statistical branch of mechanics'' \cite{Gibbs1902}.

\subsubsection{(Quasi-)ergodicity}
\label{sec:ergodicity}
The \emph{ergodicity hypothesis} was essentially born out of the incoherent use of different interpretations of probability by Boltzmann in his early work \cite{Boltzmann1868} and was formulated by him in Ref.~\cite{Boltzmann1871} as follows:
``The great irregularity of the thermal motion and the multitude of forces that act on a body make it probable that its atoms, due to the motion we call heat, traverse all positions and velocities which are compatible with the principle of [conservation of] energy.''\footnote{The English translation is taken from Ref.~\cite{UffinkFinal}.}
The concept of ergodicity was made prominent by P.\ and T.\ Ehrenfest in Ref.~\cite{Ehrenfest2002}, who proposed the \emph{ergodic foundations of statistical mechanics} \cite{UffinkFinal}.

Roughly speaking, a system is called \emph{(quasi-)ergodic} if it explores its phase space uniformly in the course of time for most initial states.
Making precise what ``uniformly'', ``most'', and ``in the course of time'' mean in this context already constitutes a major challenge \cite{UffinkFinal}.
However, if one is willing to believe that a system at hand is ergodic in an appropriate sense then it readily follows that (infinite time) temporal averages of physical quantities in that system are (approximately and/or with ``high probability'') equal to certain phase space averages, such as for example that given by the micro-canonical ensemble.

The ergodic foundations of statistical mechanics are then roughly based on arguments along the following lines:
Any physical measurement must be carried out during a finite time interval.
What one actually observes is not an instantaneous value, but an average over this time span.
The relevant time spans might seem short on a human time scale, but can at the same time be ``close to infinite'' compared to the microscopic time scales.
Think for example of the process of measuring the pressure in a gas container with a membrane.
The moment of inertia of the membrane is much too large to observe the spikes in the force due to hits by individual particles.
It is thus reasonable to assume that observations are well described by (infinite time) averages of the corresponding quantities, which, if the system is quasi-ergodic, can be calculated by averaging in an appropriate way over phase space.

The arguably most striking objection against such reasoning is the following \cite{Sklar1995}:
If it were in fact true that all realistic measurements could legitimately be described as infinite time averages, then the observation of any non-equilibrium dynamics, including the approach to equilibrium, would simply be impossible.
The latter is manifestly not the case.

Besides this issue of the ``infinite time'' averages and the other problems mentioned above it is extraordinarily difficult to show that a given system is \emph{(quasi-)ergodic}.
Despite the ground breaking works of Birkhoff and von Neumann on the concept of \emph{metric transitivity}, Sinai's work on \emph{dynamical billiards}, and more recent approaches such as \emph{Khinchin's ergodic theorem}, the full problem still awaits solution \cite{UffinkFinal}.

\subsubsection{Jaynes' maximum entropy approach}
\label{sec:jaynesmaximumentropyapproach}
Conceptually very different from the three previously discussed approaches is the work of Jaynes \cite{Jaynes}.
He fully embraces a subjective interpretation of probability and proposes to regard statistical physics as a ``form of statistical inference rather than a physical theory''.
He then introduces a \emph{maximum entropy principle}.
In short, the maximum entropy principle states that in situations where the existing knowledge is insufficient to make definite predictions the best possible predictions can be reached by finding the distribution of the state space of the system that maximises the (Shannon) entropy and is compatible with the available knowledge.
The principle is inspired by the work of Shannon \cite{Shannon1949} who, as Jaynes claims, had shown that the maximum entropy distribution is the one with the least bias towards the missing information \cite{Jaynes}: ``[The] maximum entropy distribution may be asserted for the positive reason that it is uniquely determined as the one which is maximally noncommittal with respect to missing information.''

Moreover, in Ref.~\cite{Jaynes}, Jaynes shows in quite some generality that the ``usual computational rules [as presented in Gibbs' book \cite{Gibbs1902}] are an immediate consequence of the maximum entropy principle''.
In addition, he points out various other advantages of his subjective approach.
For example that it makes predictions ``only if the available information is sufficient to justify fairly strong opinions'', and that it can account for new information in a natural way.

While Jaynes' principle can be used to justify the methods of statistical mechanics it gives little insight into why and under which conditions these methods yield results that agree with experiments.
In other words: The maximum entropy principle ensures that making predictions based on statistical mechanics is ``best practice'', but does not explain why this ``best practice'' is often good enough.
The question ``Why does statistical mechanics work?'' hence remains partially unanswered.

A last point of criticism is that Ref.~\cite{Jaynes} works in a classical setting.
While an extension to quantum mechanics is possible \cite{PhysRev.108.17} the subjective interpretation of probability advertised by Jaynes is arguably less convincing or at least debatable in this setting, although this is of course to some extend a matter of taste \cite{Fuchs,Timpson2008}.
Problems arise because mixed quantum states can be written as convex combinations of pure states in more than one way so that more complicated arguments are needed to identify the von Neumann entropy as the right entropy measure to be maximised.

\subsection{Closing remarks}
\label{sec:somegeneralcriticalremarks}
Except for Jaynes subjective maximum entropy principle, all approaches we have discussed in this chapter differ in one important point from that advertised in the main part of this review:
They are based on classical mechanics.
The applicability of classical models to systems that behave thermodynamically is, however, questionable.

Consider for example two of the most prominently used models in statistical mechanics: 
The hard sphere model for gases and the Ising model for ferromagnetism.
The atoms and molecules of a gas, as well as the interactions between them, in principle require a quantum mechanical description.
Yet, it is often claimed that in the so-called \emph{Ehrenfest limit}, i.e., if the spread of the quantum mechanical wave packets of the individual particles is small compared to the ``radius'' of the particles, the classical hard sphere approximation is eligible.
It can, however, be shown that under reasonable conditions systems typically leave the Ehrenfest limit on timescales much shorter than those of usual thermodynamic processes \cite[Chapter~4]{Gemmer09}.
Moreover, whether the Ehrenfest limit constitutes a sufficient condition for the applicability of (semi-)classical approximations in the first place is debatable \cite{Ballentine1994}.
Similarly, the relevant elementary magnetic moments of a piece of iron, namely the electronic spins, are intrinsically quantum.
In fact, it is known that classical physics alone cannot explain the phenomenon of ferromagnetism in a satisfactory way --- a statement known as \emph{Bohr--van Leeuwen theorem} \cite{Bohr1911,Aharoni2000,Nolting2009}.
The extremely simplified description employed in the Ising model can thus, despite its pedagogical value, arguably not capture all the relevant physics.

In addition to this, there are many situations where thermodynamic behaviour cannot be understood in a purely classical framework \cite{greiner}: For example, black-body radiation cannot be understood without postulating a quantisation of energy to avoid the \emph{ultraviolet catastrophe}. Further prime example for this are gases of indistinguishable particles.
An application of classical physics leads to \emph{Gibbs' paradox} for the mixing entropy and the statistics of Bose and Fermi gases at low temperatures cannot be explained classically.
Last but not least, the ``freezing out'' of certain internal degrees of freedom of molecular gases, which impacts their heat capacities, cannot be understood in a convincing way from classical physics alone.

%%%% Bibliography%%%%%%%%%%%%%%%%%%%%%%%%%%%%%%%%%%%%%%%
\emergencystretch 2.5em
\bibliography{bibliographyreview23}
\bibliographystyle{iopart-num}
%\bibliographystyle{plainnat}
%\bibliographystyle{abbrvnat}
%\bibliographystyle{numeric}
%\printbibliography
\emergencystretch 1em

\end{document}